\documentclass[11pt,a4paper]{article}
\usepackage[left=2.5cm,right=2.5cm,bottom=3cm,top=3cm,bindingoffset=0.5cm]{geometry}

\usepackage[utf8]{inputenc}
\usepackage[english]{babel}
\usepackage{amsmath}
\usepackage[final,pdfusetitle, pdftex,]{hyperref}
\hypersetup{%
	pdfauthor={S{\o}ren Fournais and Jan Philip Solovej}
}
\usepackage{amssymb}
\usepackage{amsthm}
\usepackage{fancyhdr}
\usepackage{graphicx}
\usepackage{curves}
\usepackage{subfigure}
\usepackage{boxedminipage}
\usepackage{cite}
\usepackage{lastpage}
\usepackage{textcase}
\usepackage{authblk}
\usepackage{color}
\usepackage{verbatim}
\usepackage{bbm}

\theoremstyle{definition} \newtheorem{definition}{Definition}[section]
\theoremstyle{plain} \newtheorem{theorem}[definition]{Theorem}
\theoremstyle{plain} \newtheorem{assumption}[definition]{Assumption}
\theoremstyle{plain} \newtheorem{proposition}[definition]{Proposition}
\theoremstyle{plain} \newtheorem{lemma}[definition]{Lemma}
\theoremstyle{plain} \newtheorem{corollary}[definition]{Corollary}
\theoremstyle{plain} \newtheorem{remark}[definition]{Remark}
\theoremstyle{definition} 
\theoremstyle{plain} 
\numberwithin{equation}{section}
\fancypagestyle{plain}{%
  \fancyhf{}%
}

\newcommand{\N}{\mathbb{N}}

\newcommand{\R}{\mathbb{R}}
\newcommand{\C}{\mathbb{C}}

\newcommand{\one}{\mathbbm{1}}

\newcommand{\cA}{\mathcal{A}}
\newcommand{\cB}{\mathcal{B}}

\newcommand{\rmu}{\rho_{\mu}}
\newcommand{\KH}{\widetilde{K}_H}

\newcommand{\lille}{\delta}

\newcommand{\abs}[1]{\lvert {#1} \rvert }

\DeclareMathOperator{\Ran}{\mathrm{Ran}}

\DeclareMathOperator{\supp}{\mathrm{supp}}
\DeclareMathOperator{\Spec}{\mathrm{Spec}}

\newcommand{\cM}{{M}} %The constant determining the regularity of the localization function

\newcommand{\Mvalue}{30}

\renewcommand{\epsilon}{\varepsilon}
\renewcommand{\phi}{\varphi}

\usepackage{xifthen}% provides \isempty test
\newcommand{\place}[2]{%
	\ifthenelse{\isempty{#2}}%
	{({\bf p}: ${#1}$):\quad}% if #2 is empty
	{({\bf p}: ${#1}$, {\bf l}: ${#2}$):\quad}% if #2 is empty
}

\begin{document}

\title{The energy of dilute Bose gases}
\author[1]{S\o ren Fournais}
\author[2]{Jan Philip Solovej}
\affil[1]{\small{Department of Mathematics, Aarhus University\\ Ny Munkegade 118\\ DK-8000 Aarhus C\\ Denmark}}
\affil[2]{\small{Department of Mathematics\\ University of Copenhagen\\ Universitetsparken 5\\ DK-2100 Copenhagen \O\\Denmark }}
\maketitle

\begin{abstract}
For a dilute system of non-relativistic bosons interacting through a
positive $L^1$ potential $v$ with scattering length $a$ we prove
that the ground state energy density satisfies the bound $e(\rho) \geq
4\pi a \rho^2 (1+ \frac{128}{15\sqrt{\pi}} \sqrt{\rho a^3} +o(\sqrt{\rho a^3}\,))$, thereby proving the Lee-Huang-Yang formula for
the energy density. 
\end{abstract}

\tableofcontents

\section{Introduction}

Our goal in this paper is to solve the long standing conjecture in
mathematical physics to rigorously establish the Lee-Huang-Yang (LHY)
formula for the second correction to the thermodynamic (infinite
volume) ground state energy per volume of a translation invariant Bose
gas in the dilute limit.  The formula (i.e., \eqref{eq:LHY} below with
an equality) is one of the most fundamental results in quantum
many-body theory. It appeared for the first time as equation (25) in
the seminal 1957 publication \cite{LHY}. The striking feature of the
formula is that the first two terms of the asymptotics of the ground
state energy in the dilute limit depend on the interaction potential
only through a single parameter, {\it the scattering length}. Fairly
recently the LHY formula was tested experimentally as reported in
\cite{Navon2010}. Here the coefficient $\frac{128}{15\sqrt{\pi}}=4.81$
was measured to be $4.4(5)$.

The derivation in \cite{LHY} relies on the pseudo-potential method and
offers deep insight into the problem, but nevertheless lacks in
mathematical rigor.  An alternative, but still non-rigorous, argument
was proposed in \cite{Lsa}.  We establish the LHY formula rigorously
for a large family of two-body potentials (see
Assumption~\ref{assump:v} below), which however does not include the
hard core potential.

The importance of the scattering length in understanding the energy
and excitation spectrum for interacting many-body gases had already
been observed in the celebrated 1947 paper of Bogolubov \cite{BogKiev}
where he introduced the Bogolubov approximation and laid the
foundation for the theory of superfluidity. In this paper Bogolubov
studies the excitation spectrum of a Bose gas and finds that it
depends on the integral of the potential, not the scattering length.
In a famous footnote Bogolubov thanks Landau for making
the important remark that this must be wrong and that the correct
answer must be to replace the integral of the potential by the
scattering length. To establish this rigorously has been a major
challenge ever since. The first major rigorous advance was achieved by
Dyson in \cite{dyson} where the leading order asymptotics for the
ground state energy was established as an upper bound, but where the
lower bound was off by a factor.  The correct leading order
asymptotics was finally established by Lieb and Yngvason in \cite{LY}
for all positive interaction potentials with finite scattering length
including the hard core potential. This result was extended to the
Gross-Pitaevskii limit in the case of trapped gases in
\cite{LSY}. These leading order results are reviewed in the monograph
\cite{LSSY} which also contains a non-rigorous derivation of the LHY
formula using the Bogolubov approximation.  To the best of our
knowledge the first works to rigorously establish the validity of the
Bogolubov approximation for a many-body problem were \cite{LS,LS2,Su}
which studied the one- and two-component charged Bose gases and
established a conjecture of Dyson.  Several ideas from \cite{LS} are
important also in the present work.

The first work to show an upper bound to the LHY order was \cite{ESY}
by Erd\H{o}s, Schlein, and Yau.  This paper makes a very interesting
observation about the Bogolubov approximation.  The usual approach to
the Bogolubov approximation is to approximate the Hamiltonian of the
system by what is referred to as a quadratic Hamiltonian. As mentioned
above this leads to a wrong approximation for the ground state energy
where it will be expressed in terms of the integral of the potential
rather than the scattering length. Quadratic Hamiltonians have ground
states that are quasi-free (or Gaussian) states. In \cite{ESY} it is
observed that if we do not approximate the Hamiltonian by a quadratic
Hamiltonian, but instead restrict the evaluation of the full
Hamiltonian to quasi-free states then miraculously the scattering
length appears in the leading order term, but to LHY order the answer
is still wrong. The work in \cite{ESY} emphasizes that it may often be
more fruitful to focus on classes of states rather than to approximate
the Hamiltonian. This approach was further pursued in the papers
\cite{BogFuncI,BogFuncII} where the positive temperature situation was
analyzed for the Hamiltonian restricted to quasi-free states.  The
leading order correction to the positive temperature free energy for
the full many-body problem in the dilute limit was established in
\cite{St,Yt}.

For gases confined to a region in the Gross-Pitaevskii regime there is
a formula for the second order correction to the ground state energy
similar to the LHY formula. This has recently been established in an
impressive series of papers by Boccato, Brennecke, Cenatiempo, and
Schlein \cite{BBCS0,BBCS,BBCS2}. This however does not imply the
formula in the original thermodynamic infinite volume setting
discussed here.  Our proof follows a very different strategy than the
one applied in the confined case.  

In the confined or trapped case it is also possible to analyze the
excitation spectrum of the gas, which is particularly important for
understanding superfluidity.  The excitation spectrum is also studied
in the papers by Boccato et.\ al. The first result in this direction
is, however, due to Seiringer \cite{Se} and was also analyzed in
\cite{DN,GrS,LNSS,NS}. Getting the excitation spectrum in the
thermodynamic case seems much more difficult.

The LHY formula in the translation invariant thermodynamic setting was
finally rigorously established as an upper bound in the work \cite{YY}
by Yau and Yin, where they consider smooth fastly decaying interaction
potentials.  It is this work that we complement by establishing the
lower bound in \eqref{eq:LHY}, in fact, for a much, larger class of
interaction potentials. Thus the LHY formula has been proved for all
compactly supported
potentials satisfying the assumptions in \cite{YY}. We shall not
discuss the upper bound further in this paper.  In Bogolubov theory,
the particles not in the condensate constitute pairs of opposite
momentum. An important insight, confirmed by the contributions of
\cite{YY} and the present work, is that in order to get the correct
energy to LHY order, one has to go beyond these simple pairs and also
consider `soft pairs’. This means that not only pairs of particles of
exactly opposite momentum contribute. Also pairs of particles with
nonzero total momentum - although the individual momenta are much
larger than the sum - are important for the energy to this precision.
 
The LHY formula had previously been established as a lower bound in
the restricted case where the interaction potential is allowed to
become softer as the gas becomes more dilute. This was first achieved
in \cite{GS}. In this case, however, the potential still has a range
much larger than the inter-particle spacing, which is why the paper
has ``high density'' in the title. Allowing the potential to have
range shorter than the inter-particle spacing, but still required to
be soft, was recently achieved in \cite{BS}. The softness condition
was removed in \cite{BFS}, but only to get the ground state energy to
the correct LHY order, not with the correct asymptotics. Several of
the methods developed in \cite{BS} and \cite{BFS} are crucial to this
work.  

There has been a large literature also on the dynamics of
interacting Bose gases, but we will not review that here.

We now turn to describing the problem in details.  We consider $N$
bosons in $3$ dimensions described by the Hamiltonian
\begin{align}
\label{eq:Hamiltonian}
{\mathcal H}_N= {\mathcal H}_N(v)=\sum_{i=1}^N-\Delta_i+\sum_{\leq i<j\leq N}v(x_i-x_j).
\end{align}

We will allow interactions described by the following assumptions.

\begin{assumption}[Potentials]\label{assump:v}
The potential $v\neq 0$ is non-negative and spherically symmetric,
i.e. $v(x) = v(|x|)\geq 0$, and of class $L^1({\mathbb R}^3)$ with
compact support. We fix $R>0$ such that $\supp v \subset B(0,R)$.
\end{assumption}

We are interested in the thermodynamic limit of the ground state
energy density as a function of the particle density $\rho$.
\begin{align}\label{eq:e_rho}
e(\rho, v)=\lim_{L\to\infty\atop N/L\to \rho}L^{-3}\inf_{\Psi\in
  C^\infty_0([0,L]^N)\setminus\{0\}}\frac{\langle\Psi,{\mathcal
    H}_N(v) \Psi\rangle}{\|\Psi\|^2}.
\end{align}
We will omit the dependence on $v$ from the notation and just write
$e(\rho)$, when the potential is clear from the context.  Here the
inner product $\langle\cdot,\cdot\rangle$ and the corresponding norm
$\|\cdot\|$ are in the Hilbert space $L^2(\Omega^N)$, where we have
denoted $\Omega=[0,L]^3$. If we talk about bosons the infimum above
should be over all symmetric function in $C^\infty_0(\Omega^N)$. It is
however a well-known fact that the infimum over all functions is
actually the same as if constrained to symmetric functions. When we
restrict to functions with compact support in $\Omega$ we are
effectively using Dirichlet boundary conditions, but it is not
difficult to see that the thermodynamic energy is independent of the
boundary condition used.

The main result of this work is to establish the celebrated
Lee-Huang-Yang formula that gives a two-term asymptotic formula for
$e(\rho)$ in the dilute limit. We express the diluteness in terms of
the scattering length $a$ of the potential $v$. The definition of the
scattering length and its basic properties will be given in
Section~\ref{scattering}.

\begin{theorem}[The Lee-Huang-Yang Formula]\label{thm:LHY} If $v$ satisfies
Assumption~\ref{assump:v} then in the limit $\rho a^3\to 0$,
\begin{equation}\label{eq:LHY}
  e(\rho)\geq 4\pi\rho^2 a\left(1 +\frac{128}{15\sqrt{\pi}} \sqrt{\rho a^3}-{\mathcal C}(\rho a^3)^{1/2+\eta}\right).
\end{equation}
where $\eta>0$ and ${\mathcal C}$ depend on ${\mathcal R}=\int v/(8\pi
a)$ and $R/a$ as given explicitly in Theorem~\ref{thm:LHY-Box}
below. We have not attempted to optimize this dependence. It follows
from Theorem~\ref{thm:LHY-Box} that ${\mathcal R}$ and $R/a$ can be
allowed to grow as a negative power of $\rho a^3$.
\end{theorem}

As reviewed above an upper bound consistent with the Lee-Huang-Yang
Formula was given in \cite{YY} under more restrictive assumptions on
the potential (see also \cite{AA}).  Combined with
Theorem~\ref{thm:LHY} the second term of the energy asymptotics of the
dilute Bose gas has therefore been established.  It remains an
interesting open problem to give upper bounds consistent with
\eqref{eq:LHY} under less restrictive assumptions on the potential
than in \cite{YY,AA}. It remains, in particular, an open problem to
obtain upper and lower bounds for the hard core potential.

\paragraph{Acknowledgements.}
SF was partially supported by a Sapere Aude grant from the
Independent Research Fund Denmark, Grant number DFF--4181-00221. 
JPS was partially supported by the Villum Centre of Excellence for the
Mathematics of Quantum Theory (QMATH) and the ERC Advanced grant
321029. Part
of this work was carried out while both authors visited the
Mittag-Leffler Institute.

\section{Facts about the scattering solution}\label{scattering}
In this short section we establish notation and recall results
concerning the scattering length and associated quantities. 

We suppose
that $v$ satisfies Assumption~\ref{assump:v} and refer to Appendix C
of \cite{LSSY} for details and a more general treatment.  The
scattering equation reads
\begin{align}\label{eq:Scattering2}
  (-\Delta + \frac{1}{2} v(x) )(1-\omega(x)) =0,\qquad \text{ with } \omega \rightarrow 0, \text{ as } |x| \rightarrow \infty.
\end{align}
The radial solution $\omega$ to this equation satisfies that there exists a
constant $a >0 $ such that $\omega(x) = a/|x|$ for $x$ outside
$\supp\, v$. This constant $a$ is the {\it scattering length} of the
potential $v$ and we will refer to $\omega$ as the {\it scattering
  solution}. Furthermore, $\omega$ is radially symmetric and
non-increasing with
\begin{align}
	0\leq \omega(x)\leq 1.\label{omegabounds}
\end{align}
We introduce the function
\begin{align}\label{eq:gdef}
g := v(1-\omega).
\end{align}
The scattering equation can be reformulated as
\begin{align}
\label{eq:Scattering3}
-\Delta \omega = \frac{1}{2} g.
\end{align}
From this we deduce, using the divergence theorem, that
\begin{align}
a = (8\pi)^{-1} \int g,
\end{align}
and that the Fourier transform satisfies
\begin{align}\label{es:scatteringFourier}
\widehat{\omega}(k) = \frac{\hat{g}(k)}{2 k^2}.
\end{align}

\section{Grand canonical reformulation of the problem}\label{sec:Fock}

To prove Theorem~\ref{thm:LHY} we will reformulate the problem grand canonically on Fock space. 
Consider, for given $\rmu >0$, the following operator ${\mathcal
  H}_{\rmu}$ on the symmetric Fock space ${\mathcal F}_{\rm
  s}(L^2(\Omega))$. The operator ${\mathcal H}_{\rmu}$ commutes with
particle number and satisfies, with ${\mathcal H}_{\rmu,N}$ denoting
the restriction of ${\mathcal H}_{\rmu}$ to the $N$-particle subspace
of ${\mathcal F}_{\rm s}(L^2(\Omega))$,
\begin{align}\label{eq:BackgroundH}
{\mathcal H}_{\rmu,N} &= {\mathcal H}_N- 8\pi a \rmu N=\sum_{i=1}^N -\Delta_i 
+ \sum_{i<j} v(x_i-x_j)
- 8\pi a \rmu N\\
&=\sum_{i=1}^N \left( -\Delta_i - \rmu \int_{{\mathbb R}^3} g(x_i-y)\,dy \right)
+ \sum_{i<j} v(x_i-x_j).
 \nonumber 
\end{align}
Notice that the new term in ${\mathcal H}_{\rmu,N}$ plays the role of a chemical potential justifying the notation.

Define the corresponding ground state energy density,
\begin{align}
e_0(\rmu):= \lim_{|\Omega| \rightarrow \infty} |\Omega|^{-1}
\inf_{\Psi \in {\mathcal F_{\rm s}}\setminus \{0\}} \frac{\langle
  \Psi, {\mathcal H}_{\rmu} \Psi \rangle}{\| \Psi \|^2}.
\end{align}

\noindent We formulate the following result, which will be a
consequence of Theorems~\ref{thm:CompareBoxEnergy} and
\ref{thm:LHY-Box} below.

\begin{theorem}\label{thm:LHY-Background}
Suppose that $v$ satisfies Assumption~\ref{assump:v}.  Then the
thermodynamic ground state energy density of ${\mathcal H}_{\rmu}$
satisfies for $\rmu a^3\rightarrow 0$ that
\begin{align}
  e_0(\rmu) \geq -4\pi \rmu^2 a \left(1 - \frac{128}{15\sqrt{\pi}} (\rmu a^3)^{1/2}+{\mathcal C}(\rmu a^3)^{1/2+\eta}\right),
\end{align}
where $\eta>0$ and ${\mathcal C}$ depend on ${\mathcal R}=\int v/(8\pi a)$ and $R/a$ as given explicitly in 
Theorem~\ref{thm:LHY-Box}.
\end{theorem}

\begin{proof}[Proof of Theorem~\ref{thm:LHY}]
It is easy to deduce Theorem~\ref{thm:LHY} from
Theorem~\ref{thm:LHY-Background}.  By inserting the ground state of
${\mathcal H}_N$ as a trial state in ${\mathcal H}_{\rmu}$ one gets in
the thermodynamic limit for all $\rho,\rho_\mu >0$
\begin{align}\label{eq:CompareGC}
e(\rho\,) \geq e_0(\rmu) + 8\pi a \rho \rmu \geq 8\pi a \rho \rmu
-4\pi \rmu^2 a \left(1 - \frac{128}{15\sqrt{\pi}} (\rmu a^3)^{1/2} 
+{\mathcal C}(\rmu a^3)^{1/2+\eta}\right) ,
\end{align}
where we have used the lower bound from Theorem~\ref{thm:LHY-Background}.
If we therefore choose $\rmu$ to be equal to $\rho$ we arrive at the LHY formula \eqref{eq:LHY}.
\end{proof}

\section{Strategy of the proof of Theorem~\ref{thm:LHY-Background} and the various parameters}
\label{sec:params}

As already mentioned in the introduction the important parameters 
given in the problem are 
$$
a, \int v, R
$$
All estimates will in the end depend on these. The most important combination is
the diluteness parameter
$$
\rmu a^3.
$$
The proof introduces a series of additional parameters. There is an integer 
$$
M\in\N
$$ 
which determines the regularity of the localization function
defined in Appendix~\ref{sec:chiproperties}. It will be chosen
eplicitly below. The remaining parameters will be chosen to depend on
$\rmu a^3$ and ${\mathcal R}=\int v/(8\pi a)$.  There are
dimensionless parameters
$
0< s,d,\varepsilon_T, 
$
that will be chosen small, and there are dimensionless parameters 
$
1< K_\ell, K_{\mathcal M},\KH, K_B
$ 
that will be chosen large.
The power in the error term will depend
on the choice of these 7 parameters in terms of $\rmu a^3$ and ${\mathcal R}=\int v/(8\pi a)$. 

Let us describe how these parameters enter into the proof and
list all the conditions that they must satisfy. Finally we will
make choices to show that these conditions can all be satisfied. 

The proof will use a double localization approach. First we localize
into boxes of length scale
\begin{equation}\label{eq:def_ell}
  \ell=:K_\ell(\rmu a)^{-1/2}
\end{equation}
I.e., boxes that are long on the scale $(\rmu a)^{-1/2}$ which turn
out to be the relevant length scale for the Bogolubov calculation and which is often referred to as the {\it healing length} in the litterature.
The
length of the box is chosen much longer to get the Bogolubov
calculation correct. The kinetic energy localization will be done in
such a way that constant functions in the box have zero kinetic energy
and such that there is a gap above the zero energy. This gap will
allow us to get a priori control on the number of excited particles,
i.e., those not in the condensate represented by the constant.
However, to get this apriori control we need an a priori lower bound
on the energy, which is correct to an order which is almost as in
LHY. This is achieved by localizing even further to small boxes of
length scale
\begin{equation}\label{eq:delldef}
  d\ell=dK_\ell (\rmu a)^{-1/2}\ll (\rmu a)^{-1/2}
\end{equation}
which gives us our first condition that $dK_\ell\ll 1$. 
Here and below $f\ll g$ is used in the precise meaning that $(f/g)\leq (\rmu
a^3)^{\varepsilon}$ for some positive $\varepsilon$ and likewise for $f \gg g$. 
In these small
boxes we have a much larger energy gap than in the large boxes and
this allows us to absorb errors that we cannot estimate in the larger
boxes. 

The localization of the potential energy is performed by a simple
sliding technique described in Lemma~\ref{lem:LocPotential}.  An
important step in controlling the energy in both the small and large
boxes is to split the potential energy in terms of writing $\one=P+Q$
where $P$ is the projection onto constant functions. The potential
energy can then be written as a sum of 16 terms that contain 0--4 $Q$'s. One of
the main ideas in this paper is to complete an appropriate square
containing the $4Q$ term in Lemma~\ref{lm:potsplit}. This will leave renormalized terms with 0-3
$Q$s, where the potential has essentially been replaced by the function
$g=v(1-\omega)$ from \eqref{eq:gdef}.

The analysis of the small boxes is performed in
Appendix~\ref{SmallBoxes}.  The parameters $\varepsilon_T, d, s$ appear in
the kinetic energy localization formulas of
Section~\ref{sec:lockinpot} and they must satisfy the conditions
\begin{align}
  d^{-5}s^{M+1}\ll&\, 1,\label{con:d5s}\\
  (dK_\ell)^2\ll \varepsilon_TK_\ell^{-2}\ll&\,\varepsilon_T \ll sdK_\ell,\label{con:eTdK}\\
  sK_\ell\gg&\, 1,\label{con:sKell} \\ 
  sdK_\ell \gg&\, K_B^{-1}\label{con:sdKellKB}.
\end{align}
Throughout the paper there will also be logarithmic factors. They are ignored 
here as they are always accomodated by the conditions given.  
Condition \eqref{con:d5s} is needed to prove the kinetic
energy localization into the small boxes (see
\eqref{eq:BSlidingassp}). It relies on a result from \cite{BS}. The
first condition in \eqref{con:eTdK} is needed to have a sufficently
large gap in the small boxes, but in fact, this would only require
$(dK_\ell)^2\ll \varepsilon_T$.  The need for the stronger condition
will be explained below. The condition $dK_\ell\ll 1$ noted above is
contained in \eqref{con:eTdK}.  The last condition in
\eqref{con:eTdK} is required to finally get the correct LHY constant
when the appropriate integral is estimated in Section~\ref{sec:Precise}. The
condition \eqref{con:sKell} is also needed to control the same
integral, in fact, this condition implies that the localized kinetic
energy (see \eqref{eq:T'u}) in the large boxes is essentially the
original kinetic energy at the relevant Bogolubov scales.  Finally,
\eqref{con:sdKellKB} introduces the parameter $K_B$ to control that
the small boxes are not too small.  This is required, in order, to get
a good lower bound on the the energy in the small boxes in
Appendix~\ref{SmallBoxes} (see Theorem~\ref{thm:smallbox}) and hence
for the a priori bound on the energy in the large boxes and
consequently on the number of particles and excited particles in the
large boxes (see Theorem~\ref{thm:aprioribounds}). The parameter $K_B$
has to satisfy the additional conditions that
\begin{align}
  K_B\ll&\,(\rmu a^3)^{-1/6}\label{con:KB},\\
  K_B^3K_\ell^2\ll&\, (\rmu a^3)^{-1/4}\label{con:KBKell}.
\end{align}
Here \eqref{con:KB} is a very weak condition implying that the a
priori lower bound on the energy in Theorem~\ref{thm:aprioriHLambda}
is at least better than the leading order term. 
The condition \eqref{con:KBKell} ensures that the a priori bounds on the particle number
and expected number of excited particles are both correct to leading order (see \eqref{eq:apriorinn+}). 

Having established the a priori bound on the energy and the number of
excited particles we will be able, using the technique of localizing
large matrices from \cite{LS}, to restrict the analysis to the subspace
where the number of excited particles is bounded by a parameter
\begin{equation}\label{eq:DefM}
 {\mathcal M}=:K_{\mathcal M}(\rmu a^3)^{-1/4}.
\end{equation}
It must satisfy
\begin{align}
K_{\mathcal M}^{-2}\int v/a\ll&\, 1,\label{con:KMRKB}\\
K_B^3K_\ell^5\ll&\, {\mathcal M}=K_{\mathcal M}(\rmu a^3)^{-1/4},\label{con:KBellKM}\\
K_{\mathcal M}K_\ell^{-3}\ll&\, (\rmu a^3)^{-1/4}.\label{con:M<<n}
\end{align}
Condition \eqref{con:KMRKB} is needed to control the error in the energy when restricting 
to the situation with a bounded number of excited particles. 
The condition \eqref{con:KBellKM} says that the upper bound $\mathcal M$ on the number of excited particles must be much 
bigger than the {\it expected} number of these particles, which 
in Theorem~\ref{thm:aprioribounds} is shown to be not much worse than $K_B^3K_\ell^2\rmu\ell^3(\rmu a^3)^{1/2}\sim 
K_B^3K_\ell^5$. The condition \eqref{con:M<<n} is a very weak condition that ensures  
\begin{equation}\label{eq:M<<n}
  {\mathcal M}\ll \rmu\ell^3,
\end{equation}
i.e., that the bound on the number of excited particles is much less than the total number of particles. 

When we treat the potential energy a major difficulty will be the
terms with 3 $Q$'s terms. These terms are
responsible for the ``soft'' pairs that we discussed in the
introduction.  The main contributions from these terms come when one
excited particle has low momentum and the other two have high
momenta. This requires introducing an upper cutoff for low momenta,
which we choose to be $K_L(\rmu a)^{1/2}$ and a lower cutoff for high
momenta which we choose to be (see Section~\ref{sec:3Q}) 
\begin{equation}\label{eq:KHtilde-cutoff}
  \KH^{-1}(\rmu a^3)^{5/12}a^{-1}.
\end{equation}
The relevance of the power $5/12$ is technical and will appear in the proof of 
Lemma~~\ref{lem:Q3-splitting2}. For convenience we also introduce the parameter $K_H=\KH(\rmu a^3)^{-5/12}$.

We will not choose $K_L$ as a new parameter, but take
\begin{align}\label{eq:KL_d}
K_L=:(K_\ell d^2)^{-1}\gg K_\ell, 
\end{align}
where the estimate follows from \eqref{con:eTdK}.

We get the additional conditions
\begin{align}
  K_{\mathcal M}K_\ell^4\ll&\, \KH^{3}\label{con:KMKellKH}\\
  (K_\ell K_L)^{1-M}=d^{2M-2}\ll&\,(\rmu a^3)^{1/2}\label{con:KellKLd} \\ 
  K_L\KH=(K_\ell d^2)^{-1}\KH\ll&\, (\rmu a^3)^{-1/12}.\label{con:KLKH}
\end{align}
The condition \eqref{con:KLKH} ensures that the high momenta are
disjoint from the low momenta.  The condition \eqref{con:KellKLd}
will be ensured by choosing the integer $M$ that appears in the
explicit localization function large enough. The condition is needed to
control errors that occur because of the localization function. This error will also appear in the final 
error on the lower bound on the energy (see \eqref{eq:Q3Cancels}) 
The condition \eqref{con:KMKellKH} is needed to control the error
(see\eqref{eq:Q3-1-reduce}) in cutting off the 3Q terms in momentum by
absorbing it into the energy gap. It is here that the powers in the
choice \eqref{eq:KHtilde-cutoff} become important. This step is performed after we have 
introduced second quantization in Section~\ref{sec:Second}

After introducing second quantization it turns out to be useful to do
$c$-number substitution in the spirit of \cite{LSYc}. After $c$-number
substitution, where the annihilation operator for the constant
functions is replaced by a number $z$ we need to control that the
parameter $\rho_z=|z|^2/\ell^3$, which represents the density in the
$c$-number substituted condensate, is sufficiently close to $\rmu$. This is
done in Section~\ref{sec:rough}. It will require the additional conditions
\begin{align}
  K_\ell^8(\rmu \ell^3)^{-1}{\mathcal M}K_L^6K_\ell^6=&\,K_{\mathcal M}K_\ell^5 d^{-12}(\rmu a^3)^{1/4}\ll1,\label{con:rough1}\\
  K_\ell^8(\rmu \ell^3)^{-1}{\mathcal M}^3K_L^6\KH^4(\rmu a^3)^{4/3}
  =&\,K_{\mathcal M}^3\KH^4K_\ell^{-1}d^{-12}(\rmu a^3)^{13/12}\ll1\label{con:rough2}.
\end{align}
These conditions 
ensure that $\delta_1$ defined in Lemma~\ref{lem:Roughbounds} is small enough to satisfy
\eqref{eq:NewConditionNew}. That \eqref{eq:NewConditionNew} is, indeed, satisfied then follows from \eqref{con:sdKellKB}
and \eqref{con:KBKell}.

Finally, we are then left with (see \eqref{eq:SimplifyKz})
\begin{itemize}
\item Terms with no $Q$'s that can be explicitly calculated
\item A quadratic Hamiltonian ${\mathcal K}^{\rm Bog}$ including also some linear terms (corresponding to $1Q$ terms)
\item The $3Q$ terms that are left after the momentum cut-offs and additional quadratic and linear terms. 
\end{itemize}
The quadratic Hamiltonian is  treated using
the simplified Bogolubov method in Appendix~\ref{sec:simplebog}. This
together with the no-$Q$ terms will give the correct energy up to the LHY correction 
and a positive quadratic operator (the diagonalized Bogolubov Hamiltonian), see \eqref{eq:BogHamDiag}. 
This requires, however, the condition
\begin{equation}\label{con:boghamerror}
  K_{\mathcal M}K_\ell^{-3/2}(K_\ell\KH(\rmu a^3)^{1/12})^{M-5}\ll(\rmu a^3)^{1/2}.
\end{equation}
Note that the term taken to the power $M$ here is small by \eqref{con:KLKH} and the estimate in \eqref{eq:KL_d}.

The positive quadratic operator together with
the remaining $3Q$ and other terms not treated by Bogolubov's method can be shown by a very detailed
calculation to be bounded below by a term of lower order than
LHY. This last calculation, done in Subsection~\ref{subsec:Q3}, requires the 
conditions (see Theorem~\ref{thm:Control3Q})
\begin{align}
  K_L^3K_{\mathcal M}=(K_\ell d^2)^{-3}K_{\mathcal M}\ll&\, (\rmu a^3)^{-1/4},\label{con:KellKM}\\
  K_\ell^2\KH^4 d^{-6} \ll&\, (\rmu a^3)^{-1/3},\label{con:KellKH} \\
  K_{\mathcal M}K_\ell^{-3}\KH^{-2}\ll&\,(\rmu a^3)^{-1/12},\label{con:Q3error1}\\
  K_{\mathcal M}K_\ell^{-3}d^{-12}(K_\ell^{-2}\KH^2(\rmu a^3)^{1/6})^{M-1}\ll&\, (\rmu a^3)^{3/4}.\label{con:Q3error3}
\end{align}
The conditions \eqref{con:Q3error1} and \eqref{con:Q3error3} are
needed in order for the errors in Theorem~\ref{thm:Control3Q} to be of
lower order than LHY. There are two additional error terms in
\eqref{eq:Q3Cancels} one is, however, already controlled by condition
\eqref{con:KellKLd} and the last term is small.  The condition
\eqref{con:eTdK} above which was not really needed until now will also
be needed in Subsection~\ref{subsec:Q3}.

If we choose to let all the parameters depend on a small parameter $X\ll1$ in the following way 
\begin{equation}\label{eq:Xparameter}
s=X,\ d=X^6, \ \varepsilon_T=X^{23/4}, \ K_\ell=X^{-3/2},\ K_B=X^{-6},\
K_{\mathcal M}=X^{-1},\ \KH=X^{-8/3},
\end{equation}
then all the conditions \eqref{con:eTdK}--\eqref{con:sdKellKB},
\eqref{con:KMKellKH} will be satisfied. 
In order to satisfy 
\eqref{con:KB}, \eqref{con:KBKell}, \eqref{con:KBellKM}, \eqref{con:M<<n},
\eqref{con:KLKH}--\eqref{con:rough2}, \eqref{con:KellKM}--\eqref{con:Q3error1}
of which the most restrictive is \eqref{con:rough1},
we can choose
\begin{equation}\label{eq:Xchoice}
X=(\rmu a^3)^{1/323}.
\end{equation}
We can choose the integer $M=\Mvalue$ to ensure that \eqref{con:d5s}, \eqref{con:KellKLd},\eqref{con:boghamerror},
and \eqref{con:Q3error3} hold. Finally, 
\eqref{con:KMRKB} holds if 
\begin{equation}\label{eq:intvbound}
\int v/a\ll(\rmu a^3)^{-2/323}.
\end{equation}

To get all the arguments to work we need the assumptions
\begin{equation}\label{eq:AssumptionK_R}
R\leq K_B^{1/2}(\rmu a^3)^{1/4}(\rmu a)^{-1/2},\quad R/\ell \ll (\rmu a^3)^{1/4},\quad
R/a\ll (\rmu a^3)^{-1/4}.
\end{equation}
The third assumption (which could be improved slightly) is the most
restrictive and is used in \eqref{eq:BogHamDiag}. The first assumption is
used in Appendix~\ref{SmallBoxes} and the second assumption says that
the range of the potential should be sufficiently much smaller than
the size of the large boxes.

\section{Localization}\label{sec:Localization}
\subsection{Setup and notation}
The main part of the analysis will be carried out on a box $\Lambda=[-\ell/2,\ell/2]^3$ of size 
$\ell$ given in \eqref{eq:def_ell}
In this section we will carry out
the localization to the box $\Lambda$. The main result is given at the
end of the section as Theorem~\ref{thm:CompareBoxEnergy} which states
that for a lower bound it suffices to consider a `box energy',
i.e. the ground state energy of a Hamiltonian localized to a box of
size $\ell$.  For convenience, in Theorem~\ref{thm:LHY-Box} we state
the bound on the box energy that will suffice in order to prove
Theorem~\ref{thm:LHY-Background}.

It will be important to make an explicit choice of a localization
function $\chi\in C_0^{\cM-1}(\R^3)$, for $\cM\in\N$ with support in
$[-1/2,1/2]^3$. It is given in Appendix~\ref{sec:chiproperties}. The
function will not be smooth but it will be important in the analysis
that we choose $M\in \N$ finite but sufficiently large. The explicit choice 
$M=\Mvalue$ was explained in the previous section. 
We choose
$\chi$ such that
\begin{align}\label{eq:chinormalization}
0 \leq \chi, \qquad \int \chi^2 = 1.
\end{align}
We will also use the notation 
\begin{align}
\chi_{\Lambda}(x) := \chi(x/\ell). 
\end{align}

For given $u \in {\mathbb R}^3$, we define 
\begin{align}
\chi_u(x) = \chi(\frac{x}{\ell}-u)=\chi_\Lambda(x-u\ell).
\end{align}
Notice that $\chi_u$ localizes to the box $\Lambda(u) := \ell u + [-\ell/2,\ell/2]^3$.

We will also need the sharp localization function $\theta_u$ to the box $\Lambda(u)$, i.e.
\begin{align}\label{eq:Theta}
\theta_u := \one_{\Lambda(u)}.
\end{align}

Define $P_u, Q_u$ to be the orthogonal projections in $L^2({\mathbb R}^3)$ defined by
\begin{align}\label{def: projections}
P_u \varphi := \ell^{-3} \langle \theta_u, \varphi\rangle \theta_u, \qquad  Q_u \varphi:= \theta_u \varphi - \ell^{-3} \langle \theta_u, \varphi \rangle \theta_u.
\end{align}
In the case $u=0$ we will use the notations
\begin{equation}
  P_{u=0}=P_\Lambda=P\qquad Q_{u=0}=Q_\Lambda=Q
\end{equation}
Define furthermore 
\begin{align}\label{eq:3.5}
W(x) := \frac{v(x)}{\chi*\chi(x/\ell)}.
\end{align}
That $W$ is well-defined for sufficiently small values of $\rmu$, uses that $v$ has finite range. 
Manifestly $W$ depends on $\ell$ and thus $\rho_\mu$, but we will not reflect this in our notation.

Define the localized potentials
\begin{align}\label{eq:w_u}
w_u(x,y) := \chi_u(x) W(x-y) \chi_u(y), \qquad w(x,y) := w_{u=0}(x,y).
\end{align}
Notice the translation invariance,
\begin{align}\label{eq:transInv}
w_{u+\tau}(x,y) = w_u(x-\ell \tau,y-\ell \tau).
\end{align}
For some estimates it is convenient to invoke the scattering solution
and thus we introduce the notation, which again is well-defined for
$\rho_\mu a^3$ sufficiently small,
\begin{align}\label{eq:defW_12}
W_1(x) &:= W(x) (1-\omega(x)) = \frac{g(x)}{\chi*\chi(x/\ell)},\qquad w_1(x,y):=w(x,y)(1-\omega(x-y)), \nonumber \\
W_2(x) &:= W(x) (1-\omega^2(x)) = \frac{g(x)+g\omega(x)}{\chi*\chi(x/\ell)},\qquad w_2(x,y):=w(x,y)(1-\omega^2(x-y)).
\end{align}
If we add a subscript $u$ we mean as above the translated versions $w_{1,u}(x,y)=w_1(x-\ell u,y-\ell u)$. 
For $\rmu a^3$ sufficiently small a simple change of variables yields, for all $u\in{\mathbb R}^3$, the identities
\begin{align}\label{eq:DefU}
\frac{1}{2} \ell^{-6} \iint_{{\mathbb R}^3\times{\mathbb R}^3} \chi(\frac{x}{\ell})\chi(\frac{y}{\ell}) W_1(x-y)\, dx\,dy
&= \frac{1}{2} \ell^{-6} \iint_{{\mathbb R}^3\times{\mathbb R}^3} w_{1}(x,y)\,dx\,dy \nonumber \\
&=\frac{1}{2} \ell^{-3} \int g = 4 \pi a\ell^{-3}
\end{align}
and likewise
\begin{align}\label{eq:w2int}
\frac{1}{2} \ell^{-6} \iint_{{\mathbb R}^3\times{\mathbb R}^3} w_2(x,y)\,dx\,dy 
&=\frac{1}{2}\ell^{-3} \int g(1+\omega).
\end{align}
The following simple lemma will often be useful.
\begin{lemma}
\begin{align}\label{eq:W1-g}
{g}(x) \leq {W}_1(x)  \leq {g}(x) (1+C \frac{R^2}{\ell^2}) .
\end{align}
\end{lemma}

\begin{proof}
The proof is an easy estimate of the convolution, noting that its maximum is attained at the origin.
\end{proof}

\begin{lemma}\label{lem:Convolution}
Suppose that $f \in L^1({\mathbb R}^3)$ satisfies $\supp f \subset B(0,R)$ and $f(-x)=f(x)$.
Then
\begin{align}
\left| f*\chi_{\Lambda}(x) - \chi_{\Lambda}(x) \int f \right| \leq \max_{i,j}\| \partial_i\partial_j \chi \|_{\infty} 
\left(\frac{R}{\ell}\right)^2 \int |f|.
\end{align}
\end{lemma}
\begin{proof}
The proof is an easy application of a Taylor expansion and the integral representation
$$
f*\chi_{\Lambda}(x) - \chi_{\Lambda}(x) \int f = \int f(y) [\chi_{\Lambda}(x-y) - \chi_{\Lambda}(x)] \,dy.
$$
\end{proof}

\begin{lemma}
Suppose that $R/\ell \leq 1$.
For some universal constant $C>0$ we have
\begin{align}\label{eq:I2-integral}
\left|(2\pi)^{-3}  \int \frac{\widehat{W}_1(k)^2}{2k^2} \,dk -
\widehat{g\omega}(0) \right| \leq  C(R/\ell)^2 \widehat{g\omega}(0).
\end{align}
We also get
\begin{align}\label{eq:I2-integral-2}
\int \frac{(\widehat{W}_1(k) - \widehat{g}(k))^2}{2k^2} \,dk \leq C \frac{R^4}{\ell^4} \widehat{g\omega}(0).
\end{align}
\end{lemma}

\begin{proof}
Recall that $\widehat{\omega}(k)=\frac{\widehat{g}(k)}{2k^2}$ by \eqref{es:scatteringFourier}. 
Using the Fourier transformation and \eqref{eq:W1-g} we get
\begin{align}
\left|(2\pi)^{-3} \int \frac{\widehat{W}_1^2(k)- \widehat{g}^2(k)}{2k^2} \, dk\right|
&= C   \iint \frac{(W_1-g)(x) (W_1+g)(y)}{|x-y|}\,dx\,dy  \nonumber \\
&\leq 3 C \frac{R^2}{\ell^2}  \iint \frac{g(x) g(y)}{|x-y|}\,dx\,dy  \nonumber \\
&=  C' \frac{R^2}{\ell^2} \widehat{g\omega}(0).\label{eq: HLS}
\end{align}
This finishes the proof of \eqref{eq:I2-integral}.
The proof of \eqref{eq:I2-integral-2} follows from a similar calculation and is omitted.
\end{proof}

\subsection{Localization of the kinetic and potential energies}\label{sec:lockinpot}
We will use a sliding localization technique developed in the paper \cite{BS}
where we estimate the kinetic energy $-\Delta$ in $\R^3$ below by an
integral over kinetic energy operators in the boxes $\Lambda(u)$. The
following theorem is essentially Lemma~3.7 in \cite{BS}.
\begin{lemma}[Kinetic energy localizaton]\label{lem:LocKinEn} Let $-\Delta_u^{\mathcal N}$ 
denote the Neumann Laplacian in $\Lambda(u)$. 
If the regularity of $\chi$ has $M\geq 5$ (e.g., for our choice $\Mvalue$) and the positive parameters
$\varepsilon_T,d,s,b$ are smaller than some universal constant then for all $\ell>0$ we have
\begin{equation}
  \int_{\R^3} {\mathcal T}_u  du\leq -\Delta
\end{equation}
where 
\begin{align}
\label{eq:DefT}
      {\mathcal T}_u &:=
      \frac12 \varepsilon_T (d \ell)^{-2} \frac{-\Delta_u^{\mathcal N}}{-\Delta_u^{\mathcal N}+(d\ell)^{-2}} 
      + b \ell^{-2} Q_u + b\varepsilon_T (d\ell)^{-2} Q_u\one_{(d^{-2}\ell^{-1},\infty)}(\sqrt{-\Delta})Q_u+{\mathcal T}'_u
\end{align}
with
\begin{align}\label{eq:T'u}
  {\mathcal T}'_u :=
  Q_u \chi_u \Big\{
  (1-\varepsilon_T) \Big[ \sqrt{-\Delta} - \frac{1}{2}(s\ell)^{-1} \Big]_{+}^2
  + \varepsilon_T \Big[ \sqrt{-\Delta} - \frac{1}{2}(d s\ell)^{-1} \Big]_{+}^2
  \Big\} \chi_u Q_u.
\end{align}
\end{lemma}
\begin{proof} 
In Lemma~3.7 in \cite{BS} we have the same inequality except that the terms above
$$
 \frac12 \varepsilon_T (d \ell)^{-2} \frac{-\Delta_u^{\mathcal N}}{-\Delta_u^{\mathcal N}+(d\ell)^{-2}} 
      + b\varepsilon_T (d\ell)^{-2} Q_u\one_{(d^{-2}\ell^{-1},\infty)}(\sqrt{-\Delta})Q_u.
$$
are replaced by the term  
$\varepsilon_T (d \ell)^{-2} \frac{-\Delta_u^{\mathcal N}}{-\Delta_u^{\mathcal N}+(d\ell)^{-2}}$.

Using scaling it is clear that we may assume $\ell=1$.
The proof in Lemma~3.7 in \cite{BS} relies on the inequality (see (44) in \cite{BS})
$$
d^{-2}\int_{\R^3} \frac{-\Delta_u^{\mathcal N}}{-\Delta_u^{\mathcal N}+d^{-2}} du\leq 
d^{-2}\frac{-\Delta}{-\Delta+d^{-2}}.
$$
The lemma above will follow in the same way if we can also prove that 
\begin{equation}\label{eq:Q1Qslide}
  bd^{-2} \int_{\R^3}Q_u\one_{(d^{-2},\infty)}(\sqrt{-\Delta})Q_u du\leq \frac12 d^{-2}\frac{-\Delta}{-\Delta+d^{-2}}.
\end{equation}
Using Lemma~3.3 in \cite{BS} (with $\chi_u=\theta_u=\one_{\Lambda(u)}$ and ${\mathcal K}(p)=bd^{-2}\one_{(d^{-2},\infty)}$)
we can explicitly calculate the operator on the left in \eqref{eq:Q1Qslide} to be 
$H(-i\nabla)$ where 
\begin{align}
  H(p)=&(2\pi)^{-3}b d^{-2} \int_{|q|>d^{-2}}(\widehat\theta(p)\widehat\theta(q)-\widehat\theta(q-p))^2 dq\nonumber\\
  \leq&(2\pi)^{-3} 2b d^{-2} (\widehat\theta(p)-1)^2
  \int_{|q|> d^{-2}}\widehat\theta(q)^2du+(2\pi)^{-3} 2bd^{-2}\int_{|q|>d^{-2}} (\widehat\theta(q-p)-\widehat\theta(q))^2 dq.
\end{align}
We clearly have $H(0)=0$ and $0\leq H(p)\leq Cb d^{-2}$. It is easy to see that 
$$
\theta(q)^2\leq C\frac1{(1+q_1^2)(1+q_2^2)(1+q_3^2)},
\qquad (\widehat\theta(q-p)-\widehat\theta(q))^2 \leq C\frac{|p|^2}{(1+q_1^2)(1+q_2^2)(1+q_3^2)}
$$ 
It then follows that $H(p)\leq Cb\min\{|p|^2,d^{-2}\}$. Hence
\eqref{eq:Q1Qslide} holds if $b$ is smaller than a universal constant.
\end{proof}
\begin{remark} The kinetic operator in \eqref{eq:DefT} looks complicated. This is partly because 
we need to localize it even further into smaller boxes in order to get
a priori estimates (see Appendix~\ref{SmallBoxes}).  The first term in
\eqref{eq:DefT} will give us a Neumann gap in the small boxes. The
second term in \eqref{eq:DefT} is a Neumann gap in the large
boxes. The third term in \eqref{eq:DefT} will control errors coming
from excited particles with very large monmenta (see Lemma~\ref{lem:Q3-splitting1} and the estimate \eqref{eq:EstimateQ3Tilde1}  in Lemma~\ref{lem:EstimatesOnQ3-js}). Finally the
term ${\mathcal T}'_u$ is the main kinetic energy term in the large
boxes. 
\end{remark}

The localization of the potential energy is much simpler and relies on the identity in the following lemma
which is a straightforward computation similar to Proposition~3.1 in \cite{BS}.
\begin{lemma}[Potential energy localization] \label{lem:LocPotential}
For points $x_1,\ldots,x_N\in\R^3$ we have with the definitions of
$w_{1,u}$ and $w_u$ in \eqref{eq:w_u} and \eqref{eq:defW_12} that
\begin{align}
  -\rmu \sum_{i=1}^N \int & g(x_i-y)\,dy + \sum_{i<j} v(x_i-x_j)  \nonumber \\
  &=
  \int_{\R^3} \Big[-\rmu \sum_{i=1}^N \int w_{1,u}(x_i,y)\,dy + \sum_{i<j} w_u(x_i,x_j)\Big] du,
\end{align}
\end{lemma}
\subsection{The localized Hamiltonian}\label{subsec:LocHam}

The localized Hamiltonian ${\mathcal H}_{\Lambda,u}$ will be an
operator on the symmetric Fock space over $L^2({\mathbb R}^3)$
preserving particle number. Its action on the $N$-particle sector is
as
\begin{align}\label{eq:Def_HB}
({\mathcal H}_{\Lambda,u}(\rmu))_{N} :=
 \sum_{i=1}^N \mathcal{T}_u^{(i)} -
\rmu \sum_{i=1}^N \int w_{1,u}(x_i,y)\,dy + \sum_{i<j} w_u(x_i,x_j),
\end{align}
where the kinetic energy operator was given in \eqref{eq:DefT} above.
We abbreviate
\begin{align}\label{eq:Def_HB2}
\mathcal{T}:= \mathcal{T}_{u=0},\qquad {\mathcal H}_{\Lambda}(\rmu):= {\mathcal H}_{\Lambda,u=0}(\rmu).
\end{align}
We will also write
$$
\chi_{\Lambda} := \chi_{u=0} = \chi(\,\cdot\,/\ell).
$$
Define the ground state energy and energy density in the box, by
\begin{align}
E_{\Lambda}(\rmu) &:= \inf \Spec {\mathcal H}_{\Lambda}(\rmu), \\
e_{\Lambda}(\rmu) &:= \ell^{-3} \inf \Spec {\mathcal H}_{\Lambda}(\rmu) = \ell^{-3} E_{\Lambda}(\rmu).
\end{align}
With these conventions, we find

\begin{theorem}\label{thm:CompareBoxEnergy}
We have
\begin{align}
e_0(\rmu) \geq e_{\Lambda}(\rmu).
\end{align}
\end{theorem}

\begin{proof}
The proof of this statement follows from the fact that 
$({\mathcal H}_{\Lambda,u}(\rmu))_{N}$ and $({\mathcal H}_{\Lambda,u'}(\rmu))_{N}$ are unitarily equivalent by \eqref{eq:transInv}.
Therefore, using Lemma~\ref{lem:LocKinEn} and Lemma~\ref{lem:LocPotential}  we find that
\begin{align}
{\mathcal H}_{\rmu,N}(\rmu) \geq
\int_{\ell^{-1}(\Omega + B(0,\ell/2))} ({\mathcal H}_{\Lambda,u}(\rmu))_{N} \,du  \geq \ell^{-3} | \Omega + B(0,\ell/2)| E_{\Lambda}(\rmu).
\end{align}
Now the desired result follows upon using that $|\Omega+B(0,\ell/2)|/|\Omega| \rightarrow 1$ in the thermodynamic limit.
\end{proof}

It is clear, using Theorem~\ref{thm:CompareBoxEnergy}, that Theorem~\ref{thm:LHY-Background} is a consequence of the following theorem on the box Hamiltonian.
Therefore, the remainder of the paper will be dedicated to the proof of Theorem~\ref{thm:LHY-Box} below.

\begin{theorem}\label{thm:LHY-Box}
Suppose that $v$ satisfies Assumption~\ref{assump:v}, \eqref{eq:intvbound}, and the third assumption in \eqref{eq:AssumptionK_R}.
Then with $K_\ell$, $M$, ${\mathcal R}$, and $X$ as given in Section~\ref{sec:params}
we have in the limit $\rmu a^3 \rightarrow 0$,
\begin{align}\label{eq:EnergyBoxRes1}
  e_\Lambda(\rmu) \geq -4\pi \rmu^2 a +  4\pi \rmu^2 a  \frac{128}{15\sqrt{\pi}}(\rmu a^3)^{\frac{1}{2}}
  -C\rmu^2 a (\rmu a^3)^{1/2} \left( X^2 {\mathcal R}+\frac{R^2}{a^2} (\rmu a^3)^{\frac{1}{2}} +
X^{\frac{1}{5}}\right).
\end{align}
\end{theorem}

\subsection{Potential energy splitting}\label{sec:potsplit}
Using that $P+Q = \one_{\Lambda}$ we 
will in Lemma~\ref{lm:potsplit} below arrive at a very useful decomposition of the potential.

Define the (commuting) operators
\begin{align}
n_0 = \sum_{i=1}^N P_i,\qquad n_{+} = \sum_{i=1}^N Q_i,\qquad n=\sum_{i=1}^N\one_{\Lambda,i}=n_0+n_+ .
\end{align}
We furthermore define
\begin{align}\label{eq:Densities}
\rho_{+} := n_{+} \ell^{-3}, \qquad \rho_0:= n_{0} \ell^{-3}.
\end{align}
A crucial idea in this paper is to write the potetial energy in the
form given in the next lemma, where the important observation is to
identify the positive term ${\mathcal Q}_4^{\rm ren}$ which we will
ignore in our lower bound.
\begin{lemma}[Potential energy decomposition]\label{lm:potsplit}
We have
\begin{equation} \label{eq:potsplit}
-\rmu \sum_{i=1}^N \int w_1(x_i,y)\,dy+
\frac{1}{2} \sum_{i\neq j}  w(x_i, x_j)
= {\mathcal Q}_0^{\rm ren}+{\mathcal Q}_1^{\rm ren}
+{\mathcal Q}_2^{\rm ren}+{\mathcal Q}_3^{\rm ren} + {\mathcal Q}_4^{\rm ren},
\end{equation}
where
\begin{align}
{\mathcal Q}_4^{\rm ren}:=&\,
\frac{1}{2} \sum_{i\neq j} \Big[ Q_i Q_j + (P_i P_j + P_i Q_j + Q_i P_j)\omega(x_i-x_j) \Big] w(x_i,x_j) \nonumber \\
&\,\qquad \qquad \times
\Big[ Q_j Q_i + \omega(x_i-x_j) (P_j P_i + P_j Q_i + Q_j P_i)\Big],\label{eq:DefQ4}\\
{\mathcal Q}_3^{\rm ren}:=&\,
\sum_{i\neq j} P_i Q_j w_1(x_i,x_j) Q_j Q_i + h.c. \label{eq:DefQ3} \\
{\mathcal Q}_2^{\rm ren}:=&\, \sum_{i\neq j} P_i Q_j w_2(x_i,x_j) P_j Q_i
+ \sum_{i\neq j} P_i Q_j w_2(x_i,x_j) Q_j P_i \nonumber \\&\,- \rmu \sum_{i=1}^N Q_i \int w_1(x_i,y)\,dy Q_i 
+ \frac{1}{2}\sum_{i\neq j} (P_i P_j w_1(x_i,x_j) Q_j Q_i + h.c.),\label{eq:DefQ2}\\
{\mathcal Q}_1^{\rm ren}:=&\, \sum_{i,j}P_jQ_iw_2(x_i,x_j)P_iP_j-\rmu
  \sum_{i} Q_i \int w_1(x_i,y)\,dy P_i +h.c.
\label{eq:DefQ1} \\
{\mathcal Q}_0^{\rm ren}:=&\,
\frac{1}{2} \sum_{i\neq j} P_i P_j w_2(x_i,x_j) P_j P_i - \rmu \sum_i P_i \int w_1(x_i,y)\,dy P_i\label{eq:DefQ0}
\end{align}
\end{lemma}
\begin{proof}
The identity \eqref{eq:potsplit} follows using simple algebra and the
identitites \eqref{eq:DefU} and \eqref{eq:w2int}.  We simply write
$P_i+Q_i = 1_{\Lambda,i}$ for all $i$.  Inserting this identity in
both $i$ and $j$ on both sides of $w(x_i,x_j)$ and expanding yields
$16$ terms, which we have organized in a positive ${\mathcal Q}_4$
term and terms depending on the number of $Q$'s occuring.
\end{proof}
It will be useful to rewrite and estimate these terms as in the following lemma. 
\begin{lemma}\label{lem:QDecompSimpler}
If $v$ and hence $W_1$ are non-negative we have
\begin{align}
  {\mathcal Q}_0^{\rm ren}=&\,
  \frac{n_0(n_0-1)}{2|\Lambda|^2} \iint w_2(x,y)\,dx dy - \rmu \frac{n_0}{|\Lambda|}  \iint w_1(x,y)\,dx dy \nonumber\\
  =&\,\frac{n_0(n_0-1)}{2|\Lambda|} \Big(\widehat{g}(0) + \widehat{g\omega}(0)\Big)
  - \rmu n_0  \widehat{g}(0),\label{eq:Q0n0}
\end{align}
\begin{align}  
  {\mathcal Q}_1^{\rm ren}= &\, (n_0|\Lambda|^{-1} -\rmu) \sum_{i} Q_i \chi_{\Lambda}(x_i) W_1*\chi_{\Lambda}(x_i)  P_i + h.c. 
  \nonumber\\ 
  &\,+ n_0|\Lambda|^{-1}  \sum_{i} Q_i \chi_{\Lambda}(x_i)  (W_1\omega)*\chi_{\Lambda}(x_i) P_i + h.c.
  \label{eq:Q1n0}
\end{align}
and
\begin{align}    
      {\mathcal Q}_2^{\rm ren}\geq &\, 
      \sum_{i\neq j} P_i Q_j w_2(x_i,x_j) P_j Q_i
      + \frac{1}{2}\sum_{i\neq j} (P_i P_j w_1(x_i,x_j) Q_j Q_i + h.c.) \nonumber \\
      &\, + \Big((\rho_0 -\rmu) \widehat{W_1}(0) + \rho_0 \widehat{W_1 \omega}(0)\Big)\sum_{i} Q_i \chi_{\Lambda}(x_i)^2 Q_i 
      -C(\rmu+\rho_0)a (R/\ell)^2 n_+\, .\label{eq:Q2n0}
\end{align}
\end{lemma}
\begin{proof} The rewriting of ${\mathcal Q}_0$ is straightforward.
The rewriting of ${\mathcal Q}_1^{\rm ren}$ follows from
\begin{align}
  {\mathcal Q}_1^{\rm ren} = &\,
  \Big((n_0|\Lambda|^{-1} -\rmu) \sum_{i} Q_i \int w_1(x_i,y)\,dy P_i + h.c. \Big) \nonumber\\&\,+
  \Big(n_0|\Lambda|^{-1}  \sum_{i} Q_i \int w_1(x_i,y) \omega(x_i-y) \,dy P_i + h.c. \Big). \nonumber
\end{align}
We carry out the similar calculation on the part of the
$2Q$-term where $P$ acts in the same variable on both sides of the
potential,
\begin{align}
{\mathcal Q}_2^{\rm ren,1} &\,=\sum_{i\neq j} P_i Q_j w_2(x_i,x_j) P_j Q_i
+ \frac{1}{2}\sum_{i\neq j} (P_i P_j w_1(x_i,x_j) Q_j Q_i + h.c.) \nonumber \\
&\,\quad + (\rho_0 -\rmu) \sum_{i} Q_i \chi_{\Lambda}(x_i) W_1*\chi_{\Lambda}(x_i)  Q_i + 
\rho_0  \sum_{i} Q_i \chi_{\Lambda}(x_i)  (W_1\omega)*\chi_{\Lambda}(x_i) Q_i . \nonumber
\end{align}
At this point we invoke Lemma~\ref{lem:Convolution} to get, for example,
\begin{align}
\sum_{i} Q_i \chi_{\Lambda}(x_i) W_1*\chi_{\Lambda}(x_i)  Q_i
\geq &\,
\Bigl(\int W_1\Bigr) \sum_{i} Q_i \chi_{\Lambda}(x_i)^2 Q_i\nonumber\\
&\, - \max_{i,j}\| \partial_i\partial_j \chi \|_{\infty} (R/\ell)^2 \Bigl(\int W_1\Bigr) \| \chi \|_{\infty} n_{+}.
\end{align}
\end{proof}
The decomposition in Lemma~\ref{lm:potsplit} easily implies a simple lower bound on the potential energy.
\begin{lemma}[Simple bound on the potential energy]\label{lm:potapp}
  For all $x_1,\ldots,x_N\in \R^3$ we have if the $2$-body pontential $v\geq0$ the following bound on the potential energy
  \begin{align} 
    -\rmu \sum_{i=1}^N \int w_1(x_i,y)\,dy+
    \frac{1}{2} \sum_{i\neq j}  w(x_i, x_j)
    \geq&\, -C (n^2\ell^{-3}+\rmu^2\ell^{3})a+ \frac12{\mathcal Q}_4^{\rm ren} \label{eq:Q4aprioribound}.
  \end{align}
  Moreover, we also have the bounds 
  \begin{align}
    \pm{\mathcal Q}_1^{\rm ren}\leq C(n^2\ell^{-3}+\rmu^2\ell^3)a\label{eq:d11apriori}\\
    \pm\Bigl(\sum_{i\ne j}Q_jQ_iw_1(x_i,x_j)P_iP_j+h.c.\Bigr)
    \leq Cn^2\ell^{-3} a+\frac14 {\mathcal Q}_4^{\rm ren} \label{eq:d2apriori}\\
    \pm\Bigl(\sum_{i,j}P_jQ'_iw_1(x_i,x_j)Q_iQ_j+h.c.\Bigr)\leq Cn^2\ell^{-3} a+\frac14 {\mathcal Q}_4^{\rm ren}, 
    \label{eq:d13apriori}
  \end{align}
    for any (not necessarily selfadjoint) operator $Q'$ on $L^2(\R^3)$ with $QQ'=Q'$ and $\|Q'\|\leq1$.  
\end{lemma}
\begin{proof}
Since $0\leq \int W_1\leq Ca$ we have
\begin{equation}\label{eq:app1bd}
  0\leq \rmu \sum_{i=1}^N \int w_1(x_i,y)\,dy\leq Ca\|\chi_\Lambda\|_\infty^2 \rmu n\leq Ca\|\chi_\Lambda\|_\infty^2 (\rmu^2\ell^3+ 
  n^2\ell^{-3}).
\end{equation}
The off-diagonal terms in the one-body potential can be estimated
using a Cauchy-Schwarz inequality relying on the positivity of $w_1$
\begin{align}\label{eq:app1bdod}
  \pm\rmu \Bigl(\sum_{i=1}^N P_i\int w_1(x_i,y)\,dyQ_i+
  h.c.\Bigr)\leq &\, \rmu \sum_{i=1}^N P_i\int
  w_1(x_i,y)\,dyP_i\nonumber\\&\,+\rmu \sum_{i=1}^N Q_i\int
  w_1(x_i,y)\,dyQ_i\nonumber\\ \leq &\, Ca(1+\|\chi_\Lambda\|_\infty^2
  )\rmu n.
\end{align}
We also have
$$
0\leq \sum_{i,j}P_iQ_jw_{1}(x_i,x_j)P_iQ_j=n_0|\Lambda|^{-1}\sum_j Q_j\chi_\Lambda(x_j)W_1*\chi_\Lambda(x_j)Q_j
\leq Cn_0 n_+\ell^{-3}a\|\chi_\Lambda\|_\infty^2
$$
or more generally using again Cauchy-Schwarz inequalities we have for all $k=0,1,\ldots$
\begin{align}
  \pm
  \Bigl(\sum_{i,j}P_iQ'_j(w_{1}\omega^k)(x_i,x_j)P_iQ_j+h.c.\Bigr)\leq &\,
   Cn_0\ell^{-3}a\|\chi_\Lambda\|_\infty^2
  \Bigl(\varepsilon n_++\varepsilon^{-1}\sum_iQ_i'Q'^*_i\Bigr),\label{eq:Q'app1}\\ \pm
  \Bigl(\sum_{i,j}P_iQ'_j(w_{1}\omega^k)(x_i,x_j)P_jQ_i+h.c.\Bigr)\leq
  &\, Cn_0\ell^{-3}a\|\chi_\Lambda\|_\infty^2
  \Bigl(\varepsilon n_++\varepsilon ^{-1}\sum_iQ_i'Q'^*_i\Bigr),\label{eq:Q'app2}\\
  \pm\Bigl(\sum_{i,j}P_jQ'_i(w_1\omega^k)(x_i,x_j)P_iP_j+h.c.\Bigr)\leq &\,
  \sum_{i,j}P_jQ'_i(w_1\omega^k)(x_i,x_j)Q'^*_iP_j\nonumber\\
  &\,+\sum_{i,j}P_jP_i(w_1\omega^k)(x_i,x_j)P_iP_j
  \nonumber\\ \leq &\,C
  n_0a\ell^{-3}\Bigl(\|\chi_\Lambda\|_\infty^2\sum_iQ_i'Q'^*_i+n_0\Bigr),\label{eq:Q'app}
\end{align}
for all $\varepsilon>0$, where we have abbreviated $(w_1\omega^k)(x_1,x_2)=w_1(x_1,x_2)\omega(x_1-x_2)^k$.
In this proof we will choose $\varepsilon=1$ and use $\sum_iQ_i'Q'^*_i\leq n_+\leq n$. 
The freedom to choose $\varepsilon\ne1$ will be used in the proof of 
Corollary~\ref{cl:advQ3sch} below. 
The estimates in \eqref{eq:Q'app1}-\eqref{eq:Q'app} prove \eqref{eq:d11apriori} 
if we recall that $w_2=w_1(1+\omega)$ and choose $Q'=Q$.  

To prove \eqref{eq:d13apriori}
we rewrite the terms in ${\mathcal Q}_3^{\rm ren}$ as follows
\begin{align}
  \sum_{i,j}P_iQ'_jw_1(x_i,x_j)Q_jQ_i=&\, \sum_{i,j} \Big(P_i
  Q'_j w_1(x_i,x_j) \Big[ Q_j Q_i + \omega(x_i-x_j) (P_j P_i + P_j Q_i
    + Q_j P_i)\Big]  \Big)\nonumber \\ &\,- \sum_{i,j}
  \Big(P_i Q'_j w_1(x_i,x_j) \omega(x_i-x_j) (P_j P_i +
  P_j Q_i + Q_j P_i)\Big) \label{eq:3Qto4Q}
\end{align}
and likewise for the Hermitian conjugate terms. Thus applying a Cauchy-Schwarz inequality 
and the estimates \eqref{eq:Q'app1}-\eqref{eq:Q'app} we arrive at
\begin{align} 
  \pm \Bigl(\sum_{i,j}P_iQ'_jw_1(x_i,x_j)Q_jQ_i+h.c.\Bigr) \leq &\, \frac{1}{2}
      {\mathcal Q}_4^{\rm ren} +C\sum_{i\neq j} P_i Q'_j
      w_1(x_i,x_j)(1-\omega(x_i-x_j))Q'^*_j P_i \nonumber\\&\,+Cn^2a\ell^{-3}\nonumber
\end{align}
which implies \eqref{eq:d13apriori}. The estimate \eqref{eq:d2apriori} follows in the same way. 
Finally, the estimate \eqref{eq:Q4aprioribound} follows from \eqref{eq:d2apriori}, \eqref{eq:d13apriori}, and 
\eqref{eq:app1bd}-\eqref{eq:Q'app} with $Q'=Q$.
\end{proof}
In our more detailed analysis of the ${\mathcal Q}_3$ terms in Section~\ref{sec:3Q} we will need
the following more refined version of the estimate in \eqref{eq:d13apriori}. 
\begin{corollary}\label{cl:advQ3sch}
With the same notation as in Lemma~\ref{lm:potapp} we have for all $0<\varepsilon<1$ 
\begin{align}
  \sum_{i,j}\Bigl(P_jQ'_iw_1(x_i,x_j)Q_iQ_j+P_jQ'_iw_1(x_i,x_j)\omega(x_i-x_j)P_iP_j\Bigr) + h.c. \geq \nonumber \\
  -Cn_0\ell^{-3} a\Bigl(\varepsilon^{-1}\sum_iQ'_iQ'^*_i+\varepsilon n_+\Bigr)- \frac14 {\mathcal Q}_4^{\rm ren}.\label{eq:pq'pp} 
\end{align}
\end{corollary}
\begin{proof}
We again use the identity \eqref{eq:3Qto4Q} and perform the same Cauchy-Schwarz as above, 
but the term with three $P$ operators now appear on the left and we do not have to estimate it using \eqref{eq:Q'app}. 
We, however, use \eqref{eq:Q'app1} and \eqref{eq:Q'app2} with $0<\varepsilon<1$.
\end{proof}

\section{A priori bounds on particle number and excited particles}

In the section we will give some important a priori bounds on the
particle number $n$, the number of excited particles $n_+$ as well as
on some of the potential energy terms. The bounds on $n$ and $n_+$
essentially say that for states with sufficiently low energy $n$ is
close to what one would expect, i.e., $\rmu\ell^3$ and the expectation
of $n_+$ is smaller with a factor which is not much worse than the
relative LHY error. These bounds are difficult to prove and are given
in \eqref{eq:apriorinn+} below. The proof is in
Appendix~\ref{SmallBoxes}. They rely on a very detailed analysis of a
further localization into smaller boxes.

\begin{theorem}[A priori bounds]\label{thm:aprioribounds}
Assume that the conditions \eqref{con:d5s}, \eqref{con:eTdK}, \eqref{con:sdKellKB}, 
\eqref{con:KB}, and \eqref{eq:AssumptionK_R} on
$K_B$, $R$, $\varepsilon_T$, $s$, and $d$ are satisfied and that $\rmu a^3$ is small enough. 
Then there is a universal constant $C>0$
such that if $\Psi \in {\mathcal
  F}_s(L^2(\Lambda))$ is an $n$-particle normalized state in the
bosonic Fock space over $L^2(\Lambda)$ satisfying
\begin{align}\label{eq:aprioriPsi}
  \langle \Psi, {\mathcal H}_{\Lambda}(\rmu) \Psi \rangle \leq
  - 4\pi \rmu^2 a \ell^3(1-J(\rmu a^3)^{\frac{1}{2}})
\end{align}
for a $0<J\leq K_B^3$ (the freedom to take $J<K_B^3$ will be used Lemma~\ref{lem:LocMatrices}) then 
\begin{equation}\label{eq:apriorinn+}
|n\ell^{-3}-\rmu|\leq \rmu CK_B^{3/2}K_\ell(\rmu a^3)^{1/4},\quad
\text{and}
\quad
\langle\Psi,n_+\Psi\rangle\leq C \rmu\ell^3K_B^3K_\ell^2(\rmu a^3)^{1/2}.
\end{equation}
Moreover, we also have
\begin{equation}
  0\leq \langle \Psi, {\mathcal Q}_4^{\rm ren}\Psi\rangle\leq C\rmu^2 a \ell^3,\label{eq:Q4apriori}
\end{equation}
and 
\begin{align}
  \Bigl|\langle \Psi,\rmu\sum_{i=1}^N(P_i\int
  w_1(x_i,y)dyQ_i+h.c.)\Psi\rangle\Bigr|+&
  \Bigl|\langle \Psi,\sum_{i\ne j}(Q_jP_iw(x_i,x_j)P_iP_j+h.c.)\Psi\rangle\Bigr|\nonumber\\
  +\Bigl|\langle \Psi,\sum_{i,j}(P_jQ_iw(x_i,x_j)Q_iQ_j+h.c.)\Psi\rangle\Bigr|+&
  \Bigl|\langle \Psi,\sum_{i\ne
    j}(Q_jQ_iw(x_i,x_j)P_iP_j+h.c.)\Psi\rangle\Bigr| 
  \nonumber\\ &\leq
  C\rmu^2\ell^{3} \int v.
  \label{eq:dvapriori}
\end{align}
Note that the expressions on the left of
\eqref{eq:dvapriori} above contain $w$
instead of $w_1$ which appeared in
\eqref{eq:d11apriori}--\eqref{eq:d13apriori}. We will need the
estimates with $w$ instead of $w_1$ in the next section and this will
be the only place where estimates containing $\int v$ will appear.
\end{theorem}
\begin{proof}
As explained, the bounds \eqref{eq:apriorinn+} are proved in Theorem~\ref{thm:aprioriHLambda}. 
Due to our assumptions they, in particular, imply that $n\leq C\rmu\ell^3$. 

This a priori bound on $n$, the positivity of the kinetic energy ${\mathcal T}$, and the bound in 
\eqref{eq:Q4aprioribound} immediately imply 
$$
\langle \Psi, {\mathcal H}_{\Lambda}(\rmu) \Psi \rangle \geq -C\rmu^2 a \ell^3 
+\frac12\langle \Psi, {\mathcal Q}_4^{\rm ren}\Psi\rangle
$$
which by the assumption on $\Psi$ gives the bound \eqref{eq:Q4apriori}. 

The bounds on the first two terms in \eqref{eq:dvapriori} follow exactly as the proofs of
\eqref{eq:Q'app1}-\eqref{eq:Q'app} for $k=0$ and with $w_1$ replaced
by $w$ such that $a$ has to be replaced by $\int v\geq 8\pi a$ in the bounds.
The bounds on the last two terms in \eqref{eq:dvapriori} follow the same lines as the proof of 
\eqref{eq:d2apriori} and \eqref{eq:d13apriori}. We sketch it for the last term in \eqref{eq:dvapriori}. We rewrite 
\begin{align}
\sum_{i\ne j}P_iP_jw(x_i,x_j)Q_iQ_j=&\, 
\sum_{i\ne j}P_iP_jw(x_i,x_j)(Q_iQ_j+\omega(x_i-x_j)(P_iP_j+Q_iP_j+P_iQ_j))\nonumber \\
&-\sum_{i\ne j}P_iP_jw(x_i,x_j)\omega(x_i-x_j)(P_iP_j+Q_iP_j+P_iQ_j)
\end{align}
and likewise for the Hermitian conjugate. If we recall that $0\leq \omega\leq $ the last sum is estimated as in the case of 
\eqref{eq:Q'app1}-\eqref{eq:Q'app} again with $a$ replaced by $\int v$ . 
The first term above together with its complex conjugate
is after a Cauchy-Schwarz  controlled by a similar term and $Q_4^{\rm ren}$. I.e., we get
$$
\langle\Psi,\Bigl(\sum_{i\ne j}P_iP_jw(x_i,x_j)Q_iQ_j+h.c.\Bigr)\Psi\rangle\leq C\rmu^2\ell^3\int v 
+C\langle\Psi,Q_4^{\rm ren}\Psi\rangle,
$$ 
which by the bound \eqref{eq:Q4apriori} implies what we want. 
\end{proof}

\section{Localization of the number of excited particles $n_+$}\label{sec:LocMatrices}

As in \cite{BS} we shall use the following theorem from \cite{LS} to restrict the number of excited particles.
\begin{theorem}[Localization of large matrices]~\label{thm:Localizing large matrices} Suppose that 
  ${\mathcal A}$ is an $(N+1)\times (N+1)$ 
  Hermitean matrix and let ${\mathcal A}^{(k)}$, with $k=0,1,\dots ,N$, denote the
  matrix consisting of the $k^{\mathrm{th}}$ supra- and infra-diagonal
  of ${\mathcal A}$. Let $\psi\in \C^{N+1}$ be a normalized vector and set
  $d_k=\langle\psi,{\mathcal A}^{(k)}\psi\rangle$ and $\lambda=\langle\psi,{\mathcal A}
    \psi\rangle=\sum_{ k=0}^{N}d_k$ ($\psi$ need not be an eigenvector of
  ${\mathcal A}$). Choose some positive integer ${\mathcal M}'\leq N+1$. Then, with
  ${\mathcal M}'$ fixed, there is some $n'\in [0,N+1-{\mathcal M}']$ and some normalized
  vector $\phi\in \C^{N+1}$ with the property that $\phi_j=0$ unless
  $n'+1\leq j\leq n'+{\mathcal M}'$ (i.e., $\phi$ has localization length ${\mathcal M}'$) and such that
  \begin{align}
    \langle\phi,{\mathcal A}\phi\rangle\leq \lambda
    +\frac{C}{{\mathcal M}'^2}\sum_{k=1}^{{\mathcal M}'-1}k^2\abs{d_k}+C\sum_{k={\mathcal M}'}^{N}\abs{d_k},\label{eq:Localizing
      of large matrices}
  \end{align}
  where $C>0$ is a universal constant. (Note that the first sum starts at $k=1$.)
\end{theorem}
This will allow us to prove the following result. We emphasize that
this is the only place in this paper where an estimate depends
explicitly on $\int v$ and not just on $a$.
\begin{lemma}[Restriction on $n_+$]\label{lem:LocMatrices}
Let ${\mathcal M}$ be as defined in \eqref{eq:DefM} and satisfying \eqref{con:KMRKB} and \eqref{con:KBellKM}.
Assume, moreover, that $\rmu a^3$ is small enough.
There is then a universal $C>0$ such that
if there is a normalized $n$-particle $\Psi\in{\mathcal F}_s(L^2(\Lambda))$ 
satisfying \eqref{eq:aprioriPsi} under the assumptions in
Theorem~\ref{thm:aprioribounds} with $J=\frac12 K_B^3$ then
there is also a normalized $n$-particle wave function $\widetilde{\Psi}\in
{\mathcal F}_{\rm s}(L^2(\Lambda))$ with the property that
\begin{align}\label{eq:LocalizedNPlus}
  \widetilde{\Psi} = 1_{[0,{\mathcal M}]}(n_{+}) \widetilde{\Psi},
\end{align}
i.e., only values
of $n_+$ smaller than ${\mathcal M}$ appear in $\widetilde{\Psi}$, and
such that
\begin{equation}\label{eq:EnergyBound}
\langle\widetilde\Psi,{\mathcal H}_\Lambda(\rmu)\widetilde\Psi\rangle\leq
\langle\Psi,{\mathcal H}_\Lambda(\rmu)\Psi\rangle
+CK_{\mathcal M}^{-2} \rmu^2\ell^3 (\rmu a^3)^{1/2}\int v.
\end{equation}
\end{lemma}
\begin{proof}
We may assume from \eqref{con:KBellKM} that ${\mathcal M}\geq 5$ and that 
${\mathcal M}\leq n$ since otherwise there is nothing to prove.
 
We shall apply Theorem~\ref{thm:Localizing large matrices} on localization of large matrices to the 
$(n+1)\times(n+1)$-matrix with elements
$$
{\mathcal A}_{i,j}=\|\one_{n_+=i}\Psi\|^{-1}\|\one_{n_+=j}\Psi\|^{-1}
\langle \one_{n_+=i}\Psi, H_\Lambda(\rmu)\one_{n_+=j}\Psi\rangle.
$$
(If any of the norms are zero we set the element to zero.) 
Then we get a normalized vector $\psi=(\|\one_{n_+=0}\Psi\|,\ldots,\|\one_{n_+=n}\Psi\|)$ in $\C^{n+1}$ and 
$$
\langle\psi,{\mathcal A} \psi\rangle= \langle \Psi, H_\Lambda(\rmu)\Psi\rangle. 
$$
Moreover, for the matrix ${\mathcal A}$, using the notation of Theorem~\ref{thm:Localizing large matrices}, 
only the ${\mathcal A}^{(k)}$ with $k=0,1,2$ are non-vanishing. In fact, we have
\begin{align}
  d_1=\langle \psi,{\mathcal A}^{(1)} \psi \rangle=\Bigl\langle\Psi,\Bigl(-\rmu\sum_{i=1}^N(P_i\int
  w_1(x_i,y)dyQ_i+h.c.)+\sum_{i\ne j}(Q_jP_iw(x_i,x_j)P_iP_j+h.c.)\nonumber\\
  +\sum_{i\ne j}(P_iQ_jw(x_i,x_j)Q_jQ_i+h.c.)\Bigr)\Psi\Bigr\rangle\nonumber
\end{align}
and 
$$
d_2=\langle \psi,{\mathcal A}^{(2)} \psi \rangle=
\Bigl\langle\Psi,\Bigl(\sum_{i\ne j}(P_iP_jw(x_i,x_j)Q_jQ_i+h.c.)\Bigr)\Psi\Bigr\rangle.
$$ 
It thus follows from \eqref{eq:dvapriori} that $|d_1|,|d_2|\leq \rmu^2\ell^3\int v$. 

The theorem on localization of large matrices tells us that if we choose ${\mathcal M}'$ equal to the integer part of 
${\mathcal M}/2$ we can find a normalized $\phi\in \C^{n+1}$ with localization length ${\mathcal M}'$
such that 
\begin{align}
  \langle \phi, {\mathcal A}\phi\rangle\leq&\, \langle\psi,{\mathcal A} \psi\rangle+C{\mathcal M}'^{-2}(|d_1|+|d_2|))
  \nonumber\\ 
  \leq&\,\langle \Psi, H_\Lambda(\rmu)\Psi\rangle+CK_{\mathcal M}^{-2}\rmu^2\ell^3 (\rmu a^3)^{1/2} \int v
\end{align}
Let $\widetilde\phi\in\C^{n+1}$ be given by
$\widetilde\phi_i=\phi_i$ if $\|\one_{n_+=i}\Psi\|\ne 0$ and
$\widetilde\phi_i=0$ if $\|\one_{n_+=i}\Psi\|=0$.  Then
$\|\widetilde\phi\|\leq 1$. 
We then have 
\begin{align}
  \langle \widetilde \phi, {\mathcal A}\widetilde\phi\rangle=\langle \phi, {\mathcal A}\phi\rangle\leq
  \langle \Psi, H_\Lambda(\rmu)\Psi\rangle+CK_{\mathcal M}^{-2}\rmu^2\ell^3 (\rmu a^3)^{1/2} \int v< 0.
\end{align}
where the negativity follows from $J=\frac12 K_B^3$, \eqref{con:KB}, and \eqref{con:KMRKB}.
In particular,
$\widetilde\phi\ne0$. Define
$$
\widetilde \Psi=\|\widetilde\phi\|^{-1}\sum_{i=0}^n\widetilde\phi_i\|\one_{n_+=i}\Psi\|^{-1}\one_{n_+=i}\Psi.
$$
Then $\widetilde \Psi$ is normalized and satisfies 
$$
\langle\widetilde\Psi,H_\Lambda(\rmu)\widetilde\Psi\rangle=
\|\widetilde\phi\|^{-2}\langle \widetilde \phi, {\mathcal A}\widetilde\phi\rangle\leq 
\langle \widetilde \phi, {\mathcal A}\widetilde\phi\rangle
$$ since the term on the right is negative and
$\|\widetilde\phi\|^{-2}\geq1$. This proves that $\widetilde\Psi$
satisfies \eqref{eq:EnergyBound}. It remains to prove that
$\widetilde\Psi$ satisfies \eqref{eq:LocalizedNPlus}.  We know from
the construction that the possible values of $n_+$ that occur in
$\widetilde \Psi$ lie in an interval of length ${\mathcal M}'$. We
need to prove that this interval lies close to zero.  This follows
from the estimate \eqref{eq:EnergyBound}, $J=\frac12K_B^3$, and \eqref{con:KMRKB}.
which imply that we may use the a priori
bound \eqref{eq:apriorinn+} on the expectation value of $n_+$ in
$\widetilde \Psi$. The consequence is that the interval of $n_+$ values in
$\widetilde \Psi$ must be contained in
$$
[0, {\mathcal M}'+C\rmu\ell^3 K_B^3K_\ell^2(\rmu a^3)^{1/2}]=[0, {\mathcal M}'+CK_B^3K_\ell^5]\subseteq [0,{\mathcal M}]
$$
by \eqref{con:KBellKM}.

\end{proof}

\section{Localization of the $3Q$-term}
\label{sec:3Q}
In this section we will absorb an unimportant part of the $3Q$ term in the positive $4Q$ term. 

We first define the `low' and `high' momentum regions as follows.
\begin{align}
P_L := \{ |p| \leq K_L \sqrt{\rmu a} \},\qquad P_H := \{ |p| \geq \KH^{-1} (\rmu a^3)^{\frac{5}{12}}a^{-1}\} = \{ |p| \geq K_H^{-1} a^{-1}\}
\end{align}
where $K_L, \KH$ were defined in Section~\ref{sec:params}.
The somewhat peculiar definition of $P_H$ is convenient for later estimates (see proof of Lemma~\ref{lem:Q3-splitting2}).
We will always assume that \eqref{con:KLKH} is satisfied.
This assures that $P_L$ and $P_H$ are disjoint.

We will define the low momentum localization operator $Q_L$ as follows.
Let $f \in C^{\infty}({\mathbb R})$ be a monotone non-increasing function satisfying that $f(s) = 1 $ for $s\leq 1$ and $f(s) = 0$ for $s\geq 2$.
We further define
\begin{align}
f_L(s) := f\big(\frac{s}{K_L \sqrt{\rmu a}}\big).
\end{align}
I.e.  $f_L$ is a smooth localization to the low momenta $P_L$. With this notation, we define
\begin{align}
Q_L := Q f_L(\sqrt{-\Delta}), \qquad \overline{Q_L} := Q (1- f_L(\sqrt{-\Delta})).
\end{align}
Notice that $Q_L$ is not self-adjoint.

We will choose $K_L$ such that $K_L \sqrt{\rmu a} = d^{-2} \ell^{-1}$---this is equivalent to \eqref{eq:KL_d}---where $d$ is from the definition of the `small boxes'.

We define
\begin{align}
n_{+}^{H} := \sum Q \one_{(d^{-2} \ell^{-1}, \infty)}(\sqrt{-\Delta}) Q.
\end{align}
With this definition and the choice of $K_L$ above, we have
\begin{align}\label{eq:QLs_giveNplusH}
\sum \overline{Q}_L (\overline{Q}_L)^{*} \leq n_{+}^{H}.
\end{align}

\begin{lemma}\label{lem:Q3-splitting1}
Define
\begin{align}
\widetilde{Q}_3^{(1)} := 
\sum_{i\neq j} (P_i Q_{L,j} w_1(x_i,x_j) Q_j Q_i + h.c.) ,
\end{align}
We assume \eqref{con:eTdK}, \eqref{con:KellKLd} and \eqref{eq:apriorinn+}.
With the notation from \eqref{eq:DefQ4}, \eqref{eq:DefQ3}, we get,
\begin{align}\label{eq:ReduceQ3}
{\mathcal Q}_3^{\rm ren} + \frac{1}{4} {\mathcal Q}_4^{\rm ren} +\frac{b}{100}  \left( \ell^{-2} n_{+} + \varepsilon_T (d\ell)^{-2} n_{+}^H \right) 
\geq \widetilde{Q}^{(1)}_3 
 - C \rmu^2 a \ell^3   \left( (K_{\ell} K_L)^{1-M} + \frac{R^2}{\ell^2} \right) .
\end{align}
\end{lemma}

\begin{proof}
Using Corollary~\ref{cl:advQ3sch}, with $Q' = \overline{Q}_{L,}$ and $\varepsilon= c K_{\ell}^{-2}$ for some sufficiently small constant $c$, as well as \eqref{eq:QLs_giveNplusH}
we find 
\begin{align}
\frac{1}{4} {\mathcal Q}_4^{\rm ren} + \frac{b}{100 \ell^2} n_{+}+ 
{\mathcal Q}_3^{\rm ren}
- \widetilde{Q}^{(1)}_3 
&\geq
  \sum_{i,j}\Bigl(P_j \overline{Q}_{L,i} w_1(x_i,x_j)\omega(x_i-x_j)P_iP_j+ h.c. \Bigr) \nonumber \\
  &\quad   -C \ell^{-2} K_{\ell}^4 n_{+}^H. \label{eq:pq'pp-sf} 
\end{align}
Using \eqref{con:eTdK} it is clear that the $n_{+}^H$ term is dominated by half of the positive $n_{+}^H$ term from \eqref{eq:ReduceQ3}.

To estimate the remaining terms in \eqref{eq:pq'pp-sf} we start by using the estimate on the convolution from Lemma~\ref{lem:Convolution} to get 
\begin{align}
&- \sum_{i\neq j} \left( P_i \overline{Q}_{L,j}  w_1(x_i,x_j) \omega(x_i-x_j) P_j P_i + h.c. \right) \nonumber \\
&\geq - I \ell^{-3} \Big( n_0 \sum_j \overline{Q}_{L,j} \chi_{\Lambda}^2(x_j) P_j + h.c.\Big)
- C a n^2 \ell^{-3} \frac{R^2}{\ell^2},
\end{align}
where $I := \int W_1(y) \omega(y) \leq C a$.

To complete the proof we write, with $M-1 \leq 2 \widetilde{M} \leq M$ 
\begin{align}
\overline{Q}_{L} \chi_{\Lambda}^2 P + h.c.
=
\overline{Q}_{L} (\ell^{-2} - \Delta)^{-\widetilde{M}} \left[(\ell^{-2} - \Delta)^{\widetilde{M}}\chi_{\Lambda}^2\right]  P + h.c.
\end{align}
and notice that
\begin{align}
|(\ell^{-2} - \Delta)^{\widetilde{M}}\chi_{\Lambda}^2| \leq C \ell^{-2\widetilde{M}}.
\end{align}
Therefore,
\begin{align}
\overline{Q}_{L} \chi_{\Lambda}^2 P + h.c. &\leq
\varepsilon_2^ P + \varepsilon_2^{-1} \ell^{2\widetilde{M}} \overline{Q}_{L} (\ell^{-2} - \Delta)^{-2\widetilde{M}}
(\overline{Q}_{L})^* \nonumber \\
&\leq \varepsilon_2 P + \varepsilon_2^{-1} (K_{\ell} K_{L})^{-2\widetilde{M}} \overline{Q}_{L}(\overline{Q}_{L})^*.
\end{align}
Choosing $\varepsilon_2 = (K_{\ell} K_L)^{-2\widetilde{M}}$ and using again \eqref{con:eTdK} we get \eqref{eq:ReduceQ3}
upon summing this estimate in the particle indices
and absorbing the $n_{+}^H$ term as before.
\end{proof}

\section{Second quantized operators}
\label{sec:Second}
\subsection{Creation/annihilation operators}
We will use $a, a^{\dagger}$ to denote the standard bosonic annihilation/creation operators on the bosonic Fock space ${\mathcal F}_s(L^2(\Lambda))$.

We define $a_0$ as the annihilation operator associated to the condensate function for the box $\Lambda$, i.e. $a_0 = \ell^{-3/2} a(\theta)$, where we recall that $\theta$ defined in \eqref{eq:Theta} is the characteristic function of the box. In more detail, for $\Psi \in \otimes_{s}^N L^2(\Lambda)$ we have
$$
(a_0 \Psi)(x_2, \ldots, x_N) := \frac{\sqrt{N}}{\ell^{3/2}} \int_{\Lambda}  \Psi(y,x_2,\ldots, x_N) \,dy
$$
Therefore,
\begin{align}\label{eq:Appa0}
\langle \Psi ,n_0\Psi\rangle=\langle \Psi\,|\, a_0^{\dagger} a_0 \Psi \rangle 
=
\frac{N}{\ell^3} \int \left| \int_{\Lambda} \Psi(y,x_2,\ldots, x_N) \,dy\right|^2 dx_2\cdots dx_N. 
\end{align}

Due to the localization function $\chi_{\Lambda}$ it is convenient to work with the localized annihilation/creation operators $a_k, a_k^{\dagger}$ defined in \eqref{eq:ak-akdagger} below. However, we will also need the non-localized versions $\widetilde{a}_k, \widetilde{a}_k^{\dagger}$. Since these are more standard, we give their definition first.

For $k \in {\mathbb R}^3\setminus\{0\}$ we let 
\begin{align}
\widetilde{a}_k := \ell^{-3/2} a(Q(e^{ikx} \theta)), \qquad \widetilde{a}_k^{\dagger} := \ell^{-3/2} a^{\dagger}(Q(e^{ikx} \theta))
\end{align}
Clearly, for $k,k' \in {\mathbb R}^3\setminus\{0\}$,
\begin{align}\label{eq:ak-akdagger}
[\widetilde{a}_k, \widetilde{a}_{k'}] = 0, \qquad [\widetilde{a}_k, \widetilde{a}_{k'}^{\dagger}] = \ell^{-3} \langle e^{ikx} \theta, Q e^{ik'x} \theta \rangle.
\end{align}
We also define, for $k \in {\mathbb R}^3\setminus \{0\}$,
\begin{align}
a_k := \ell^{-3/2} a( Q(e^{ikx} \chi_\Lambda))\qquad \textrm{and}\qquad a_k^{\dagger} := \ell^{-3/2} a( Q(e^{ikx} \chi_\Lambda))^{*}.
\end{align}
Then, for all $k,k' \in {\mathbb R}^3\setminus \{ 0 \}$,
\begin{align}
[ a_k, a_{k'} ] = 0, 
\end{align}
and  
\begin{align}\label{eq:Commutator_general}
[a_k, a_{k'}^{\dagger}] &= \ell^{-3}
\langle Q(e^{ikx} \chi_\Lambda), Q (e^{ik'x} \chi_\Lambda) \rangle 
= \widehat{\chi^2}((k-k')\ell) -  \widehat{\chi}(k \ell) \overline{\widehat{\chi}(k'\ell)}.
\end{align}
In particular,
\begin{align}\label{eq:Commutator}
[a_k, a_{k}^{\dagger}] \leq  1.
\end{align}
Furthermore, we introduce the Fourier multiplier corresponding to the localized kinetic energy (after the separation of the constant term), i.e.
\begin{align}\label{def: tau}
\tau(k) := (1-\varepsilon_T) \Big[ |k| - \frac{1}{2}(s\ell)^{-1} \Big]_{+}^2
+ \varepsilon_T \Big[ |k| - \frac{1}{2}(d s\ell)^{-1} \Big]_{+}^2.
\end{align}
We can express the different parts of the Hamiltonian ${\mathcal H}_{\Lambda}(\rmu)$ in second quantized formalism. We give this as the following Lemma~\ref{lem:2ndQuand}. The proof is a standard calculation and will be omitted.

\begin{lemma}\label{lem:2ndQuand}
We have the following expressions for the operators in second quantized formalism (with $\mathcal{T}'$ the part of the kinetic energy operator defined in \eqref{eq:T'u})
\begin{align}
n_0 = a_0^{\dagger} a_0,&\qquad n_0^2 = (a_0^{\dagger} a_0)^2 = (a_0^{\dagger})^2 a_0^2 - a_0^{\dagger} a_0,  \nonumber \\
n_{+} &= (2\pi)^{-3} \ell^3 \int \widetilde{a}_k^{\dagger} \widetilde{a}_k \,dk\nonumber \\
\sum_{j=1}^N  \mathcal{T}'_{j} &= \big( (2\pi)^{-3} \ell^3 \int_{k \in {\mathbb R}^3} \tau(k) a_k^{\dagger} a_k\big)_N,\nonumber \\
\sum_{i \neq j}P_i P_j w_1(x_i,x_j)  Q_j Q_i &=
(2\pi)^{-3}\int \widehat{W}_1(k) a_0^{\dagger} a_0^{\dagger} a_{k}a_{-k} \,dk,\nonumber \\
\sum_{j \neq s} P_j Q_s  w_2(x_i,x_j)  P_s Q_j
&= (2\pi)^{-3} \int \widehat{W}_2(k) a_{-k}^{\dagger} a_0^{\dagger} a_0 a_{-k}  \,dk,\nonumber \\
 \sum_i Q_i f(x_i) \chi_{\Lambda}(x_i) P_i & =
 (2\pi)^{-3} \int \widehat{f}(k) a_{k}^{\dagger}  a_0\,dk, \nonumber \\
 \sum_{i\neq j} P_i Q_{L,j} w_1(x_i,x_j) Q_j Q_i &= \ell^3 (2\pi)^{-6} \iint
f_L(s)
\widehat{W}_1(k) a_0^{\dagger} \widetilde{a}_s^{\dagger}  a_{s-k} a_k
\,dk\,ds.
\end{align}
\end{lemma}

\begin{proposition}\label{prop:Hamilton2ndQuant}
Assume that $\widetilde{\Psi}$ satisfies \eqref{eq:LocalizedNPlus} and \eqref{eq:EnergyBound} and that the parameters satisfy
\eqref{con:KMKellKH}, \eqref{con:eTdK}, and \eqref{con:KellKLd}.
Then, in 2nd quantization the operator ${\mathcal H}_{\Lambda}(\rmu)$ defined in \eqref{eq:Def_HB} satisfies
\begin{align}
\langle \widetilde{\Psi}, {\mathcal H}_{\Lambda}(\rmu) \widetilde{\Psi} \rangle 
\geq \langle \widetilde{\Psi}, {\mathcal H}_{\Lambda}^{\rm 2nd}(\rmu) \widetilde{\Psi} \rangle
- C \rmu^2 a \ell^3   \left( (K_{\ell} K_L)^{1-M} + \frac{R^2}{\ell^2} \right) , 
\end{align}
where
\begin{align}\label{eq:HLambdaSecond}
{\mathcal H}_{\Lambda}^{\rm 2nd} &=
(2\pi)^{-3} \ell^3 \int \tau(k) a_k^{\dagger} a_k\,dk + \frac{b}{2 \ell^2} n_{+} + \varepsilon_T\frac{b}{2 d^2 \ell^2} n_{+}^{H} \nonumber \\
&\quad+ \frac{1}{2} \ell^{-3} a_0^{\dagger} a_0^{\dagger} a_0 a_0 
\Big(\widehat{g}(0) + \widehat{g\omega}(0)\Big)
  - \rmu   \widehat{g}(0) a_0^{\dagger} a_0 \nonumber \\
  &\quad +\Big( (\ell^{-3} a_0^{\dagger} a_0 - \rmu) \widehat{W}_1(0) (2\pi)^{-3} \int \widehat{\chi}_{\Lambda}(k) a_k^{\dagger} a_0 \,dk + h.c. \Big) \nonumber \\
  &\quad +\Big( \ell^{-3} a_0^{\dagger} a_0 \widehat{W \omega}_1(0)(2\pi)^{-3} \int \widehat{\chi}_{\Lambda}(k) a_k^{\dagger} a_0 \,dk + h.c. \Big) \nonumber \\
  &\quad +
  (2\pi)^{-3} \int \left(\widehat{W}_1(k) + \widehat{W_1 \omega}(k)\right) a_0^{\dagger} a_k^{\dagger} a_k a_0 
  + \frac{1}{2} \widehat{W}_1(k) \left( a_0^{\dagger} a_0^{\dagger} a_k a_{-k} + a_k^{\dagger} a_{-k}^{\dagger} a_0 a_0 \right) \,dk \nonumber \\
  & \quad + \Big((\ell^{-3} a_0^{\dagger} a_0 -\rmu) \widehat{W_1}(0) + \ell^{-3} a_0^{\dagger} a_0\widehat{W_1 \omega}(0)\Big) (2\pi)^{-3} \ell^{-3} \int a_k^{\dagger} a_k\,dk  \nonumber\\
  &\quad + \widetilde{Q}_3,
\end{align}
where
\begin{align}\label{eq:3Qtilde}
\widetilde{Q}_3 := \ell^3 (2\pi)^{-6} \iint_{\{ k \in P_H\}}
f_L(s)
\widehat{W}_1(k) (a_0^{\dagger} \widetilde{a}_s^{\dagger}  a_{s-k} a_k
+ a_k^{\dagger}  a^{\dagger}_{s-k} \widetilde{a}_s a_0).
\end{align}
\end{proposition}

\begin{proof}
Notice that \eqref{eq:apriorinn+} holds, using \eqref{eq:EnergyBound} and Theorem~\ref{thm:aprioribounds}.

We apply Lemma~\ref{lm:potsplit}. 
For the operators ${\mathcal Q}_0^{\rm ren}$ and ${\mathcal Q}_1^{\rm ren}$ we use the simplifications of Lemma~\ref{lem:QDecompSimpler} before making the explicit calculation of their 2nd quantifications.
For ${\mathcal Q}_2^{\rm ren}$ we also use the simplifications of Lemma~\ref{lem:QDecompSimpler}. The error term in \eqref{eq:Q2n0} is absorbed in the gap in the kinetic energy. This uses that $R \ll (\rmu a)^{-1/2}$ and the relation $n \approx \rmu \ell^3$ 
from \eqref{eq:apriorinn+}.

Finally we consider ${\mathcal Q}_3^{\rm ren}$ and ${\mathcal Q}_4^{\rm ren}$.
By Lemma~\ref{lem:Q3-splitting1} and the positivity of $v$ we have the lower bound
\eqref{eq:ReduceQ3}. What remains of ${\mathcal Q}_4^{\rm ren}$ will be discarded for a lower bound. The application of \eqref{eq:ReduceQ3} also costs a bit of the gap in the kinetic energy.
We have left to compare $\widetilde{Q}^{(1)}_3$ with $\widetilde{Q}_3$.
But that is the content of Lemma~\ref{lem:Q3-splitting2} below.
Notice that using \eqref{con:KMKellKH} the error term from \eqref{eq:Q3-1-reduce} can be absorbed in the gap in the kinetic energy.
This finishes the proof of Proposition~\ref{prop:Hamilton2ndQuant}.
\end{proof}

In the above proof we used the following localization of the $3Q$-term.

\begin{lemma}\label{lem:Q3-splitting2}
Assume that $\widetilde{\Psi}$ satisfies \eqref{eq:LocalizedNPlus} and \eqref{eq:EnergyBound}.
Let $\widetilde{Q}^{(1)}_3$ be as defined in Lemma~\ref{lem:Q3-splitting1} and
$\widetilde{Q}_3$ from \eqref{eq:3Qtilde}.

Then,
\begin{align}\label{eq:Q3-1-reduce}
\langle \widetilde{\Psi}, \widetilde{Q}^{(1)}_3 \widetilde{\Psi} \rangle
\geq \langle \widetilde{\Psi}, \widetilde{Q}_3 \widetilde{\Psi} \rangle - C a n \frac{n_{+}}{\ell^3} 
\KH^{-3/2} K_{{\mathcal M}}^{1/2} .
\end{align}
\end{lemma}

\begin{proof}
Notice that \eqref{eq:apriorinn+} holds, using \eqref{eq:EnergyBound} and Theorem~\ref{thm:aprioribounds}.

In second quantization we have
\begin{align}
\widetilde{Q}^{(1)}_3
= 
\ell^3 (2\pi)^{-6} \iint
f_L(s)
\widehat{W}_1(k) (a_0^{\dagger} \widetilde{a}_s^{\dagger}  a_{s-k} a_k
+ a_k^{\dagger}  a^{\dagger}_{s-k} \widetilde{a}_s a_0) \,dk\,ds,
\end{align}
so we have to estimate the part of the integral where $k \notin P_H$.
Let $\varepsilon >0$. Then,
\begin{align}\label{eq:EstimQ3-reduce}
\ell^3 (2\pi)^{-6} &\iint_{\{ |k| \leq K_H^{-1} a^{-1} \}}
f_L(s)
\widehat{W}_1(k) (a_0^{\dagger} \widetilde{a}_s^{\dagger}  a_{s-k} a_k
+ a_k^{\dagger}  a^{\dagger}_{s-k} \widetilde{a}_s a_0) \nonumber \\
&\geq
- Ca \ell^3 (2\pi)^{-6} \iint_{\{ |k| \leq K_H^{-1} a^{-1} \}}
f_L(s) \left(
\varepsilon  \widetilde{a}_s^{\dagger} a_0^{\dagger} a_0 \widetilde{a}_s  +
\varepsilon^{-1} a_k^{\dagger} a_{s-k}^{\dagger} a_{s-k} a_k
\right) \nonumber \\
&\geq
- C a n \frac{n_{+}}{\ell^3} \left( \varepsilon \ell^3 (K_H a)^{-3} + \varepsilon^{-1} \frac{{\mathcal M}}{n}\right).
\end{align}
Notice that we have not assumed that $\widehat{W}_1(k)$ has a sign and that the Cauchy-Schwarz inequality in \eqref{eq:EstimQ3-reduce} is valid for $\widehat{W}_1(k)$ of variable sign.

We choose $\varepsilon = \left( \frac{{\mathcal M} K_H^3 a^3}{n \ell^3} \right)^{1/2}$.
Using the relation $n \approx \rmu \ell^3$ 
from \eqref{eq:apriorinn+}
the error term in $(\cdot)$ becomes of magnitude
$
\sqrt{\frac{{\mathcal M}}{\rmu a^3 K_H^3}} = \KH^{-3/2} K_{{\mathcal M}}^{1/2} 
$.
\end{proof}

It will also be useful to notice the following representation in terms of the operators $\widetilde{a}_k$.
\begin{lemma}
We have the identities
\begin{align}
\left((2\pi)^{-6} \ell^6 \iint \widetilde{a}_k^{\dagger} \widehat{\chi^2}\left((k-k')\ell\right) \widetilde{a}_{k'}
\right)_{N} = \sum_{j=1}^N Q_j \chi_{\Lambda}^2(x_j) Q_j,
\end{align}
and
\begin{align}\label{eq:RepInTermsOfatilde}
\left((2\pi)^{-6} \ell^6 \iint  f_L(k)f_L(k')\widetilde{a}_k^{\dagger} \widehat{\chi^2}\left((k-k')\ell\right) \widetilde{a}_{k'}
\right)_{N} = \sum_{j=1}^N Q_{L,j} \chi_{\Lambda}^2(x_j) Q_{L,j}.
\end{align}

\end{lemma}

\subsection{$c$-number substitution}
It is convenient to apply the technique of $c$-number substitution as described in \cite{LSYc}.

Let $\Psi \in {\mathcal F}(L^2(\Lambda))$.
We can think of $L^2(\Lambda) = \Ran(P) \oplus \Ran(Q)$, with $\Ran(P)$ being, of course, spanned by the constant vector $\theta$ (defined in \eqref{eq:Theta}).
This leads to the splitting ${\mathcal F}(L^2(\Lambda)) = {\mathcal F}(\Ran(P)) \otimes {\mathcal F}(\Ran(Q))$. We let $\Omega$ denote the vacuum vector in ${\mathcal F}(L^2(\Lambda))$.

For $z \in {\mathbb C}$ we define 
\begin{align}
|z \rangle:= \exp(-\frac{|z|^2}{2} - z a_0^{\dagger}) \Omega.
\end{align}
Given $z$ and $\Psi$ we can define 
\begin{align}
\Phi(z) := \langle z | \Psi \rangle \in {\mathcal F}(\Ran(Q)),
\end{align}
where the inner product is 
considered as a partial inner product induced by the representation ${\mathcal F}(L^2(\Lambda)) = {\mathcal F}(\Ran(P)) \otimes {\mathcal F}(\Ran(Q))$.

It is a simple calculation that
\begin{align}
1 = \int_{{\mathbb C}} |z \rangle \langle z | \,d^2z, \qquad 
\text{ and } \qquad 
a_0  |z \rangle = z  |z \rangle.
\end{align}

\begin{theorem}\label{thm:Kafz}
Define
\begin{align}
\rho_z := |z|^2 \ell^{-3},
\end{align}
and
\begin{align}\label{eq:Kz}
{\mathcal K}(z) 
&=
 \frac{b}{2 \ell^2} n_{+} + \varepsilon_T\frac{b}{2 d^2 \ell^2} n_{+}^{H} 
 +  
 \frac{1}{2} \rho_z^2 \ell^{3} 
\Big(\widehat{g}(0) + \widehat{g\omega}(0)\Big)
  - \rmu   \widehat{g}(0) \rho_z \ell^3 \nonumber \\
  &\quad +
  (2\pi)^{-3} \ell^3 \int \left( \tau(k) + \rho_z \widehat{W}_1(k) \right) a_k^{\dagger} a_k
  + \frac{1}{2} \rho_z \widehat{W}_1(k) \left(  a_k a_{-k} + a_k^{\dagger} a_{-k}^{\dagger}  \right) \,dk \nonumber \\
  & \quad + (\rho_z -\rmu) \widehat{W_1}(0)  (2\pi)^{-3} \ell^{3} \int a_k^{\dagger} a_k\,dk  \nonumber\\
  &\quad + {\mathcal Q}_1(z)+  {\mathcal Q}_1^{\rm ex}(z)+
  {\mathcal Q}_2^{\rm ex}(z) + {\mathcal Q}_3(z),
\end{align}
with
\begin{align}\label{eq:DefQ1z}
 {\mathcal Q}_1(z)&:= 
  \Big( (\rho_z - \rmu) \widehat{W}_1(0) (2\pi)^{-3} \int \widehat{\chi}_{\Lambda}(k) a_k^{\dagger} z \,dk + h.c. \Big),  \\
  \label{eq:DefQ1zex}
 {\mathcal Q}_1^{\rm ex}(z) &:= \Big( \rho_z \widehat{W_1 \omega}(0)(2\pi)^{-3} \int \widehat{\chi}_{\Lambda}(k) a_k^{\dagger} z \,dk + h.c. \Big),\\
  \label{eq:kroelleQ3z}
 {\mathcal Q}_3(z)&:= \ell^3 (2\pi)^{-6} z \iint_{\{ k \in P_H\}}
f_L(s)
\widehat{W}_1(k) ( \widetilde{a}_s^{\dagger}  a_{s-k} a_k
+ a_k^{\dagger}  a^{\dagger}_{s-k} \widetilde{a}_s),
\end{align}
and
\begin{align}
{\mathcal Q}_2^{\rm ex} = {\mathcal Q}_2^{\rm ex}(z) :=
(2\pi)^{-3} \rho_z \ell^3  \int \big(\widehat{W_1 \omega}(k) + \widehat{W_1 \omega}(0)\big)  a_k^{\dagger} a_k .
\end{align}

Assume that $\widetilde{\Psi}$ satisfies \eqref{eq:LocalizedNPlus}.

Then,
\begin{align}\label{eq:IntroK(z)}
\langle \widetilde{\Psi}, {\mathcal H}_{\Lambda}^{\rm 2nd}(\rmu) \widetilde{\Psi} \rangle
\geq
\inf_{z \in {\mathbb R}_{+} } \inf_{\Phi} \langle \Phi, {\mathcal K(z)} \Phi \rangle
- C \rmu a,
\end{align}
where the second infimum is over all normalized $\Phi  \in {\mathcal F}(\Ran(Q))$ with 
\begin{align}\label{eq:LocalizedAftercNumber}
\Phi =  1_{[0,{\mathcal M}]}(n_{+}) \Phi.
\end{align}
\end{theorem}

\begin{proof}
Notice that \eqref{eq:apriorinn+} holds, using \eqref{eq:EnergyBound} and Theorem~\ref{thm:aprioribounds}.

We define $\widetilde{\mathcal K}(z)$ to be the operator ${\mathcal H}_{\Lambda}^{\rm 2nd}$ defined in \eqref{eq:HLambdaSecond} above, but where the following substitutions have been performed:
\begin{align}\label{eq:subst-z}
&a_0^{\dagger} a_0^{\dagger} a_0 a_0 \mapsto |z|^4 - 4 |z|^2 + 2, \nonumber \\
&a_0^{\dagger} a_0 a_0   \mapsto |z|^2 z - 2 z, \qquad
a_0 a_0^{\dagger} a_0^{\dagger}    \mapsto |z|^2 \overline{z} , \nonumber \\
&a_0 a_0 \mapsto z^2,\qquad a_0^{\dagger} a_0^{\dagger} \mapsto \overline{z}^2,\qquad
a_0^{\dagger} a_0 \mapsto |z|^2 - 1, \nonumber \\
&a_0 \mapsto z,\qquad a_0^{\dagger} \mapsto \overline{z}.
\end{align}
Then, we will prove that
\begin{align}\label{eq:Csubst}
\langle \widetilde{\Psi}, {\mathcal H}_{\Lambda}^{\rm 2nd} \widetilde{\Psi} \rangle 
&= \Re \int \langle \Phi(z), \widetilde{\mathcal K}(z)   \Phi(z) \rangle \,d^2z 
=  \Re \int \langle \widetilde{\Phi}(z), \widetilde{\mathcal K}(z)   \widetilde{\Phi}(z) \rangle n^2(z) \,d^2z,
\end{align}
where $n(z) = \| \Phi(z) \|_{{\mathcal F}(\Ran(Q))}$ and $ \widetilde{\Phi}(z)  =  \Phi(z)/n(z)$.

To obtain \eqref{eq:Csubst}  we write all polynomials in $a_0, a_0^{\dagger}$ in anti-Wick ordering, for example $a_0^{\dagger} a_0 = a_0 a_0^{\dagger} - 1$. Therefore,
\begin{align}
\langle \Psi, a_0^{\dagger} a_0 \Psi \rangle = \int \langle a_0^{\dagger} \Psi |z \rangle \langle z | a_{0}^{\dagger} \Psi \rangle - \langle \Psi| z \rangle \langle z | \Psi \rangle \,d^2 z = 
\int (|z|^2-1) \langle \Phi(z) | \Phi(z) \rangle \,d^2 z .
\end{align}
Performing this type of calculation for each term in $ {\mathcal H}_{\Lambda}^{\rm 2nd}$ yields \eqref{eq:Csubst}.

Suppose that $\widetilde{\Psi} \in {\mathcal F}_{\rm s}(L^2(B))$ is such that 
\begin{align}
\widetilde{\Psi} = 1_{[0,{\mathcal M}]}(n_{+}) \widetilde{\Psi}.
\end{align}
Notice that this relation only involves the part of $ \widetilde{\Psi} \in {\mathcal F}(\Ran(Q))$.
Therefore, we also have for all $z \in {\mathbb C}$,
\begin{align}\label{eq:ControlNPlusRemains}
\widetilde{\Phi}(z)  = 1_{[0,{\mathcal M}]}(n_{+}) \widetilde{\Phi}(z).
\end{align}
with $\widetilde{\Phi}(z) = \langle z | \widetilde{\Psi} \rangle$ as above.

The next step of the proof is to remove the lower order terms coming from the substitutions in \eqref{eq:subst-z} above.

Let us first consider the negative term $-4|z|^2$ in the substitution of $a_0^{\dagger} a_0^{\dagger} a_0 a_0$.
By undoing the integrations leading to $\widetilde{\mathcal K}(z)$ for this term, we see that it contributes with
\begin{align}\label{eq:ErrorInZ4} 
\int \langle \Phi(z), 
-4  \frac{1}{2} |z|^2 \ell^{-3} 
\Big(\widehat{g}(0) + \widehat{g\omega}(0)\Big)\Phi(z) \rangle\,d^2z 
&\geq
-C a \ell^{-3} \langle \widetilde{\Psi}, a_0 a_0^{\dagger} \widetilde{\Psi} \rangle  \nonumber \\
&\geq -C a \ell^{-3} (n+1),
\end{align}
in agreement with the error term in \eqref{eq:IntroK(z)} (using that $n \approx \rmu \ell^3 \gg 1$).

We also estimate the term linear in $z$ coming from the substitution of $a_0^{\dagger} a_0 a_0$ in \eqref{eq:subst-z}. This substitution occurs twice, but we will only explicitly treat one of them, namely the term
\begin{align}
& \Re \int \Big\langle \Phi(z), -2 \ell^{-3} \widehat{W}_1(0) (2\pi)^{-3} \int \widehat{\chi}_{\Lambda}(k) a_k^{\dagger} z \,dk \Phi(z) \Big\rangle d^2 z \nonumber \\
&=
-2 \ell^{-3} \widehat{W}_1(0) (2\pi)^{-3}
\int \Big\langle \Phi(z),  \int \widehat{\chi}_{\Lambda}(k) (a_k^{\dagger} z + a_k \overline{z}) \,dk \Phi(z) \Big\rangle d^2 z \nonumber \\
&\geq -C a \ell^{-3} 
\int \Big\langle \Phi(z),  \int |\widehat{\chi}_{\Lambda}(k)| ( \varepsilon a_k^{\dagger} a_k + \varepsilon^{-1} |z|^2) \,dk \Phi(z) \Big\rangle d^2 z,
\end{align}
where $\varepsilon >0$ will be chosen in the end.
Notice that $|\widehat{\chi}_{\Lambda}(k)| = \ell^{3} |\widehat{\chi}(k\ell)|$ and that $\widehat{\chi} \in L^1({\mathbb R}^3)$ for $M \geq 4$.
Redoing the calculation in \eqref{eq:ErrorInZ4} we therefore find with $\varepsilon = \sqrt{ \langle \widetilde{\Psi}, n_{+} \widetilde{\Psi}\rangle }/\sqrt{n+1}$,
\begin{align}
\Re \int \Big\langle \Phi(z), -2 \ell^{-3} \widehat{W}_1(0) (2\pi)^{-3} \int \widehat{\chi}_{\Lambda}(k) a_k^{\dagger} z \,dk \Phi(z) \Big\rangle d^2 z
\geq -C a \ell^{-3} \sqrt{n+1} \sqrt{ \langle \widetilde{\Psi}, n_{+} \widetilde{\Psi}\rangle }.
\end{align}
This is also easily absorbed in the error term in \eqref{eq:IntroK(z)}.

The other error terms from the substitutions are \eqref{eq:subst-z} estimated in a similar manner and we will leave out the details.

Finally, we need to restrict to non-negative $z$.
Suppose $z = |z| e^{i\phi}$.
In the operator we can replace $a_{\pm k}$ by $e^{i\phi} a_{\pm k}$. In this way all occurences of $z$ will be replaced by $|z|$. 
Notice that this substitution will not affect the commutation relations. 
This finishes the proof.
\end{proof}

\section{First energy bounds}
\label{sec:rough}
In this section we will make a rough estimate on the energy. 
This rough estimate will be used to eliminate the values of $\rho_z$ that are far away from $\rmu$.

\begin{lemma}\label{lem:Roughbounds}
For any state $\Phi$ satisfying \eqref{eq:LocalizedAftercNumber} and assuming that ${\mathcal M} \leq C^{-1} \rmu \ell^3$ for some sufficiently large constant $C$, we have the bound

\begin{align}
\label{eq:RoughBelow}
\langle \Phi, {\mathcal K}(z) \Phi \rangle &\geq - \frac{\widehat{g}(0)}{2} \rmu^2 \ell^3 + \frac{\widehat{g}(0)}{2} (\rmu - \rho_z)^2 \ell^3 - a (\rho_z + \rmu)^{3/2} \rmu^{1/2} \ell^3 \delta_1 - \rho_z^2 a \ell^3 \delta_2
\nonumber  \\
&\quad 
 - C \rmu^2 a \ell^3 \frac{\rmu a^3}{K_{\ell}^6 (d s)^6} .
\end{align}
with 
\begin{align}
\delta_1 &:= C 
 \sqrt{\frac{{\mathcal M}}{\rmu \ell^3}}\Big(K_L^3 \KH^2 (\rmu a^3)^{2/3} {\mathcal M}
+
 K_L^3 K_{\ell}^3
\Big) , \nonumber \\
\delta_2 &:= C  \left(\frac{R^2}{\ell^2} + \frac{a}{d s \ell} \big(1 + \log\big( \frac{d s \ell}{a}\big)\big)
 \right).
\end{align}
\end{lemma}

Before we give the proof of Lemma~\ref{lem:Roughbounds} we wil state its main consequence, Proposition~\ref{prop:FarAway} below.

Notice that by Section~\ref{sec:params} our choice of parameters ensures that $\delta_1 + \delta_2 \ll 1$.

\begin{proposition}\label{prop:FarAway}
Suppose that $\delta_1+\delta_2\leq \frac{1}{2}$.
Suppose furthermore, for some sufficiently large universal constant $C>0$, we have
\begin{align}\label{eq:Close-I}
|\rho_z - \rmu | \geq C \rmu \max\left( \Big(\delta_1 + \delta_2+ \frac{\rmu a^3}{K_{\ell}^6 (d s)^6}\big)^{\frac{1}{2}}, (\rmu a^3)^{\frac{1}{4}} \right).
\end{align}
Then, for any state $\Phi$ satisfying \eqref{eq:LocalizedAftercNumber}, we have
\begin{align}
\langle \Phi, {\mathcal K}(z) \Phi \rangle
\geq - \frac{\widehat{g}(0)}{2} \rmu^2 \ell^3+ 2 \rmu^2 a \ell^3 \frac{128}{15 \sqrt{\pi}} \sqrt{\rmu a^3}.
\end{align}
\end{proposition}

\begin{proof}
Using the convexity of $t \mapsto t^{\sigma}$, for $\sigma \in \{3/2, 2\}$ and Jensen's inequality, \eqref{eq:RoughBelow} implies the bound
\begin{align}
\langle \Phi, {\mathcal K}(z) \Phi \rangle &\geq - \frac{\widehat{g}(0)}{2} \rmu^2 \ell^3 + \frac{\widehat{g}(0)}{2} (1 - \delta_1 -\delta_2)  (\rmu - \rho_z)^2 \ell^3 
 - C \rmu^2 a \ell^3 \Big(\delta_1 + \delta_2 + \frac{\rmu a^3}{K_{\ell}^6 (d s)^6}\Big) \nonumber \\
 &\geq
 - \frac{\widehat{g}(0)}{2} \rmu^2 \ell^3 + \frac{\widehat{g}(0)}{4}   (\rmu - \rho_z)^2 \ell^3 
 - C \rmu^2 a \ell^3 \Big(\delta_1 + \delta_2 + \frac{\rmu a^3}{K_{\ell}^6 (d s)^6}\Big).
\end{align}
If \eqref{eq:Close-I} is satisfied, then the term quadratic in $(\rmu - \rho_z)$ dominates both the error term above and the LHY correction. This finishes the proof of Proposition~\ref{prop:FarAway}.
\end{proof}

\begin{proof}[Proof of Lemma~\ref{lem:Roughbounds}]
Since, for any $\delta'>0$,
$z a_k^{\dagger} + \overline{z} a_k \leq \delta' |z|^2 + (\delta')^{-1} a_k^{\dagger} a_k$,
we find
\begin{align}
 \int \widehat{\chi_{\Lambda}}(k) (z a_k^{\dagger} + \overline{z} a_k)\,dk
 &\leq \delta'  |z|^2 \int | \widehat{\chi_{\Lambda}}(k) |\,dk + |\widehat{\chi_{\Lambda}}(0) |
(\delta')^{-1}  \int a_k^{\dagger} a_k\,dk\nonumber \\
 &\leq C ( \delta' |z|^2 + (\delta')^{-1} n_{+}).
\end{align}
Therefore, we easily get, setting $\delta' = \sqrt{ {\mathcal M}/(\rho_z \ell^3)}$ and using \eqref{eq:LocalizedAftercNumber} and the definitions in \eqref{eq:DefQ1z} and \eqref{eq:DefQ1zex},
\begin{align}
\langle \Phi, \big({\mathcal Q}_1(z) + {\mathcal Q}_1^{\rm ex}(z) \big) \Phi \rangle \geq
-C \ell^3 a \sqrt{\frac{{\mathcal M}|z|^2}{\ell^3}} ( |\rho_z-\rmu| + \rho_z) ,
\end{align}
in agreement with \eqref{eq:RoughBelow} (where we used that $K_L, K_{\ell} \geq1$)

Notice that quadratic terms of the form
$\ell^{3} \int \widehat{W}(k) a_k^{\dagger} a_k \,dk$ are easily estimated as
\begin{align}
\pm \langle \Phi, \ell^{3} \int \widehat{W}(k) a_k^{\dagger} a_k \,dk \Phi \rangle
\leq C a {\mathcal M}.
\end{align}
This allows us to estimate all the quadratic terms in ${\mathcal K}(z)$ except the kinetic energy and the off-diagonal quadratic terms and to absorb the corresponding terms in the error in \eqref{eq:RoughBelow} (using in particular that ${\mathcal M} \leq \rmu \ell^3$).

Therefore, to establish \eqref{eq:RoughBelow} we only have left to estimate the sum of the kinetic energy, ${\mathcal Q}_3(z)$ and the `off-diagonal' quadratic terms.
This we will do by first adding and subtracting an $n_{+}$ term, which is easily estimated as above.
We will prove the following 3 inequalities, where $\varepsilon <1/2 $ is a (small) parameter that we will optimize in the end(see \eqref{eq:ChoiceEpsilon}), and where $\Phi$ is a state satisfying \eqref{eq:LocalizedAftercNumber},
\begin{align}\label{eq:Rough0}
-  \left\langle \Phi, (2\pi)^{-3} \ell^3  \rho_z \varepsilon^{-1/2} a \int  a_k^{\dagger} a_k \,dk \,\Phi \right\rangle
 \geq - C \varepsilon^{-1/2} \ell^3 \rho_z a \frac{{\mathcal M}}{\ell^3},
\end{align}
\begin{align}\label{eq:Rough1}
\left\langle \Phi,  \left((2\pi)^{-3} \ell^3 \int \varepsilon \tau (k) a_k^{\dagger} a_k \,dk + {\mathcal Q}_3(z)\right) \, \Phi \right\rangle 
 \geq -\varepsilon^{-1} C \rho_z a \frac{{\mathcal M}}{\ell^3} \ell^3 \Big(K_L^3 K_H^2 K_{\ell}^3  \frac{{\mathcal M}a^3}{\ell^3}
+
 K_L^3 K_{\ell}^3
\Big),
\end{align}
and
\begin{align}\label{eq:Rough2}
(2\pi)^{-3} \ell^3 &\int \Big({\mathcal A}_1(k) a_k^{\dagger} a_k
 +\frac{1}{2}{\mathcal B}_1(k) \left( a_k^{\dagger} a_{-k}^{\dagger} + a_k a_{-k}\right) \Big) dk  \\
 &\geq
- \frac{1}{2} \rho_z^2 \ell^3 \widehat{g\omega}(0)  - 
 C \ell^3 \rho_z^2 a \left(\varepsilon + \frac{R^2}{\ell^2} + \frac{a}{d s \ell} \big(1 + \log\big( \frac{d s \ell}{a}\big)\big)
 \right)
 - C \rmu^2 a \ell^3 \frac{\rmu a^3}{K_{\ell}^6 (d s)^6}, \nonumber 
\end{align}
where we have introduced
\begin{align}
{\mathcal A}_1(k) := (1-\varepsilon)\tau(k) + \rho_z \varepsilon^{-1/2} a,\qquad
{\mathcal B}_1(k) :=  \rho_z \widehat{W}_1(k).
\end{align}
Notice that \eqref{eq:Rough0} is easy given the discussion above.

We proceed to prove \eqref{eq:Rough2}.  
We symmetrize the term in $k$ as
\begin{align}
(2\pi)^{-3} \ell^3 &\int \Big({\mathcal A}_1(k) a_k^{\dagger} a_k
 +\frac{1}{2}{\mathcal B}_1(k) a_k^{\dagger} a_{-k}^{\dagger} + \frac{1}{2}{\mathcal B}_1(k)a_k a_{-k}\Big) dk  \\
 &=
 \frac{1}{2} (2\pi)^{-3} \ell^3 \int \Big({\mathcal A}_1(k) a_k^{\dagger} a_k
 +{\mathcal A}_1(k) a_{-k}^{\dagger} a_{-k}
 +{\mathcal B}_1(k) a_k^{\dagger} a_{-k}^{\dagger} + {\mathcal B}_1(k)a_k a_{-k}\Big) dk. \nonumber
\end{align}
At this point we apply the `Bogolubov lemma', Lemma~\ref{eq:BogIneq}, to get
\begin{align}
{\mathcal A}_1(k) a_k^{\dagger} a_k
 +{\mathcal A}_1(k) a_{-k}^{\dagger} a_{-k}&
 +{\mathcal B}_1(k) a_k^{\dagger} a_{-k}^{\dagger} + {\mathcal B}_1(k)a_k a_{-k} \nonumber \\
& \geq
 - \left( {\mathcal A}_1(k)  - \sqrt{ {\mathcal A}_1(k)^2 - |{\mathcal B}_1(k) |^2}\right),
 \end{align}
where we have also used \eqref{eq:Commutator}.

Notice that (using \eqref{eq:W1-g}) $|{\mathcal B}_1(k) |/{\mathcal A}_1(k) \leq C \varepsilon^{1/2}$. 
Therefore, for $\varepsilon$ sufficiently small, a Taylor expansion gives
\begin{align}
- \left( {\mathcal A}_1(k)  - \sqrt{ {\mathcal A}_1(k)^2 - |{\mathcal B}_1(k) |^2}\right)
\geq -(\frac{1}{2} + C \varepsilon) \frac{|{\mathcal B}_1|^2}{{\mathcal A}_1}.
\end{align}
Below we will need the following estimate of an integral,
\begin{align}
\int_{\{|k| \geq (d s \ell)^{-1}\}}  \frac{|{\mathcal B}_1|^2}{2  \tau(k) }
&\leq \rho_z^2 \Big( \int \frac{\widehat{W}_1(k)^2}{2 k^2}
+ C \frac{a^2}{d s \ell}\int_{\{ (d s \ell)^{-1} \leq |k| \leq a^{-1}\}} |k|^{-3}
+ C \frac{a}{d s \ell} \int \frac{\widehat{W}_1(k)^2}{2 k^2} \Big) \nonumber \\
&\leq \rho_z^2 (1+ \frac{R^2}{\ell^2}) \widehat{g\omega}(0)
+
C \rho_z^2 a^2 (d s \ell)^{-1} (1 + \log( d s \ell/a)).
\end{align}
where we used that $0 \leq k^2 - \tau(k) \leq 2 |k| (d s \ell)^{-1}$ for $|k| \geq (d s \ell)^{-1}$ and we also used \eqref{eq:I2-integral}.

Inserting these considerations, we find
\begin{align}\label{eq:estimate}
&(2\pi)^{-3} \ell^3 \int {\mathcal A}_1(k) a_k^{\dagger} a_k
 +\frac{1}{2}{\mathcal B}_1(k) \left( a_k^{\dagger} a_{-k}^{\dagger} + a_k a_{-k}\right)dk \nonumber \\
 &\geq
 -(\frac{1}{2} + C\varepsilon)(2\pi)^{-3} \ell^3 \Big( \int_{|k| \leq (d s \ell)^{-1}} \varepsilon^{1/2} \frac{|{\mathcal B}_1|^2}{2  \rho_z a }+
 \int_{|k| \geq (d s \ell)^{-1}}  \frac{|{\mathcal B}_1|^2}{2 (1-\varepsilon) \tau(k) }\Big) \nonumber \\
&\geq
- \frac{1}{2} \rho_z^2 \ell^3 \widehat{g\omega}(0) 
- C \ell^3 \rho_z a\Big(  \varepsilon^{1/2}  (d s \ell)^{-3} + 
\rho_z ( \varepsilon + \frac{R^2}{\ell^2}) 
+ C \rho_z \frac{a}{d s \ell} \big(1 + \log\big( \frac{d s \ell}{a}\big)\big)
 \Big) \nonumber \\
 &\geq 
 - \frac{1}{2} \rho_z^2 \ell^3 \widehat{g\omega}(0)  - 
 C \ell^3 \rho_z^2 a \left(\varepsilon + \frac{R^2}{\ell^2} + \frac{a}{d s \ell} \big(1 + \log\big( \frac{d s \ell}{a}\big)\big)
 \right)
 - C a (d s)^{-6} \ell^{-3}.
\end{align}
This is easily seen to be consistent with \eqref{eq:Rough2} and
finishes the proof of \eqref{eq:Rough2}.

To prove \eqref{eq:Rough1} we use a similar approach. Notice that by definition \eqref{eq:kroelleQ3z}, ${\mathcal Q}_3(z)$ only lives in the high momentum region $P_H$.
For these momenta we have $\tau(k) \geq \frac{1}{2} k^2$. Therefore, dropping a part of the kinetic energy, it suffices to bound
\begin{align}
\ell^3 (2\pi)^{-3} \int_{\{ k \in P_H\}} \left(\frac{\varepsilon}{2} k^2 a_k^{\dagger} a_k +
(2\pi)^{-3} \int f_L(s)
\widehat{W}_1(k) (\overline{z} \widetilde{a}_s^{\dagger}  a_{s-k} a_k
+ a_k^{\dagger}  a^{\dagger}_{s-k} \widetilde{a}_s z) \,ds\right)\,dk.
\end{align}
We estimate, with $\widetilde{b}_k := a_k + 2(2\pi)^{-3} \int f_L(s)
\frac{ \widehat{W}_1(k)}{\varepsilon k^2} z a_{s-k}^{\dagger} \widetilde{a}_s   \,ds$,
\begin{align}\label{eq:RoughQ3_1}
&\ell^3 (2\pi)^{-3} \int_{\{ k \in P_H\}} \left(\frac{\varepsilon}{2} k^2 a_k^{\dagger} a_k +
(2\pi)^{-3} \int f_L(s)
\widehat{W}_1(k) (\overline{z} \widetilde{a}_s^{\dagger}  a_{s-k} a_k
+ a_k^{\dagger}  a^{\dagger}_{s-k} \widetilde{a}_s z) \,ds\right)\,dk\nonumber \\
&=\ell^3 (2\pi)^{-3} \int_{\{ k \in P_H\}}\left(
\frac{\varepsilon}{2} k^2 \widetilde{b}_k^{\dagger} \widetilde{b}_k
- 4 (2\pi)^{-6} 
\iint f_L(s) f_L(s') \frac{ \widehat{W}_1(k)^2}{\varepsilon k^2} |z|^2 \widetilde{a}_{s'}^{\dagger}  a_{s'-k} a_{s-k}^{\dagger} \widetilde{a}_s\right)dk \nonumber \\
&\geq- 4\varepsilon^{-1}
\ell^3 (2\pi)^{-9} \int_{\{ k \in P_H\}}  \frac{ \widehat{W}_1(k)^2}{ k^2} |z|^2
\iint f_L(s) f_L(s') \widetilde{a}_{s'}^{\dagger}  (a_{s-k}^{\dagger} a_{s'-k} +[a_{s'-k}, a_{s-k}^{\dagger}]) \widetilde{a}_s .
\end{align}
On the term without a commutator, we estimate $\widetilde{a}_{s'}^{\dagger}  a_{s-k}^{\dagger} a_{s'-k}\widetilde{a}_s$ by Cauchy-Schwarz and (since $k \in P_H$),
$\frac{ \widehat{W}_1(k)^2}{k^2} \leq C K_H^2 a^4$.
Therefore, for a $\Phi$ satisfying \eqref{eq:LocalizedAftercNumber}, we find
\begin{align}\label{eq:RoughQ3_2}
\langle \Phi, &
\ell^3 \int_{\{ k \in P_H\}}  \frac{ \widehat{W}_1(k)^2}{ k^2} |z|^2 \iint f_L(s) f_L(s') \widetilde{a}_{s'}^{\dagger}  a_{s-k}^{\dagger} a_{s'-k}  \widetilde{a}_s \Phi \rangle \nonumber \\
&\leq
C \rho_z  (\int_{\{|s| \leq 2 K_L \sqrt{\rmu a}\}}\,ds )K_H^2 a^4 {\mathcal M}^2\nonumber \\
&\leq
C \rho_z  a \ell^3 K_L^3  K_{\ell}^3 \frac{ a^3 K_H^2{\mathcal M}^2}{\ell^6}.
\end{align}
For the commutator term, we estimate (using \eqref{eq:Commutator_general} and the Cauchy-Schwarz inequality)
$$
\widetilde{a}_{s'}^{\dagger} [a_{s'-k}, a_{s-k}^{\dagger}]) \widetilde{a}_s \leq 2 \widetilde{a}_{s'}^{\dagger}\widetilde{a}_{s'} + 2 \widetilde{a}_s^{\dagger} \widetilde{a}_s
$$
and $\int \frac{ \widehat{W}_1(k)^2}{k^2} \leq Ca$.
This leads to (for a $\Phi$ satisfying \eqref{eq:LocalizedAftercNumber}),
\begin{align}\label{eq:RoughQ3_3}
\langle \Phi, &\ell^3 \int_{\{ k \in P_H\}}  \frac{ \widehat{W}_1(k)^2}{ k^2} |z|^2
\iint f_L(s) f_L(s') \widetilde{a}_{s'}^{\dagger} [a_{s'-k}, a_{s-k}^{\dagger}] \widetilde{a}_s \Phi \rangle \nonumber \\
&\leq 
C {\mathcal M} a |z|^2 \int_{\{|s| \leq 2 K_L \sqrt{\rmu a}\}}\,ds \nonumber \\
&\leq C a \rho_z \frac{\mathcal M}{\ell^3} \ell^3 K_L^3 K_{\ell}^3.
\end{align}
Combining the estimates \eqref{eq:RoughQ3_1}, \eqref{eq:RoughQ3_2} and \eqref{eq:RoughQ3_3} proves \eqref{eq:Rough1}.

We choose 
\begin{align}\label{eq:ChoiceEpsilon}
\varepsilon=\frac{{\mathcal M}^{1/2}}{\sqrt{(\rmu + \rho_z) \ell^3}}.
\end{align}
We will add the estimates of \eqref{eq:Rough0}, \eqref{eq:Rough1} and \eqref{eq:Rough2} with this choice of $\varepsilon$. Notice that since ${\mathcal M }
\leq \rmu \ell^3$ the contribution from \eqref{eq:Rough0} will be smaller than the terms appearing in the other estimates.
Therefore we get,
\begin{align}
\Big\langle \Phi,  &\left(\frac{1}{2} \rho_z^2 \ell^3 \widehat{g\omega}(0) + (2\pi)^{-3} \ell^3 \int  \tau (k) a_k^{\dagger} a_k \,dk + {\mathcal Q}_3(z)\right) \, \Phi \Big\rangle \nonumber \\
 &\geq
 - C \rho_z a \ell^3 \Big( 
 \frac{
 \sqrt{{\mathcal M}(\rmu+\rho_z)\ell^3}}{\ell^3}\Big(K_L^3 K_H^2 (\rmu a^3)^{3/2} {\mathcal M}
+
 K_L^3 K_{\ell}^3
\Big) + \rho_z \frac{R^2}{\ell^2}
 \Big).
\end{align}
This finishes the proof of \eqref{eq:RoughBelow}.
\end{proof}

\section{More precise energy estimates}
\label{sec:Precise}
From Proposition~\ref{prop:FarAway} above, we see that the energy is too high unless $\rho_z \approx \rmu$. In this section we will give precise energy bounds in the complementary regime. We will always assume that
\begin{align}\label{eq:Close}
|\rho_z - \rmu | \leq \rmu C \max\left( \Big(\delta_1 + \delta_2+ \frac{\rmu a^3}{K_{\ell}^6 (d s)^6}\big)^{\frac{1}{2}}, (\rmu a^3)^{\frac{1}{4}} \right),
\end{align}
with the notation from Proposition~\ref{prop:FarAway}.

We will need the condition that
\begin{align}\label{eq:NewConditionNew}
K_{\ell}^2 \max\left( \Big(\delta_1 + \delta_2+ \frac{\rmu a^3}{K_{\ell}^6 (d s)^6}\big)^{\frac{1}{2}}, (\rmu a^3)^{\frac{1}{4}} \right) \leq C^{-1},
\end{align}
for some sufficiently large universal constant.
This condition is satisfied by \eqref{con:rough1}, \eqref{con:rough2}, \eqref{con:sdKellKB}
and \eqref{con:KBKell}.

Notice, using \eqref{eq:Close} and \eqref{eq:NewConditionNew}, that
\begin{align}
\frac{|\rho_z-\rmu|}{\rmu} \leq C^{-1} K_{\ell}^{-2}.
\end{align}

For convenience of notation, we define the parameter $\lille$ to be the square of the ratio between $\sqrt{\rmu a}$ and the inner radius of $P_H$, i.e.
\begin{align}
\label{eq:DefLille}
\lille:= \frac{\rmu a}{K_H^{-2} a^{-2}} = (\rmu a^3)^{\frac{1}{6}}  \KH^2.
\end{align}
Using \eqref{con:KLKH}, we see that $\lille \ll 1$.

We define the quadratic Bogolubov Hamiltonian as follows,
\begin{align}\label{eq:BogHamPositive-1}
{\mathcal K}^{\rm Bog} =
\frac{1}{2} (2\pi)^{-3} \ell^3 \int \Big(& {\mathcal A}(k) (a_k^{\dagger} a_k + a^{\dagger}_{-k} a_{-k}) + {\mathcal B}(k) (a_k^{\dagger} a_{-k}^{\dagger} +  a_{k} a_{-k} )\nonumber \\
&+
{\mathcal C}(k) (a^{\dagger}_k + a^{\dagger}_{-k}+a_k + a_{-k}) \Big)
 \,dk ,
\end{align}
with
\begin{align}
{\mathcal A}(k) :=
\tau(k) + \rho_z \widehat{W_1}(k), \quad
{\mathcal B}(k) :=\rho_z \widehat{W_1}(k),\quad
{\mathcal C}(k):=  \ell^{-3} (\rho_z - \rmu) \widehat{W_1}(0) \widehat{\chi_{\Lambda}}(k) z\,.\quad
\end{align}
With this notation, we can rewrite/estimate ${\mathcal K}(z)$ from \eqref{eq:Kz} as follows,
\begin{align}\label{eq:SimplifyKz}
{\mathcal K}(z) 
& = {\mathcal K}^{\rm Bog} 
+ \frac{1}{2} \rho_z^2 \ell^{3} 
\Big(\widehat{g}(0) + \widehat{g\omega}(0)\Big)
  - \rmu   \widehat{g}(0) \rho_z \ell^3\nonumber \\
  &\quad  +  \frac{b}{2 \ell^2} n_{+} + \varepsilon_T\frac{b}{2 d^2 \ell^2} n_{+}^{H} 
  + (\rho_z - \rmu) \widehat{W}_1(0) (2\pi)^{-3} \ell^{3} \int a_k^{\dagger} a_k\,dk \nonumber \\
  &\quad
  + {\mathcal Q}_1^{\rm ex}(z)+
  {\mathcal Q}_2^{\rm ex}(z) + {\mathcal Q}_3(z) \nonumber \\
 & \geq
  - \frac{1}{2} \rmu^2 \ell^{3} \widehat{g}(0)
+ \frac{1}{2} \rho_z^2 \ell^{3} \widehat{g\omega}(0)
+ \frac{1}{2} (\rho_z-\rmu)^2 \ell^{3} \widehat{g}(0) + {\mathcal K}^{\rm Bog} \nonumber \\
  &\quad  +  \frac{b}{4 \ell^2} n_{+} + \varepsilon_T\frac{b}{2 d^2 \ell^2} n_{+}^{H} 
  + {\mathcal Q}_1^{\rm ex}(z)+
  {\mathcal Q}_2^{\rm ex}(z) + {\mathcal Q}_3(z).
\end{align}
Here we used \eqref{eq:NewConditionNew} to absorb a quadratic part in the gap.

\subsection{The Bogolubov Hamiltonian}
\begin{theorem}[Analysis of Bogolubov Hamiltonian]\label{thm:BogHamDiag}
Assume that $\Phi$ satisfies \eqref{eq:LocalizedAftercNumber} and that $\frac{1}{2} \rmu \leq \rho_z \leq 2\rmu$.
Let $\lille$ be the parameter defined in \eqref{eq:DefLille}.
Then,
\begin{align}\label{eq:BogHamDiag}
\langle \Phi, {\mathcal K}^{\rm Bog} \Phi \rangle &\geq
(2\pi)^{-3} \ell^3 \langle \Phi, \int {\mathcal D}_k b_k^{\dagger} b_{k}\,dk \, \Phi \rangle  \nonumber \\
&\quad -\frac{1}{2} (2\pi)^{-3} \ell^3 \int  \left( {\mathcal A}(k) - \sqrt{{\mathcal A}(k)^2 - {\mathcal B}(k)^2}\right) \,dk \nonumber\\
&\quad
- (\rho_z -\rmu)^2 \frac{\widehat{g}(0)}{2} \ell^3 \left(1+ C \frac{R^2}{a^2} (\rmu a^3)\right)
- C \rmu^2 a \ell^3 K_{{\mathcal M}} K_{\ell}^{-3/2} (K_{\ell}^2 \lille)^{\frac{M-5}{2}}.
\end{align}
Here
\begin{align}
{\mathcal D}_k := \frac{1}{2}\left( {\mathcal A}(k) + \sqrt{{\mathcal A}(k)^2 - {\mathcal B}(k)^2}\right),
\end{align}
and
\begin{align}\label{eq:Defb_k}
b_k :=  a_{k} + \alpha_k a_{-k}^{\dagger} + c_k,
\end{align}
with
\begin{align}
\alpha_k:= {\mathcal B}(k)^{-1} \left( {\mathcal A}(k) - \sqrt{{\mathcal A}(k)^2 - {\mathcal B}(k)^2}\right),
\end{align}
and
\begin{align}
c_k:=\begin{cases}\frac{2{\mathcal C}(k)}{{\mathcal A}(k) + {\mathcal B}(k) + \sqrt{{\mathcal A}(k)^2 - {\mathcal B}(k)^2}}, & |k| \leq \frac{1}{2}K_H^{-1} a^{-1},\\
0, & |k| > \frac{1}{2} K_H^{-1} a^{-1}.
\end{cases}
\end{align}
\end{theorem}

\begin{proof}

To simplify later calculations we start by removing ${\mathcal C}(k)$ for $|k| > \frac{1}{2} K_H^{-1} a^{-1}$ from ${\mathcal K}^{\rm Bog}$, so we aim to prove
\begin{align}\label{eq:RemoveC's}
\frac{1}{2}& (2\pi)^{-3} \ell^3 \int_{\{|k| > \frac{1}{2} K_H^{-1} a^{-1}\}} {\mathcal C}(k) (a^{\dagger}_k + a^{\dagger}_{-k}+a_k + a_{-k}) \Big)
 \,dk 
 \geq - C \rmu^2 a \ell^3 (\rmu a^3)^{\frac{1}{2} + \frac{1}{6}}.
\end{align}
Obviously,
$$
a_k+ a_k^{\dagger} \leq a_k^{\dagger} a_k + 1.
$$
Therefore,
\begin{align}
\frac{1}{2}& (2\pi)^{-3} \ell^3 \int_{\{|k| > \frac{1}{2} K_H^{-1} a^{-1}\}} {\mathcal C}(k) (a^{\dagger}_k + a^{\dagger}_{-k}+a_k + a_{-k}) \Big)
 \,dk \nonumber \\
 &\geq
 - (2\pi)^{-3} |\rho_z -\rmu| \widehat{W_1}(0) |z|
 \int_{\{|k| > \frac{1}{2} K_H^{-1} a^{-1}\}} |\widehat{\chi_{\Lambda}}(k)| (a_k^{\dagger} a_k +1)\,dk \nonumber \\
 &\geq
 -C |\rho_z -\rmu| \widehat{W_1}(0) |z| ( n_{+} + 1) \epsilon(\chi),
\end{align}
where
\begin{align}
\epsilon(\chi) := \ell^{-3}  \sup_{\{|k| > \frac{1}{2} K_H^{-1} a^{-1}\}} (1+(k\ell)^2)^2  |\widehat{\chi_{\Lambda}}(k)|
\leq
C ( K_{\ell}^{-2} \lille )^{\widetilde{M}-2},
\end{align}
where we used Lemma~\ref{lem:DecaychiHat} to get the last estimate.
Estimating $n_{+}$ using \eqref{eq:LocalizedAftercNumber} and using \eqref{eq:Close} to control $|z|$, it is elementary to conclude \eqref{eq:RemoveC's} for this part.

By the estimate above, it suffices to consider
\begin{align}
\widetilde{\mathcal K}^{\rm Bog} :=
\frac{1}{2} (2\pi)^{-3} \ell^3 \int \Big( &{\mathcal A}(k) (a_k^{\dagger} a_k + a^{\dagger}_{-k} a_{-k}) + {\mathcal B}(k) (a_k^{\dagger} a_{-k}^{\dagger} +  a_{k} a_{-k} )\nonumber \\
&+
\widetilde{\mathcal C}(k) (a^{\dagger}_k + a^{\dagger}_{-k}+a_k + a_{-k}) \Big)
 \,dk ,
\end{align}
with ${\mathcal A}, {\mathcal B}$ from \eqref{eq:BogHamPositive-1} and
\begin{align}
\widetilde{\mathcal C}(k):=  \begin{cases} 0, & |k| \geq \frac{1}{2} K_{H}^{-1} a^{-1}, \\
\ell^{-3} (\rho_z - \rmu) \widehat{W_1}(0) \widehat{\chi_{\Lambda}}(k) z, & \text{ else}.
\end{cases}
\end{align}

With the notation from Theorem~\ref{thm:BogHamDiag} and using Theorem~\ref{thm:bogolubov-complete} combined with \eqref{eq:Commutator} we find
\begin{align}\label{eq:LowerBoundHBog}
\widetilde{\mathcal K}^{\rm Bog} 
&\geq
(2\pi)^{-3} \ell^3 \int {\mathcal D}_k b_k^{\dagger} b_{k}\,dk  \nonumber\\
&\quad-\frac{1}{2} (2\pi)^{-3} \ell^3 \int  \left( {\mathcal A}(k) - \sqrt{{\mathcal A}(k)^2 - {\mathcal B}(k)^2}\right) \,dk \nonumber \\
&\quad - (\rho_z - \rmu)^2 \widehat{W_1}(0)^2 z^2  (2\pi)^{-3} \ell^{-3} \int_{\{|k| \leq \frac{1}{2} K_{H}^{-1} a^{-1}\}} \frac{|{\widehat \chi}_{\Lambda}(k)|^2}{{\mathcal A}(k)+ {\mathcal B}(k)}.
\end{align}
It is elementary, using that $W_1$ is even, that 
\begin{align}
\left| \widehat{W}_1(k) - \widehat{W}_1(0) \right|
\leq C a (kR)^2.
\end{align}
Therefore we easily get the lower bound
\begin{align}
{\mathcal A}(k)+ {\mathcal B}(k) \geq 2 \rho_z \widehat{W_1(0)} (1 - C (\rmu a^3) \frac{R^2}{a^2}),
\end{align}
using that the kinetic energy is dominating, unless $|k| \leq C \sqrt{\rmu a}$.

Therefore, the last term in \eqref{eq:LowerBoundHBog} becomes controlled as
\begin{align}
&(\rho_z - \rmu)^2 \widehat{W_1}(0)^2 z^2  (2\pi)^{-3} \ell^{-3} \int_{\{|k| \leq \frac{1}{2} K_{H}^{-1} a^{-1}\}} \frac{|{\widehat \chi}_{\Lambda}(k)|^2}{{\mathcal A}(k)+ {\mathcal B}(k)} \nonumber \\
&\leq
(\rho_z - \rmu)^2 \frac{\widehat{W_1}(0)}{2} \ell^3 (1 + C (\rmu a^3) \frac{R^2}{a^2})\nonumber \\
&\leq
(\rho_z - \rmu)^2 \frac{\widehat{g}(0)}{2} \ell^3 (1 + C (\rmu a^3)\frac{R^2}{a^2}),
\end{align}
where we used that $\ell^{-2} \ll \rmu a$ to get the last estimate.

This finishes the proof of Theorem~\ref{thm:BogHamDiag}.
\end{proof}

\begin{remark}
We notice that following commutation relations (using the ones for the $a_k$'s \eqref{eq:Commutator_general} and that $\widehat{\chi}$ is even and real).
\begin{align}\label{eq:bComms1}
[ b_k, b_{k'}]
= (\alpha_k - \alpha_{k'}) \Big(  \widehat{\chi^2}((k+k')\ell) - \widehat{\chi}(k \ell) \widehat{\chi}(k'\ell) \Big).
\end{align}
Also,
\begin{align}\label{eq:bComms2}
[ b_k, b_{k'}^{\dagger}] =
(1-\alpha_k \alpha_{k'}) \Big(\widehat{\chi^2}((k-k')\ell) -  \widehat{\chi}(k \ell) \widehat{\chi}(k'\ell) \Big).
\end{align}
\end{remark}

\begin{lemma}\label{lem:BogIntegral}
Assume that \eqref{eq:Close} holds and that 
$\frac{9}{10} \rmu \leq \rho_z \leq \frac{11}{10} \rmu$.
We have the following estimate
\begin{align}\label{eq:BogIntegral}
&-\frac{1}{2} (2\pi)^{-3} \ell^3 \int  \left( {\mathcal A}(k) - \sqrt{{\mathcal A}(k)^2 - {\mathcal B}(k)^2}\right) \,dk \nonumber \\
&\geq - \frac{\widehat{g\omega}(0)}{2} \rho_z^2 \ell^3
+ 4\pi  \frac{128 }{15\sqrt{\pi}} \rho_z^2 a \sqrt{\rho_z a^3} \ell^3 
- C \epsilon(\rmu,\rho_z) \rho_z^2 a\sqrt{\rho_z a^3} \ell^3
- C \rho_z^2 \ell^3 \frac{R^2}{\ell^2} \widehat{g\omega}(0),
\end{align}
with
\begin{align}\label{eq:LebesgueError-Total}
\epsilon(\rmu,\rho_z) &= (\rmu a)^{\frac{1}{4}} \sqrt{R} + \varepsilon_T +
 (K_{\ell} s)^{-1} \left( 1 + \log(d^{-1}) + \log( \frac{K_{\ell} d s}{(\rmu a^3)^{1/2}}) \right) \nonumber \\
 &\quad 
+ \varepsilon_T (K_{\ell} d s)^{-1} \left( 1 +\log( \frac{K_{\ell} d s}{(\rmu a^3)^{1/2}}) \right).
\end{align}
\end{lemma}

\begin{proof}
We regularize the integral as
\begin{align}\label{LHYIntegral1}
&\int  {\mathcal A}(k) - \sqrt{{\mathcal A}(k)^2 - {\mathcal B}(k)^2} \,dk \nonumber \\
&=
\int  {\mathcal A}(k) - \sqrt{{\mathcal A}(k)^2 - {\mathcal B}(k)^2} - \rho_z^2 \frac{\widehat{W_1}(k)^2}{2 k^2}  \,dk
+
\rho_z^2 \int \frac{\widehat{W_1}(k)^2}{2 k^2}  \,dk.
\end{align}

The last integral is controlled by \eqref{eq:I2-integral} and contributes with the first and the last term in \eqref{eq:BogIntegral}.

In the regularized integral in \eqref{LHYIntegral1} we perform the change of variables $\sqrt{\rho_z a} \,t = k$.
In this way we get
\begin{align}
&\int  {\mathcal A}(k) - \sqrt{{\mathcal A}(k)^2 - {\mathcal B}(k)^2} - \rho_z^2 \frac{\widehat{W_1}(k)^2}{2 k^2}  \,dk
=
\rho_z^2 \sqrt{\rho_z a^3} a I_1,
\end{align}
with
\begin{align}
I_1 &= \int \alpha(t)  -
\sqrt{\alpha(t)^2 - \beta(t)^2} - \frac{\widehat{W_1}(\sqrt{\rho_z a} t)^2}{2 a^2 t^2}\,dt ,\nonumber \\
\alpha(t) &= \widetilde{\tau}(t) + a^{-1} \widehat{W_1}(\sqrt{\rho_z a} t), \nonumber \\
\beta(t) &=  a^{-1} \widehat{W_1}(\sqrt{\rho_z a} t),\nonumber\\
\widetilde{\tau}(t) &= (1-\varepsilon_T) \Big[ |t| - \frac{1}{2K_{\ell}s} (\frac{\rmu}{\rho_z})^{1/2}  \Big]_{+}^2
+ \varepsilon_T \Big[ |t| - \frac{1}{2K_{\ell}ds} (\frac{\rmu}{\rho_z})^{1/2} \Big]_{+}^2.
\end{align}
We will prove that $I_1 \approx - 64 \pi^4 \frac{128}{15\sqrt{\pi}}$.
For this we write $I_1$ as
\begin{align}
I_1 &= \int \alpha(t)  -
\sqrt{\alpha(t)^2 - \beta(t)^2} - \frac{\beta^2}{2t^2} \nonumber \\
&= \int \alpha(t)  - \frac{\beta^2}{2\alpha}-
\sqrt{\alpha(t)^2 - \beta(t)^2} + (\frac{\beta^2}{2\alpha}- \frac{\beta^2}{2t^2}) \nonumber \\
&=
\int \alpha(t)  - \frac{\beta^2}{2\alpha}-
\sqrt{\alpha(t)^2 - \beta(t)^2} + \frac{\beta^2}{2} \frac{t^2-\widetilde{\tau} - \beta}{t^2 \alpha} \nonumber \\
&= I_1' + I_1'',
\end{align}
with
\begin{align}\label{eq:SplitIntegrals}
I_1' &:= \int \alpha(t)  - \frac{\beta^2}{2\alpha}-
\sqrt{\alpha(t)^2 - \beta(t)^2} - \frac{\beta^3}{2 t^2 \alpha}, \nonumber\\
I_1''&:= \int \frac{\beta^2}{2} \frac{t^2-\widetilde{\tau}}{t^2 \alpha}.
\end{align}
It is not difficult to apply dominated convergence to the integral $I_1'$ to get
\begin{align}
I_1' \approx \int_{{\mathbb R}^3} t^2 + 8\pi - \frac{(8\pi)^2}{2 t^2}
- \sqrt{(t^2+8\pi)^2-(8\pi)^2} \,dt = -64 \pi^4 \frac{128}{15\sqrt{\pi}} .
\end{align}
More precisely, 
we will prove that
\begin{align}\label{eq:LebesgueError}
\left| I_1' - \int_{{\mathbb R}^3} t^2 + 8\pi - \frac{(8\pi)^2}{2 t^2}
- \sqrt{(t^2+8\pi)^2-(8\pi)^2} \,dt \right| \leq C 
 (\rmu a)^{\frac{1}{4}} \sqrt{R} + \varepsilon_T + (K_{\ell} s)^{-1}.
\end{align}
Notice that this is consistent with \eqref{eq:LebesgueError-Total}.

The part of both integrals where $|t| \leq 10 (K_{\ell } s)^{-1}$ is bounded by
$$
C (K_{\ell } s)^{-1},
$$
for sufficiently small $\rmu$ (using that $\rho_z \approx \rmu$).
This is in agreement with \eqref{eq:LebesgueError}.

For $|t| \geq 10 (K_{\ell } s)^{-1}$ we will use
\begin{align}\label{eq:tautilde-Est}
|\beta(t) - 8\pi| \leq C \sqrt{\rmu a} R |t|, \qquad
0 \leq t^2 - \widetilde{\tau}(t) \leq \varepsilon_T t^2 + \frac{1}{K_{\ell}s} (\frac{\rmu}{\rho_z})^{1/2} |t|.
\end{align}
Notice that it follows by interpolation that $|\beta(t) - 8\pi| \leq C (\rmu a)^{\frac{1}{4}} R^{\frac{1}{2}} |t|^{\frac{1}{2}}$ and also that
$\widetilde{\tau} \geq \frac{1}{2} t^2$ when $\varepsilon_T$ is sufficiently small (since $\frac{\rmu}{\rho_z}$ is close to $1$).

For $|t| \geq 100$ we use Taylor's formula with remainder (applied to $\sqrt{1-x}$) to write
\begin{align}\label{eq:Taylor}
&\int_{\{|t| \geq 10 (K_{\ell } s)^{-1}\}} \alpha(t)  - \frac{\beta^2}{2\alpha}-
\sqrt{\alpha(t)^2 - \beta(t)^2} - \frac{\beta^3}{2 t^2 \alpha}\,dt \nonumber \\
&\qquad -
\int_{\{|t| \geq 10 (K_{\ell } s)^{-1}\}}  t^2 + 8\pi - \frac{(8\pi)^2}{2 t^2}
- \sqrt{(t^2+8\pi)^2-(8\pi)^2} \,dt  \nonumber \\
&= 
\int_{\{ 10 (K_{\ell } s)^{-1} \leq |t| \leq 100\}} \Big( 
(\alpha - t^2 - 8\pi) - ( \frac{\beta^2}{2\alpha} - \frac{(8\pi)^2}{2 (t^2+8\pi)}) \nonumber \\
&\qquad \qquad \qquad 
- (\sqrt{\alpha(t)^2 - \beta(t)^2} -  \sqrt{(t^2+8\pi)^2-(8\pi)^2} )\Big) \,dt 
\nonumber \\
&\quad +\int_{\{|t| \geq 100\}} \int_0^1  f(\widetilde{\tau}, \beta, \sigma) -f(t^2, 8\pi , \sigma)  \,d\sigma\,dt \nonumber \\
&\quad - \int_{\{|t| \geq 10 (K_{\ell } s)^{-1}\}}   \frac{\beta^3}{2 t^2 \alpha} - \frac{ (8 \pi)^3}{2 t^2 (t^2 + 8 \pi)}\,dt,
\end{align}
with
\begin{align}
f(\tau, \beta, \sigma) := \frac{-\beta^4}{4} [ \tau^2 + 2 \beta \tau + (1-\sigma) \beta^2]^{-3/2} (1-\sigma).
\end{align}
The last integral in \eqref{eq:Taylor} is easily estimated, as
\begin{align}
\left| \int_{\{|t| \geq 10 (K_{\ell } s)^{-1}\}}   \frac{\beta^3}{2 t^2 \alpha} - \frac{ (8 \pi)^3}{2 t^2 (t^2 + 8 \pi)}\,dt \right| \leq
C \left( (\rmu a)^{\frac{1}{4}} \sqrt{R} + \varepsilon_T + (K_{\ell} s)^{-1}\right),
\end{align}
in agreement with \eqref{eq:LebesgueError}.

For the Taylor expansion part in \eqref{eq:Taylor}, we use that
$\widetilde{\tau}^2 + 2 \beta \widetilde{\tau} + (1-\sigma) \beta \geq \frac{1}{4} t^4$, when $|t| \geq 100$. Therefore,
\begin{align}
  \Big| f(\widetilde{\tau}, \beta, \sigma) &-f(t^2, 8\pi , \sigma) \Big|  \nonumber\\
&\leq
C |\beta^4 - (8\pi)^4| t^{-6} \nonumber \\
& \quad+ C t^2 \big| [ \widetilde{\tau}^2 + 2 \beta \widetilde{\tau} + (1-\sigma) \beta^2]^{-2} - 
[ t^4 + 16\pi t^2 + (1-\sigma) (8\pi)^2]^{-2} \Big| \nonumber \\
& \quad + C t^{-8}
\frac{ (\widetilde{\tau}^2 + 2 \beta \widetilde{\tau} + (1-\sigma) \beta^2) - (t^4 + 16\pi t^2 + (1-\sigma) (8\pi)^2)}{
\sqrt{\widetilde{\tau}^2 + 2 \beta \widetilde{\tau} + (1-\sigma) \beta^2} + \sqrt{t^4 + 16\pi t^2 + (1-\sigma) (8\pi)^2}}.
\end{align}
Now the integrals can easily be estimated to get an error consistent with \eqref{eq:LebesgueError}.

Finally, we consider the integral over $\{ 10 (K_{\ell } s)^{-1} \leq |t| \leq 100\}$ in \eqref{eq:Taylor}. Here one may estimate term by term and use the finiteness of the domain of integration.
Therefore, this part is also consistent with \eqref{eq:LebesgueError}, which finishes the proof of \eqref{eq:LebesgueError}.

The integral $I_1''$ from \eqref{eq:SplitIntegrals} is split in $3$ parts. For $|t| \leq 10 (K_{\ell}s)^{-1}$, we have $0 \leq t^2 - \widetilde{\tau}(t) \leq t^2$. Therefore,
\begin{align}
\left|  \int_{\{ |t| \leq 10 (K_{\ell}s)^{-1}\} } \frac{\beta^2}{2} \frac{t^2-\widetilde{\tau}}{t^2 \alpha} \right|
\leq C (K_{\ell}s)^{-1},
\end{align}
which is consistent with \eqref{eq:LebesgueError-Total}.

For $10 (K_{\ell} s)^{-1} \leq |t| \leq 10 (K_{\ell} d s)^{-1} $, we have \eqref{eq:tautilde-Est} above.
Therefore,
\begin{align}
\left|  \int_{\{ 10 (K_{\ell} s)^{-1} \leq |t| \leq 10 (K_{\ell} d s)^{-1}\} } \frac{\beta^2}{2} \frac{t^2-\widetilde{\tau}}{t^2 \alpha} \right|
\leq
C \varepsilon_T (K_{\ell} d s)^{-1} + C (K_{\ell} s)^{-1} \log(d^{-1}),
\end{align}
in agreement with \eqref{eq:LebesgueError-Total}.

Finally the case $|t| \geq 10 (K_{\ell} d s)^{-1}$.
Here, $0 \leq t^2 - \widetilde{\tau}(t) \leq C |t| ( (K_{\ell} s)^{-1} + \varepsilon_T (K_{\ell} d s)^{-1})$ and $\alpha \geq \frac{1}{2} t^2$.
Therefore, 
\begin{align}\label{eq:Logarithm}
&\left|  \int_{\{ |t| \geq 10 (K_{\ell} d s)^{-1} \} } \frac{\beta^2}{2} \frac{t^2-\widetilde{\tau}}{t^2 \alpha} \right| \nonumber \\
&\leq
C ( (K_{\ell} s)^{-1} + \varepsilon_T (K_{\ell} d s)^{-1}) \int_{\{ 10 (K_{\ell} d s)^{-1}\leq |t| \leq (\rho_z a^3)^{-1/2}\}} |t|^{-3} \nonumber \\
&\quad + C ( (K_{\ell} s)^{-1} + \varepsilon_T (K_{\ell} d s)^{-1}) (\rho_z a^3)^{1/2} a^{-2} 
 \int\frac{\widehat{W}_1(\sqrt{\rho_z a} t)^2}{t^2} \nonumber \\
 &\leq C ( (K_{\ell} s)^{-1} + \varepsilon_T (K_{\ell} d s)^{-1}) \left( \log\big( \frac{K_{\ell} d s }{(\rmu a^3)^{1/2}}\big) + 1 \right). 
\end{align}
Since this estimate is also in agreement with \eqref{eq:LebesgueError-Total}, 
this finishes the proof of Lemma~\ref{lem:BogIntegral}.
\end{proof}

\subsection{The control of ${\mathcal Q}_3(z)$}
\label{subsec:Q3}
The quadratic Hamiltonian $(2\pi)^{-3} \ell^3 \int {\mathcal D}_k b_k^{\dagger} b_{k}\,dk$ from \eqref{eq:BogHamDiag} turns out to control the $3Q$-term ${\mathcal Q}_3(z)$ from \eqref{eq:kroelleQ3z}.
This we summarize as follows

\begin{theorem}\label{thm:Control3Q}
Assume that $\Phi$ satisfies \eqref{eq:LocalizedAftercNumber}.
Assume furthermore that \eqref{eq:Close} and \eqref{eq:AssumptionK_R} are satisfied.
Let $\lille$ be as defined in \eqref{eq:DefLille}.
We will furthermore assume \eqref{con:eTdK}, \eqref{con:KBellKM}, \eqref{con:KLKH}, \eqref{con:rough1}, \eqref{con:KellKM}, and \eqref{con:KellKH}.

Then,
\begin{align}\label{eq:Q3Cancels}
\Big\langle &
\Phi, (2\pi)^{-3} \ell^3 \int {\mathcal D}_k b_k^{\dagger} b_{k}\,dk
+
{\mathcal Q}_3(z) + {\mathcal Q}_2^{\rm ex}
+
\rho_z z \widehat{W_1 \omega}(0)
(2\pi)^{-3} \int \widehat{\chi_{\Lambda}^2}(s)  ( \widetilde{a}_s^{\dagger} + \widetilde{a}_s)\,ds \nonumber \\
&
+ \frac{b}{50} (\frac{1}{\ell^2} n_{+} + \frac{\varepsilon_{T}}{(d \ell)^2} n_{+}^H)
\Big) \Phi \Big\rangle \nonumber \\
&\geq
-C \rmu^2 a \ell^3 \Big[
\sqrt{\frac{{\mathcal M}}{|z|^2}} \Big(
\KH^{-1} (\rmu a^3)^{\frac{5}{12}} + (K_{\ell} K_{L})^{-M}
+ K_{\ell}^{3} K_{L}^{3} ( K_{\ell}^{-2} \lille )^{\frac{M-1}{2}}
\Big)  
\nonumber \\
&\quad\quad\quad\quad\quad\quad 
+  \lille^4 \frac{a}{\ell}
+  \sqrt{\rmu a^3} \left( K_{\ell}^{-3} d^{-12}\delta^2 
\left( K_{\ell}^{-2} \delta\right)^{M-1}\right) \Big].
\end{align}
\end{theorem}

\begin{proof}[Proof of Theorem~\ref{thm:Control3Q}]

Notice that 
\begin{align}\label{eq:BoverA}
|{\mathcal B}(k)/{\mathcal A}(k)| \leq C \lille , \qquad \forall |k| \geq \frac{1}{2} \KH^{-1} (\rmu a^3)^{5/12} a^{-1}.
\end{align}
In particular, $|{\mathcal B}(k)/{\mathcal A}(k)| \leq \frac{1}{2}$, for $\rmu$ sufficiently small.

This implies, by expansion of the square root that for all $|k| \geq \frac{1}{2} \KH^{-1} (\rmu a^3)^{5/12} a^{-1}$,
\begin{align}\label{eq:alphaInPH}
| \alpha_k| = | {\mathcal B}(k)|^{-1} \left( {\mathcal A}(k) - \sqrt{{\mathcal A}(k)^2 - {\mathcal B}(k)^2}\right)| \leq C \lille.
\end{align}
In particular, \eqref{eq:BoverA} and \eqref{eq:alphaInPH} are valid for $k=k'-s$, when $s \in P_L$ and $k' \in P_H$.

For later convenience, we reformulate the first-order operator in\eqref{eq:Q3Cancels} in terms of the $\widetilde{a}_s$. We get
\begin{align}
&-\rho_z z \widehat{W_1 \omega}(0)
(2\pi)^{-3} \int \widehat{\chi_{\Lambda}}(s)  ( a_s^{\dagger} +  a_s)\,ds \nonumber \\
&=
-\rho_z z \widehat{W_1 \omega}(0)
(2\pi)^{-3} \int \widehat{\chi_{\Lambda}^2}(s)  ( \widetilde{a}_s^{\dagger} + \widetilde{a}_s)\,ds \nonumber \\
&= -\rho_z z \widehat{W_1 \omega}(0)
(2\pi)^{-3}\ell^3  \int \widehat{\chi^2}(s\ell)  ( \widetilde{a}_s^{\dagger} + \widetilde{a}_s)\,ds.
\end{align}

We start by rewriting ${\mathcal Q}_3(z)$ in terms of the $b_k$'s defined in \eqref{eq:Defb_k}.
Notice that $c_k, c_{s-k}=0$ if $k \in P_H$ and $s\in P_L$.
We find the basic relation (we will freely use that all involved functions are symmetric, e.g. $\alpha_k = \alpha_{-k}$)
\begin{align}
a_{s-k} &= \frac{1}{1-\alpha_{s-k}^2} \Big(b_{s-k} - \alpha_{s-k} b_{k-s}^{\dagger}
\Big), &
a_k &= \frac{1}{1-\alpha_k^2} \Big(b_k - \alpha_k b_{-k}^{\dagger}
\Big).
\end{align}
Therefore,
\begin{align}\label{eq:asInTermsOfbs}
&a_{s-k} a_k  \\
&= \frac{1}{1-\alpha_k^2}  \frac{1}{1-\alpha_{s-k}^2}\Big(
b_{s-k} b_k  - \alpha_k b_{-k}^{\dagger} b_{s-k} - \alpha_{s-k} b_{k-s}^{\dagger} b_k
+ \alpha_k \alpha_{s-k} b_{k-s}^{\dagger}  b_{-k}^{\dagger} 
-\alpha_{k} [b_{s-k}, b_{-k}^{\dagger}]
\Big) .\nonumber
\end{align}

We will decompose ${\mathcal Q}_3(z)$ according to the different terms in \eqref{eq:asInTermsOfbs}, i.e.
\begin{align}
{\mathcal Q}_3(z) = {\mathcal Q}_3^{(1)}(z) + {\mathcal Q}_3^{(2)}(z) + {\mathcal Q}_3^{(3)}(z) + {\mathcal Q}_3^{(4)}(z), 
\end{align}
where
\begin{align}
{\mathcal Q}_3^{(1)}(z) &:= z \ell^3 (2\pi)^{-6} \iint_{\{ k \in P_H\}}
\frac{f_L(s)
\widehat{W}_1(k)}{(1-\alpha_k^2)(1-\alpha_{s-k}^2)} \left(\widetilde{a}_s^{\dagger}  
b_{s-k} b_k + \alpha_k \alpha_{s-k} \widetilde{a}_s^{\dagger}   b_{k-s}^{\dagger} b_{-k}^{\dagger} + h.c.
 \right), \nonumber \\
{\mathcal Q}_3^{(2)}(z) &:= - z \ell^3 (2\pi)^{-6} \iint_{\{ k \in P_H\}}
\frac{f_L(s)
\widehat{W}_1(k)\alpha_k}{(1-\alpha_k^2)(1-\alpha_{s-k}^2)} \left(\widetilde{a}_s^{\dagger}  b_{-k}^{\dagger} b_{s-k}
+ b_{s-k}^{\dagger}  b_{-k} \widetilde{a}_s \right), \nonumber \\
{\mathcal Q}_3^{(3)}(z) &:= - z \ell^3 (2\pi)^{-6} \iint_{\{ k \in P_H\}}
\frac{f_L(s)
\widehat{W}_1(k)\alpha_{s-k}}{(1-\alpha_k^2)(1-\alpha_{s-k}^2)} \left(\widetilde{a}_s^{\dagger}  b_{k-s}^{\dagger} b_{k}
+ b_{k}^{\dagger}  b_{k-s} \widetilde{a}_s \right), \nonumber \\
\intertext{and}
{\mathcal Q}_3^{(4)}(z) &:= (2\pi)^{-6} z \ell^3 \iint_{k \in P_H} f_L(s) \widehat{W}_1(k) \frac{-\alpha_k}{(1-\alpha_k^2)(1-\alpha_{s-k}^2)}[b_{s-k},b_{-k}^{\dagger}] (\widetilde{a}_s^{\dagger} + \widetilde{a}_s).
\end{align}
The different ${\mathcal Q}_3^{(j)}(z)$'s will be estimated individually. The result of this is summarized in Lemma~\ref{lem:EstimatesOnQ3-js}.
Theorem~\ref{thm:Control3Q} follows by adding the estimates of Lemma~\ref{lem:EstimatesOnQ3-js}.
We have used that the $K$'s are larger than $1$ and \eqref{con:KBellKM} to simplify the total remainder.
This finishes the proof.
\end{proof}

\begin{lemma}\label{lem:EstimatesOnQ3-js}
Let $\lille$ be as defined in \eqref{eq:DefLille}.
Assume that $\Phi$ satisfies \eqref{eq:LocalizedAftercNumber}.
Assume furthermore that \eqref{eq:AssumptionK_R}, \eqref{eq:Close}, \eqref{con:KLKH}, \eqref{con:rough1}, \eqref{con:KellKM}, \eqref{con:eTdK} and  \eqref{con:KellKH} are satisfied.
Then,
\begin{align}
\label{eq:EstimateQ3Tilde1} 
\Big\langle \Phi, &\Big( {\mathcal Q}_3^{(1)}(z)
+ (1-\lille^2) (2\pi)^{-3} \ell^3 \int_{\{|k| \geq \frac{1}{2} K_H^{-1} a^{-1}\}} {\mathcal D}_k b_{k}^{\dagger} b_k 
+ {\mathcal Q}_2^{\rm ex}
+ \frac{b}{100} (\frac{1}{\ell^2} n_{+} + \frac{\varepsilon_{T}}{(d \ell)^2} n_{+}^H)
\Big) \Phi \Big\rangle \nonumber \\
&\geq  -C \rmu^2 a \ell^3 \lille K_{\ell}^{-3/2} (\rmu a^3)^{\frac{1}{4}}  
({\mathcal M} + (K_{\ell}^3 K_L^3)) 
\left( K_{\ell}^{-2} \lille \right)^{\frac{M-1}{2}} \nonumber \\
&\quad  - C \rmu^2 a \ell^3 \sqrt{\rmu a^3} \left( K_{\ell}^{-3} d^{-12}\delta^2 
\left( K_{\ell}^{-2} \KH^2 (\rmu a^3)^{\frac{1}{6}}\right)^{M-1}\right), \\
\label{eq:EstimateQ3Tilde2+3}
\Big\langle \Phi,& \Big( {\mathcal Q}_3^{(2)}(z) + {\mathcal Q}_3^{(3)}(z)
+ \lille^2 (2\pi)^{-3} \ell^3 \int_{\{|k| \geq \frac{1}{2} K_H^{-1} a^{-1}\}} {\mathcal D}_k b_{k}^{\dagger} b_k
\Big) \Phi \Big\rangle 
\nonumber \\ &
\geq - C \rmu^2 a \ell^3 
( K_{\ell}^{-2}\lille )^{\widetilde{M}} K_L^3 K_{\ell}^{3/2} (\rmu a^3)^{\frac{1}{4}},\\
\label{eq:EstimateQ3Tilde4}
\Big\langle \Phi, &\Big( {\mathcal Q}_3^{(4)}(z) + \rho_z z \widehat{W_1 \omega}(0)
(2\pi)^{-3} \int \widehat{\chi_{\Lambda}^2}(s)  ( \widetilde{a}_s^{\dagger} + \widetilde{a}_s)\,ds 
+  \frac{1}{100} \frac{b}{\ell^2} n_{+}
\Big) \Phi \Big\rangle \nonumber \\
&\geq
-C \rho_z^2 a \ell^3 \sqrt{\frac{{\mathcal M}}{|z|^2}} \Big(
\KH^{-1} (\rmu a^3)^{\frac{5}{12}} + (K_{\ell} K_{L})^{-M}
+ K_{\ell}^{3/2} K_{L}^{3/2} ( K_{\ell}^{-2} \lille )^{\frac{M-1}{2}}
\Big)  - C \rmu^2 a \ell^3 \lille^4 \frac{a}{\ell}.
\end{align}

\end{lemma}

\begin{proof}[Proof of Lemma~\ref{lem:EstimatesOnQ3-js}]
The proofs of \eqref{eq:EstimateQ3Tilde1}, \eqref{eq:EstimateQ3Tilde2+3} and \eqref{eq:EstimateQ3Tilde4} are each rather lengthy and will be carried out individually.

\begin{proof}[Proof of \eqref{eq:EstimateQ3Tilde4}]
Notice, using Lemma~\ref{lem:DecaychiHat} applied to $\chi^2$ that
\begin{align}
\left\| \widehat{\chi_{\Lambda}^2}(s) \Big(1-f_L(s)\Big) \right\|_{\infty} \leq C_0 \ell^3 (1+(K_{\ell} K_L)^2)^{-M},
\end{align}
with $C_0 = \int \big|(1-\Delta)^M \chi^2\big|$.
Therefore, 
by a simple application of the Cauchy-Schwarz inequality, we get for any state $\Phi$ satisfying \eqref{eq:LocalizedAftercNumber}
\begin{align}\label{eq:CSon1a-terms}
\Big| \langle \Phi, \int \widehat{\chi_{\Lambda}^2}(s)  ( \widetilde{a}_s^{\dagger} + \widetilde{a}_s)\,ds \, \Phi \rangle \Big|
&\leq C \sqrt{{\mathcal M}}, 
\intertext{and}
\Big| \langle \Phi, \int \widehat{\chi_{\Lambda}^2}(s) \Big(1-f_L(s)\Big) ( \widetilde{a}_s^{\dagger} + \widetilde{a}_s)\,ds \, \Phi \rangle \Big|
&\leq C \sqrt{{\mathcal M}} (K_{\ell} K_L)^{-M}.
\end{align}
Therefore, using Lemma~\ref{lem:EstimateIntegral} below
 to estimate the $k$-integral, we find
\begin{align}
\Big| \langle \Phi, z \Big( \rho_z  \widehat{W_1 \omega}(0) (2\pi)^{-3} & \int \widehat{\chi_{\Lambda}^2}(s)  ( \widetilde{a}_s^{\dagger} + \widetilde{a}_s)\,ds\, \nonumber \\
&\quad - (2\pi)^{-6} \iint_{k \in P_H}  \widehat{W}_1(k) \alpha_{k} \widehat{\chi_{\Lambda}^2}(s)  f_L(s) ( \widetilde{a}_s^{\dagger} + \widetilde{a}_s)\,ds \Big) \Phi \rangle \Big| \nonumber \\
&
\leq
C \rho_z^2 a \ell^3 \sqrt{\frac{{\mathcal M}}{|z|^2}} \Big( \KH^{-1} (\rmu a^3)^{\frac{5}{12}}
+  (K_{\ell} K_L)^{-M} \Big).
\end{align}
The estimate is in agreement with the error term in \eqref{eq:EstimateQ3Tilde4}.

What remains in order to prove \eqref{eq:EstimateQ3Tilde4} is to estimate a difference of two integrals over the same domain. Writing out the commutator using \eqref{eq:bComms2} we have to estimate
\begin{align}\label{eq:1Q-integral}
z (2\pi)^{-6} \ell^3 \iint_{k \in P_H} \widehat{W}_1(k) \alpha_k 
\widehat{\chi^2}(s\ell) f_L(s) \left( 1- \frac{1-\alpha_{s-k} \alpha_{-k}}{(1-\alpha_k^2)(1-\alpha_{s-k}^2)}\right) ( \widetilde{a}_s^{\dagger} + \widetilde{a}_s),
\end{align}
and
\begin{align}\label{eq:1Q-integral2}
z (2\pi)^{-6} \ell^3 \iint_{k \in P_H} \widehat{W}_1(k) \alpha_k f_L(s) 
\frac{1-\alpha_{s-k} \alpha_{-k}}{(1-\alpha_k^2)(1-\alpha_{s-k}^2)}
\widehat{\chi}(k\ell) \widehat{\chi}\big((k-s)\ell\big)
 ( \widetilde{a}_s^{\dagger} + \widetilde{a}_s).
\end{align}
To estimate \eqref{eq:1Q-integral} we use \eqref{eq:alphaInPH}, \eqref{eq:CompIntegrals-Final} and Cauchy-Schwarz to get 
\begin{align}
\eqref{eq:1Q-integral}
&\leq
C \rho_z a \lille^2 \ell^3 \int \widehat{\chi^2}(s\ell)  ( \varepsilon^{-1} +\varepsilon \widetilde{a}_s^{\dagger} \widetilde{a}_s) 
\leq
C \rho_z a \lille^2  ( \varepsilon^{-1}  + \varepsilon n_{+}). 
\end{align}
We choose $\varepsilon^{-1} = D \rho_z a \ell^2 \lille^2$,
for some sufficiently large constant $D$ to allow the $n_{+}$ term to be absorbed in the kinetic energy gap. 
Thereby, the magnitude of the error (the $\varepsilon^{-1}$-term) becomes (using \eqref{eq:Close})
\begin{align}
C \rmu^2 a \ell^3 \lille^4 \frac{a}{\ell},
\end{align}
which can clearly be absorbed in the error term in \eqref{eq:EstimateQ3Tilde4}.

In the second integral \eqref{eq:1Q-integral2} the terms $\widehat{\chi}(k \ell)$ are very small due to regularity of $\chi$ and the fact that $k \in P_H$. Therefore this integral is much smaller.
We easily get, for arbitrary $\varepsilon >0$,
\begin{align}
&\langle \Phi, z (2\pi)^{-6} \ell^3 \iint_{k \in P_H} \widehat{W}_1(k) \alpha_k f_L(s) 
\frac{1-\alpha_{s-k} \alpha_{-k}}{(1-\alpha_k^2)(1-\alpha_{s-k}^2)}
\widehat{\chi}(k\ell) \widehat{\chi}\big((k-s)\ell\big)
 ( \widetilde{a}_s^{\dagger} + \widetilde{a}_s) \Phi \rangle \nonumber \\
& \geq
 - C z \rmu a  \sup_{k \in P_H} |\widehat{\chi}(k\ell)| \ell^3 \langle \Phi, \int f_L(s)  ( \varepsilon \widetilde{a}_s^{\dagger} \widetilde{a}_s + \varepsilon^{-1} ) \Phi \rangle \nonumber \\
 &\geq - C \rmu^2 a \ell^3 \sqrt{\frac{{\mathcal M}}{|z|^2}} K_{\ell}^{3/2} K_{L}^{3/2} ( K_{\ell}^{-2} \lille )^{\widetilde{M}}
\end{align}
where we optimized in $\varepsilon$ and used Lemma~\ref{lem:DecaychiHat} to get the last estimate. This error term is clearly in agreement with \eqref{eq:EstimateQ3Tilde4}.
This finishes the proof of \eqref{eq:EstimateQ3Tilde4}.
\end{proof}

In the proof of \eqref{eq:EstimateQ3Tilde4} we used the following result.
\begin{lemma}\label{lem:EstimateIntegral}
Suppose \eqref{eq:Close} and \eqref{con:KLKH}.
We also need the following weaker version of \eqref{con:KLKH},
\begin{align}\label{eq:NewCondition}
(\rmu a^3)^{-\frac{1}{12}} \geq \frac{\KH}{ds K_{\ell}} (1 + \log(K_H)).
\end{align}

Then for sufficiently small values of $\rmu$ we have,
\begin{align}\label{eq:CompIntegrals-Final}
\left|  \rho_z \widehat{W_1 \omega}(0) - (2\pi)^{-3} \int_{k \in P_H}  \widehat{W}_1(k) \alpha_{k} \,dk \right|
\leq C \rho_z a    (\rmu a^3)^{\frac{5}{12}}\KH^{-1}  .
\end{align}
Furthermore,
\begin{align}\label{eq:CompIntegrals-Final-2}
\left|    \widehat{W_1 \omega}(0) - (2\pi)^{-3} \int_{k \in P_H}  \frac{\widehat{W}_1(k)^2}{ 2{\mathcal D}_k} \,dk \right|
\leq C  a  (\rmu a^3)^{\frac{5}{12}}\KH^{-1}.
\end{align}
\end{lemma}

\begin{proof}
Collecting the estimates below, we really get
\begin{align}\label{eq:CompIntegrals}
&\left|  \rho_z \widehat{W_1 \omega}(0) - (2\pi)^{-3} \int_{k \in P_H}  \widehat{W}_1(k) \alpha_{k} \,dk \right| \nonumber \\
&\qquad \leq C \rho_z a \Big( K_H^{-1} + (\rho_z a^3) K_H+R^2/\ell^2 + (\rho_z a^3)^2 K_H^3
+ \frac{a}{d s \ell} (1+\log K_H)\Big).
\end{align}
From this \eqref{eq:CompIntegrals-Final} follows upon using \eqref{con:KLKH}, \eqref{eq:NewCondition} and \eqref{eq:AssumptionK_R} to compare the magnitudes of the different terms.

We calculate,
\begin{align}
&\rho_z \widehat{W_1 \omega}(0) - (2\pi)^{-3} \int_{k \in P_H}  \widehat{W_1}(k) \alpha_{k} \,dk \nonumber \\
&=
(2\pi)^{-3} \int_{k \in P_H}  \widehat{W_1}(k) \Big( \rho_z \frac{\widehat{g}(k)}{2k^2} - \alpha_{k}\Big) \,dk 
+
(2\pi)^{-3} \int_{k \notin P_H}  \rho_z \widehat{W_1}(k) \frac{\widehat{g}(k)}{2k^2}\,dk.
\end{align}
We first estimate the last integral, 
\begin{align}
\Big| \int_{k \notin P_H}  \widehat{W_1}(k) \frac{\widehat{g}(k)}{2k^2}\,dk \Big|
\leq 
C a^2 \int_{\{ |k| \leq  K_H^{-1} a^{-1}\}} k^{-2} \,dk 
=
C a K_H^{-1}.
\end{align}
This is consistent with the error term in \eqref{eq:CompIntegrals}.

To continue, we write
\begin{align}
\widehat{W_1}(k) \alpha_k = \rho_{z}^{-1} {\mathcal A}(k) \Big( 1 - \sqrt{1- {\mathcal B}(k)^2/{\mathcal A}(k)^2}\Big).
\end{align}

Notice that $|{\mathcal B}(k)/{\mathcal A}(k)| \leq \frac{1}{2}$, for $\rmu$ sufficiently small using \eqref{eq:BoverA} and \eqref{con:KLKH}.
Therefore,
\begin{align}
\Big| \widehat{W_1}(k) \alpha_k - \frac{ \rho_z \widehat{W_1}(k)^2}{2 {\mathcal A}(k)}\Big|
\leq C \rho_z^3 \frac{ \widehat{W_1}(k)^4}{{\mathcal A}(k)^3} \leq C \rho_z^3 a^4 k^{-6},
\end{align}
where we used that ${\mathcal A}(k) \geq \frac{1}{2} k^2$ in $P_H$.
Upon integrating over $P_H$ we find a term of magnitude
\begin{align}\label{eq:Comparealpha-A}
\int_{P_H} \Big| \widehat{W_1}(k) \alpha_k - \frac{ \rho_z \widehat{W_1}(k)^2}{2 {\mathcal A}(k)}\Big| \leq C \rho_z a (\rho_z a^3)^2 K_H^3,
\end{align}
in agreement with \eqref{eq:CompIntegrals}.

Finally, we estimate, using $0 \leq k^2 - \tau(k) \leq 2|k| (d s \ell)^{-1}$ in $P_H$,
\begin{align}
\rho_z &\Big| \int_{k \in P_H } \widehat{W_1}(k) \Big( \frac{\widehat{g}(k)}{2k^2} - \frac{\widehat{W}_1(k)}{2{\mathcal A}(k)} \Big) \Big|\nonumber \\
&\leq  \rho_z \Big| \int_{k \in P_H } \widehat{W_1}(k)  \frac{\widehat{g}(k)-\widehat{W}_1(k)}{2k^2}\Big|+ \rho_z \Big|\int_{k \in P_H } \frac{\widehat{W_1}(k)^2}{2k^2}  \Big( 1 - \frac{k^2}{{\mathcal A}(k)} \Big)\Big| \nonumber \\
&\leq  \rho_z \Big| \int_{k \in P_H } \widehat{W_1}(k)  \frac{\widehat{g}(k)-\widehat{W}_1(k)}{2k^2}\Big|+ C \rho_z^2 a^3 \int_{k \in P_H } k^{-4}  \nonumber \\
&\quad +
C\rho_z (d s \ell)^{-1} \Big(\int_{\{K_H^{-1} \leq a |k| \leq 1 \}} a^2 |k|^{-3}+ a \int \frac{\widehat{W_1}(k)^2}{2k^2} \Big) \nonumber \\
&\leq
\rho_z a \frac{R^2}{\ell^2} + \rho_z a (\rho_z a^3) K_H
C\rho_z a^2(d s \ell)^{-1} (1+ \log(K_H)),
\end{align}
where the estimate of the first term follows from Cauchy-Schwarz and \eqref{eq:I2-integral-2}.
This finishes the proof of \eqref{eq:CompIntegrals-Final}.

The proof of \eqref{eq:CompIntegrals-Final-2} is similar. One can for instance use \eqref{eq:CompIntegrals-Final} and \eqref{eq:Comparealpha-A} and the fact that
$|1 - \frac{{\mathcal A}(k)}{{\mathcal D}_k}| \leq C \frac{\mathcal B(k)^2}{{\mathcal A}(k)^2} \leq C \rmu^2 a^2 k^{-4}$ in $P_H$. Then \eqref{eq:CompIntegrals-Final-2} follows.
\end{proof}

\begin{proof}[Proof of \eqref{eq:EstimateQ3Tilde2+3}]
The two operators $ {\mathcal Q}_3^{(2)}(z)$ and ${\mathcal Q}_3^{(3)}(z)$ are very similar and can be estimated in identical fashion, so we will only explicity consider the first.
We decompose 
\begin{align}
{\mathcal Q}_3^{(2)}(z) = I + II,
\end{align}
where
\begin{align}\label{eq:OneAndTwo}
I&:= - z \ell^3 (2\pi)^{-6} \iint_{\{ k \in P_H\}}
\frac{f_L(s)
\widehat{W}_1(k)\alpha_k}{(1-\alpha_k^2)(1-\alpha_{s-k}^2)} \left( b_{-k}^{\dagger} \widetilde{a}_s^{\dagger}  b_{s-k}
+ b_{s-k}^{\dagger} \widetilde{a}_s b_{-k} \right) , \nonumber \\
II&:= - z \ell^3 (2\pi)^{-6} \iint_{\{ k \in P_H\}}
\frac{f_L(s)
\widehat{W}_1(k)\alpha_k}{(1-\alpha_k^2)(1-\alpha_{s-k}^2)} \left([\widetilde{a}_s^{\dagger} , b_{-k}^{\dagger}] b_{s-k}
+ b_{s-k}^{\dagger}  [b_{-k}, \widetilde{a}_s] \right).
\end{align}
The second term $II$ will be very small, due to the smallness of the commutator (notice that $s$ and $k$ are `far apart' since $s \in P_L$ and $k \in P_H$).
So the main term is $I$, which we estimate using Cauchy-Schwarz and \eqref{eq:alphaInPH} as
\begin{align}
I \geq
-C \ell^3 z a \lille  \iint_{k \in P_H} f_L(s) \left( \varepsilon b_{-k}^{\dagger} \widetilde{a}_s^{\dagger} \widetilde{a}_s b_{-k} + \varepsilon^{-1} b_{s-k}^{\dagger} b_{s-k} \right).
\end{align}
We estimate $\int \widetilde{a}_s^{\dagger} \widetilde{a}_s \leq \ell^{-3} {\mathcal M}$. Upon choosing $\varepsilon = \sqrt{K_{\ell}^3 K_{L}^3/{\mathcal M}}$ 
and using an easy bound on ${\mathcal D}_k$, 
this leads to the estimate
\begin{align}\label{eq:LoseABitOfb's}
\langle \Phi, I  \Phi \rangle & \geq - Cz a \lille 
K_{\ell}^{3/2} K_{L}^{3/2} {\mathcal M}^{1/2}
\langle \Phi, \int_{\{|k| \geq \frac{1}{2} K_H^{-1} a^{-1}\}} b_{k}^{\dagger} b_k \Phi \rangle \nonumber \\
&\geq - C \lille^2 \ell^3 \Big( \frac{K_{\ell}^3 K_{L}^3  {\mathcal M}}{\rmu \ell^3} \Big)^{1/2}   \langle \Phi, \int_{\{|k| \geq \frac{1}{2} K_H^{-1} a^{-1}\}} {\mathcal D}_k b_{k}^{\dagger} b_k \Phi \rangle 
\end{align}
Notice that
\begin{align}
\frac{K_{\ell}^3 K_{L}^3  {\mathcal M}}{\rmu \ell^3} = K_L^3 K_{{\mathcal M}} (\rmu a^3)^{\frac{1}{4}} \ll 1,
\end{align}
using \eqref{con:KellKM}. Therefore, $I$ can be absorbed in the $\lille^2 (2\pi)^{-3} \ell^3 \int_{\{|k| \geq \frac{1}{2} K_H^{-1} a^{-1}\}} {\mathcal D}_k b_{k}^{\dagger} b_k$ term in \eqref{eq:EstimateQ3Tilde2+3}.

We now return to the term $II$ from \eqref{eq:OneAndTwo}.
This is easily estimated as
\begin{align}
II &\geq - 2 z \ell^3 (2\pi)^{-6} 
\sup | [\widetilde{a}_s^{\dagger} , b_{-k}^{\dagger}] |
\iint_{\{ k \in P_H\}}
f_L(s)
|\widehat{W}_1(k)\alpha_k| \left( b_{s-k}^{\dagger} b_{s-k}
+ 1\right)\nonumber \\
&\geq
-C z \Big(\sup_{|p|\geq \frac{1}{2} K_H^{-1} a^{-1}} \widehat{\chi}(p\ell)\Big) 
(K_L K_{\ell})^3 
\left(\rmu a
+ a \lille \int_{|k| \geq \frac{1}{2} K_H^{-1} a^{-1}} b_k^{\dagger} b_k \right)
\end{align}
The $b_k^{\dagger} b_k$ is easily absorbed in the $\lille^2 \ell^3 \int_{\{|k| \geq \frac{1}{2} K_H^{-1} a^{-1}\}} {\mathcal D}_k b_{k}^{\dagger} b_k$ term in \eqref{eq:EstimateQ3Tilde2+3}.
Therefore, using \eqref{eq:Close} and Lemma~\ref{lem:DecaychiHat}, $II$ contributes with an error term of order
\begin{align}
\rmu^2 a \ell^3 
( K_{\ell}^{-2} \lille )^{\widetilde{M}} K_L^3 K_{\ell}^{3/2} (\rmu a^3)^{\frac{1}{4}}
\end{align}
to \eqref{eq:EstimateQ3Tilde2+3}.

This finishes the proof of \eqref{eq:EstimateQ3Tilde2+3}.
\end{proof}

\begin{proof}[Proof of \eqref{eq:EstimateQ3Tilde1} ]
Finally, we estimate ${\mathcal Q}_3^{(1)}(z)$. We rewrite
\begin{align}
{\mathcal Q}_3^{(1)}(z) =
 z \ell^3 (2\pi)^{-6} \iint_{\{ k \in P_H\}}
\frac{f_L(s)
\widehat{W}_1(k)}{(1-\alpha_k^2)(1-\alpha_{s-k}^2)} \left(\widetilde{a}_s^{\dagger}  
b_{s-k} b_k + \alpha_k \alpha_{s-k} \widetilde{a}_{-s}^{\dagger}   b_{s-k}^{\dagger} b_{k}^{\dagger} + h.c.
 \right),
\end{align}
where we performed a change of variables in the second term to get the equality.

We combine this term with the diagonalized Bogolubov Hamiltonian.
We leave a $\lille^2$-part of this operator in order to control error terms appearing below.

Therefore, we consider
\begin{align}\label{eq:CentralEstimateOfQ31}
&(2\pi)^{-3} \ell^3 \int_{\{ k \in P_H\}} (1- 2\lille^2) {\mathcal D}_k b_k^{\dagger} b_{k}\,dk \nonumber \\
&\quad +
z \ell^3 (2\pi)^{-6} \iint_{\{ k \in P_H\}}
\frac{f_L(s)
\widehat{W}_1(k)}{(1-\alpha_k^2)(1-\alpha_{s-k}^2)} \left(\widetilde{a}_s^{\dagger}  
b_{s-k} b_k + \alpha_k \alpha_{s-k} \widetilde{a}_{-s}^{\dagger}   b_{s-k}^{\dagger} b_{k}^{\dagger} + h.c.
 \right) \nonumber \\
 &=
 (2\pi)^{-3} \ell^3 \int_{\{ k \in P_H\}} (1-2\lille^2){\mathcal D}_k c_k^{\dagger} c_k +
 T_1(k) +T_2(k) \nonumber \\
 &\geq (2\pi)^{-3} \ell^3 \int_{\{ k \in P_H\}} 
 T_1(k) +T_2(k).
 \end{align}
Here we have introduced the operators,
\begin{align}
c_k &:= b_k + z (2\pi)^{-3} \int \frac{f_L(s)
\widehat{W}_1(k)}{(1-2\lille^2){\mathcal D}_k(1-\alpha_k^2)(1-\alpha_{s-k}^2)} 
\left( b_{s-k}^{\dagger} \widetilde{a}_s + \alpha_k \alpha_{s-k} \widetilde{a}_{-s} b_{s-k}^{\dagger}\right)ds,\\
\label{eq:T1_def}
T_1(k)  &:= - z (2\pi)^{-3} \int \frac{f_L(s)
\widehat{W}_1(k) \alpha_k \alpha_{s-k} }{(1-\alpha_k^2)(1-\alpha_{s-k}^2)} \left( [b_k^{\dagger}, \widetilde{a}_{-s} b_{s-k}^{\dagger}] + h.c. \right)ds,
\intertext{and}
T_2(k) &:= -\frac{|z|^2 \widehat{W}_1(k)^2}{(1-2\lille^2){\mathcal D}_k(1-\alpha_k^2)^2} (2\pi)^{-6} \iint \frac{f_L(s)f_L(s')
}{(1-\alpha_{s-k}^2)(1-\alpha_{s'-k}^2)}  \nonumber \\
&\quad \times 
\left(\widetilde{a}_{s'}^{\dagger} b_{s'-k}      + \alpha_k \alpha_{s'-k} b_{s'-k} \widetilde{a}_{-s'}^{\dagger} \right)
\left( b_{s-k}^{\dagger} \widetilde{a}_s + \alpha_k \alpha_{s-k} \widetilde{a}_{-s} b_{s-k}^{\dagger}\right)\,ds\,ds'\nonumber \\
&\geq
-\left(1+ C \lille^2\right) |z|^2\frac{ \widehat{W}_1(k)^2}{{\mathcal D}_k} (2\pi)^{-6} \iint f_L(s)f_L(s')
 \nonumber \\
&\quad\times 
\left(\widetilde{a}_{s'}^{\dagger} b_{s'-k}      + \alpha_k \alpha_{s'-k} b_{s'-k} \widetilde{a}_{-s'}^{\dagger} \right)
\left( b_{s-k}^{\dagger} \widetilde{a}_s + \alpha_k \alpha_{s-k} \widetilde{a}_{-s} b_{s-k}^{\dagger}\right)\,ds\,ds',
\end{align}
where we used \eqref{eq:alphaInPH} to get the estimate on $T_2$.
Notice that
\begin{align}
\widetilde{a}_{s'}^{\dagger} b_{s'-k}      + \alpha_k \alpha_{s'-k} b_{s'-k} \widetilde{a}_{-s'}^{\dagger}
=
\left( \widetilde{a}_{s'}^{\dagger}      + \alpha_k \alpha_{s'-k} \widetilde{a}_{-s'}^{\dagger}\right) b_{s'-k} + \alpha_k \alpha_{s'-k} [b_{s'-k}, \widetilde{a}_{-s'}^{\dagger}].
\end{align}
The contribution from the commutator term is very small, both due to the factors of $\alpha$ and to the commutator,
since $k \in P_H$, $s'\in P_L$. 
Therefore, we estimate
\begin{align}
T_2(k) \geq (1+\varepsilon) T_2'(k) + (1+ \varepsilon^{-1}) T_2''(k),
\end{align}
where
\begin{align}
T_2'(k) &:= - \left(1+ C \lille^2\right) |z|^2 \frac{ \widehat{W}_1(k)^2}{{\mathcal D}_k} (2\pi)^{-6} \iint f_L(s)f_L(s')
 \nonumber \\
&\qquad \times 
\left( \widetilde{a}_{s'}^{\dagger}      + \alpha_k \alpha_{s'-k} \widetilde{a}_{-s'}^{\dagger}\right) b_{s'-k} b_{s-k}^{\dagger} 
\left( \widetilde{a}_s + \alpha_k \alpha_{s-k} \widetilde{a}_{-s} \right)\,ds\,ds' \nonumber \\
T_2''(k) &:=- \left(1+ C \lille^2\right) |z|^2 \frac{ \widehat{W}_1(k)^2}{{\mathcal D}_k} (2\pi)^{-6} \nonumber \\
&\qquad \times \iint f_L(s)f_L(s') |\alpha_k|^2 \alpha_{s'-k} \alpha_{s-k}
[b_{s'-k}, \widetilde{a}_{-s'}^{\dagger}] [ \widetilde{a}_{-s}, b_{s-k}^{\dagger}] 
\end{align}
For simplicity, we choose $\varepsilon = \lille^2$ and can therefore absorb the factor of $(1+\varepsilon)$ in $T_2'(k)$ by simply changing the value of $C$.
With this choice, we estimate using \eqref{eq:CompIntegrals-Final-2}, \eqref{con:KLKH} and \eqref{eq:alphaInPH},
\begin{align}
(2\pi)^{-3} \ell^3 \int_{k \in P_H} (1+ \varepsilon^{-1}) T_2''(k) \,dk
&\geq 
-C \rho_z a (K_{\ell} K_L)^6 \lille^2 \sup_{k\in P_H, s \in P_L}
| [ \widetilde{a}_{-s}, b_{s-k}^{\dagger}] |^2 \nonumber \\
&\geq
- C \rmu^2 a \ell^3 (\rmu a^3)^{\frac{1}{2}} K_{\ell}^3 K_L^6 \lille^2 \sup_{k\in P_H, s \in P_L}
| [ \widetilde{a}_{-s}, b_{s-k}^{\dagger}] |^2.
\end{align}
We continue to estimate the other part of $T_2(k)$.
\begin{align}\label{eq:T2prime}
T_2'(k) &:= - \left(1+ C \lille^2 \right) |z|^2 \frac{ \widehat{W}_1(k)^2}{{\mathcal D}_k} (2\pi)^{-6} \iint f_L(s)f_L(s')
 \nonumber \\
&\qquad \times 
\left( \widetilde{a}_{s'}^{\dagger}      + \alpha_k \alpha_{s'-k} \widetilde{a}_{-s'}^{\dagger}\right) b_{s'-k} b_{s-k}^{\dagger} 
\left( \widetilde{a}_s + \alpha_k \alpha_{s-k} \widetilde{a}_{-s} \right)\,ds\,ds' 
\nonumber \\
&=T_{2,{\rm comm}}'(k) + T_{2,{\rm op}}'(k),
\end{align}
with
\begin{align}
T_{2,{\rm comm}}'(k)&:=- \left(1+ C \lille^2 \right) |z|^2\frac{ \widehat{W}_1(k)^2}{{\mathcal D}_k} (2\pi)^{-6} \iint f_L(s)f_L(s')
 \nonumber \\
&\qquad \times 
\left( \widetilde{a}_{s'}^{\dagger}      + \alpha_k \alpha_{s'-k} \widetilde{a}_{-s'}^{\dagger}\right) [b_{s'-k} ,b_{s-k}^{\dagger} ]
\left( \widetilde{a}_s + \alpha_k \alpha_{s-k} \widetilde{a}_{-s} \right)\,ds\,ds' ,\nonumber \\
T_{2,{\rm op}}'(k)&:=-\left(1+ C \lille^2 \right) |z|^2
\frac{ \widehat{W}_1(k)^2}{{\mathcal D}_k} (2\pi)^{-6} \iint f_L(s)f_L(s')
 \nonumber \\
&\qquad \times 
\left( \widetilde{a}_{s'}^{\dagger}      + \alpha_k \alpha_{s'-k} \widetilde{a}_{-s'}^{\dagger}\right) b_{s-k}^{\dagger} b_{s'-k}
\left( \widetilde{a}_s + \alpha_k \alpha_{s-k} \widetilde{a}_{-s} \right)\,ds\,ds'.
\end{align}
We start by estimating the last term in \eqref{eq:T2prime}.
We introduce the notation
\begin{align}
{\mathcal C} := \sup_{s,s' \in P_L, k \in P_H } \left| [ \widetilde{a}_{s'}^{\dagger}      + \alpha_k \alpha_{s'-k} \widetilde{a}_{-s'}^{\dagger}, b_{s-k}^{\dagger}]\,\right| \leq 1,
\end{align}
In fact,
it follows from \eqref{eq:Commutator_general}, \eqref{eq:Defb_k}, \eqref{eq:alphaInPH}, and \eqref{eq:DecayPH} that
\begin{align}\label{eq:StoerrelseC}
{\mathcal C} \leq C \delta \left( K_{\ell}^{-2} \KH^2 (\rmu a^3)^{\frac{1}{6}}\right)^{\frac{M-1}{2}}.
\end{align}

To estimate the last term in \eqref{eq:T2prime} we first apply Cauchy-Schwarz, then commute the $\widetilde{a}$'s through the $b$'s and apply Cauchy-Schwarz to the commutator terms.
This yields,
\begin{align}
&\left\langle \Phi, \iint f_L(s)f_L(s')  \left( \widetilde{a}_{s'}^{\dagger}      + \alpha_k \alpha_{s'-k} \widetilde{a}_{-s'}^{\dagger}\right) b_{s-k}^{\dagger} b_{s'-k}
\left( \widetilde{a}_s + \alpha_k \alpha_{s-k} \widetilde{a}_{-s} \right) \Phi \right\rangle\nonumber \\
&\leq 
2 \iint f_L(s)f_L(s') \left\langle \Phi, b_{s-k}^{\dagger}
\left( \widetilde{a}_{s'}^{\dagger}      + \alpha_k \alpha_{s'-k} \widetilde{a}_{-s'}^{\dagger}\right) 
\left( \widetilde{a}_{s' }+ \alpha_k \alpha_{s'-k} \widetilde{a}_{-s'} \right)  b_{s-k} \Phi \right\rangle\nonumber \\
&\quad 
+ C {\mathcal C}   \iint f_L(s)f_L(s') \left\langle \Phi, \left( \varepsilon b_{s-k}^{\dagger}b_{s-k} + C \varepsilon^{-1} \widetilde{a}_{s' } \widetilde{a}_{s' }^{\dagger} + {\mathcal C}\right)\Phi \right\rangle \nonumber \\
&\leq
C (\ell^{-3} {\mathcal M}  + \varepsilon |P_L| {\mathcal C} ) \int f_L(s)  \langle \Phi, b_{s-k}^{\dagger}b_{s-k} \Phi \rangle \nonumber \\
&\quad + C \varepsilon^{-1} |P_L| {\mathcal C} (\ell^{-3} {\mathcal M} +  |P_L| )
+ C |P_L|^2 {\mathcal C}^2.
\end{align}
For simplicity, we choose $\varepsilon = \frac{{\mathcal M}}{\ell^3 |P_L| {\mathcal C}}$ and get
\begin{align}
&\left\langle \Phi, \iint f_L(s)f_L(s') \left( \widetilde{a}_{s'}^{\dagger}      + \alpha_k \alpha_{s'-k} \widetilde{a}_{-s'}^{\dagger}\right) b_{s-k}^{\dagger} b_{s'-k}
\left( \widetilde{a}_s + \alpha_k \alpha_{s-k} \widetilde{a}_{-s} \right) \Phi \right\rangle\nonumber \\
&\leq 
C \ell^{-3} {\mathcal M}  \left\langle \Phi,  \int f_L(s)  b_{s-k}^{\dagger}b_{s-k} \Phi \right \rangle+ C  |P_L|^2 {\mathcal C}^2 \Big(1 + \frac{\ell^3 |P_L|}{{\mathcal M}}\Big).
\end{align}
Therefore, using \eqref{eq:CompIntegrals-Final-2},
\begin{align}\label{eq:T2opprime}
\left\langle \Phi, (2\pi)^{-3} \ell^3 \int_{k \in P_H} T_{2,{\rm op}}'(k) \Phi \right\rangle
&\geq
- C \rmu \ell^3 \frac{a^2}{(\min_{k \in \frac{1}{2} P_H} {\mathcal D}_k)^2} {\mathcal M} |P_L| \langle \Phi \int_{\{q \in \frac{1}{2} P_H\}} {\mathcal D}_q b_q^{\dagger} b_q \Phi \rangle \nonumber \\
&\quad - \rmu \ell^6 a |P_L|^2 {\mathcal C}^2 \left( 1 + \frac{\ell^3 |P_L|}{{\mathcal M}}\right).
\end{align}
Notice that ${\mathcal D_k} \geq C^{-1} \KH^{-2} (\rmu a^3)^{\frac{5}{6}} a^{-2}$, for $k \in \frac{1}{2} P_H$.
Therefore, using \eqref{eq:DefLille} and \eqref{con:KellKM},
\begin{align}
\rmu \frac{a^2}{(\min_{k \in \frac{1}{2} P_H} {\mathcal D}_k)^2} {\mathcal M} |P_L| 
\leq \lille^2 K_L^3 K_{{\mathcal M}} (\rmu a^3)^{\frac{1}{4}} \ll \lille^2.
\end{align}
Therefore, the negative ${\mathcal D}_q b_q^{\dagger} b_q$-term in \eqref{eq:T2opprime} can be absorbed in a fraction of the similar (positive) term left out in \eqref{eq:CentralEstimateOfQ31} exactly for this purpose.

Notice that $\ell^3 |P_L| \leq  C (K_L K_{\ell})^3 = d^{-6}$ using \eqref{eq:KL_d}. Therefore, it follows from \eqref{con:rough1} that $\frac{\ell^3 |P_L|}{{\mathcal M}} \ll 1$.
So, using \eqref{eq:StoerrelseC} we can estimate the error term in \eqref{eq:T2opprime} as
\begin{align}
- \rmu \ell^6 a |P_L|^2 {\mathcal C}^2 \left( 1 + \frac{\ell^3 |P_L|}{{\mathcal M}}\right)
&\geq - C \rmu^2 a \ell^3 \sqrt{\rmu a^3} \left( K_{\ell}^{-3} d^{-12}\delta^2 
\left( K_{\ell}^{-2} \KH^2 (\rmu a^3)^{\frac{1}{6}}\right)^{M-1}\right),
\end{align}
This is clearly seen to agree with \eqref{eq:EstimateQ3Tilde1}.

We next consider the commutator term $T_{2,{\rm comm}}'(k)$ from \eqref{eq:T2prime}.

From \eqref{eq:bComms2} and using Lemma~\ref{lem:DecaychiHat}, we see that
\begin{align}\label{eq:bcommutator}
\left|  [b_{s'-k} ,b_{s-k}^{\dagger} ]  - \widehat{\chi^2}((s-s')\ell) \right| \leq 
C\lille^2 | \widehat{\chi^2}((s-s')\ell)| 
+ C ( K_{\ell}^{-2} \lille )^{\frac{M-1}{2}}.
\end{align}

Therefore, and using that $M \geq 5$,
\begin{align}\label{eq:EstimT2comm}
T_{2,{\rm comm}}'(k) &\geq
\big(1+C \lille^2 \big) |z|^2 \frac{ \widehat{W}_1(k)^2}{{\mathcal D}_k} (2\pi)^{-6} \iint f_L(s)f_L(s')
\widetilde{a}_{s'}^{\dagger} \widehat{\chi^2}((s-s')\ell) \widetilde{a}_s \nonumber \\
&\quad -
C  |z|^2
\frac{ \widehat{W}_1(k)^2}{{\mathcal D}_k}  \lille^2 |P_L| \ell^{-3} n_{+}.
\end{align}

Using \eqref{eq:RepInTermsOfatilde} and \eqref{eq:CompIntegrals-Final-2} we see that 
\begin{align}
-(2\pi)^{-3} &\ell^3 \int_{k \in P_H} \big(1+C \lille^2 \big) |z|^2 \frac{ \widehat{W}_1(k)^2}{{\mathcal D}_k} (2\pi)^{-6} \iint f_L(s)f_L(s')
\widetilde{a}_{s'}^{\dagger} \widehat{\chi^2}((s-s')\ell) \widetilde{a}_s \nonumber \\
&=
-\rho_z \big(1+C \lille^2 \big) \big((2\pi)^{-3} \int_{k \in P_H} \frac{ \widehat{W}_1(k)^2} {{\mathcal D}_k} \big)\sum_j Q_{L,j} \chi^2 Q_{L,j} \nonumber \\
&\geq
- 2 \rho_z \big(1+C \lille^2  \big) \widehat{W_1 \omega}(0) \sum_j Q_{L,j} \chi^2 Q_{L,j}.
\end{align}
Here we used \eqref{con:KLKH} to control the error from \eqref{eq:CompIntegrals-Final-2}.

We now notice that, for all $\varepsilon>0$,
\begin{align}
\sum_j Q_{L,j} \chi^2 Q_{L,j} \leq (1+\varepsilon) \sum_j Q_{L,j} \chi^2 Q_{L,j} +
C \varepsilon^{-1} n_{+}^H.
\end{align}
We notice that $\rmu a = (d K_{\ell})^2 \frac{1} {d^2 \ell^2}$. Therefore, choosing $\varepsilon$ proportional to $\varepsilon_T^{-1} (d K_{\ell})^2$, we find, using \eqref{eq:CompIntegrals-Final-2},
\begin{align}
-(2\pi)^{-3} &\ell^3 \int_{k \in P_H} \big(1+C \lille^2 \big) |z|^2 \frac{ \widehat{W}_1(k)^2}{{\mathcal D}_k} (2\pi)^{-6} \iint f_L(s)f_L(s')
\widetilde{a}_{s'}^{\dagger} \widehat{\chi^2}((s-s')\ell) \widetilde{a}_s \nonumber \\
&\geq
- 2 \rho_z \big(1+C \lille^2 + C \varepsilon_T^{-1} (dK_{\ell})^2 \big) \widehat{W_1 \omega}(0) \sum_j Q_j \chi^2 Q_j 
- \frac{1}{100} \frac{1}{(d\ell)^2} n_{+}^H.
\end{align}
Notice now, using \eqref{con:eTdK}, \eqref{con:KellKH}, \eqref{eq:DefLille} and \eqref{eq:Close}, that
\begin{align}
\rho_z a [  \lille^2 +\varepsilon_T^{-1} (dK_{\ell})^2]
\ll \ell^{-2}.
\end{align}
Therefore, the above error terms can be absorbed in the energy gap.

To estimate the error term in \eqref{eq:EstimT2comm} we integrate
\begin{align}
- (2\pi)^{-3} \ell^3 \int_{k \in P_H} C  |z|^2
\frac{ \widehat{W}_1(k)^2}{{\mathcal D}_k}  \lille^2 |P_L| \ell^{-3} n_{+} &\geq
-C \rmu a \lille^2 (\ell^3 |P_L|) n_{+}
\end{align}
Notice that by \eqref{con:KellKH} and \eqref{eq:KL_d} $\rmu a \lille^2 (\ell^3 |P_L|) \ll \ell^{-2}$, so this term can also be absorbed in the energy gap.

We now estimate the other commutator term, namely $T_1(k)$ from \eqref{eq:T1_def}.
We clearly have
\begin{align}
T_1(k)
&\geq - C z \lille \sup_{k\in P_H, s \in P_L } \left(|[b_k^{\dagger}, b_{s-k}^{\dagger}]| \right)
|\alpha_k \widehat{W}_1(k)| \int f_L(s) \left( \widetilde{a}_{-s}^{\dagger} \widetilde{a}_{-s} + 1\right)\,ds \nonumber \\
&\quad - C z \lille \sup_{k\in P_H, s \in P_L } \left(|[b_k^{\dagger}, \widetilde{a}_{-s}^{\dagger}]|\right) |\alpha_k \widehat{W}_1(k)| \int f_L(s)
\left( b_{s-k}^{\dagger} b_{s-k} + 1\right) \,ds.
\end{align}
Therefore,
\begin{align}\label{eq:EndEstimateT1}
&\ell^3 (2\pi)^{-3}  \int_{k \in P_H} T_1(k)\,dk \\
&\geq
- C z \lille \Big(\sup_{k\in P_H, s \in P_L } |[b_k^{\dagger}, b_{s-k}^{\dagger}]| \Big) \rho a (n_{+} + (K_{\ell} K_L)^3) \nonumber \\
&\quad  - C z \lille \Big(\sup_{k\in P_H, s \in P_L } |[b_k^{\dagger}, \widetilde{a}_{-s}^{\dagger}]|\Big)  \rho a  (K_{\ell} K_L)^3 \nonumber \\
&\quad -Cz \lille^2 \Big(\sup_{k\in P_H, s \in P_L } |[b_k^{\dagger}, \widetilde{a}_{-s}^{\dagger}]|\Big) \Big(\sup_{k\in P_H, s \in P_L } |D_{s-k}|^{-1} \Big)
 a 
 (K_{\ell} K_L)^3 \int_{\{ |k| \geq \frac{1}{2} K_H^{-1} a^{-1}\}} D_k b_k^{\dagger} b_k.
\end{align}
The last term in this inequality is easily seen to be estimated as
\begin{align}
\geq - \lille^2 \left\{ \lille \frac{1}{\sqrt{\rmu \ell^3}} K_{\ell}^3 K_L^3 \Big(\sup_{k\in P_H, s \in P_L } |[b_k^{\dagger}, \widetilde{a}_{-s}^{\dagger}]|\Big) \right\} \ell^3 \int_{\{ |k| \geq \frac{1}{2} K_H^{-1} a^{-1}\}} D_k b_k^{\dagger} b_k,
\end{align}
and using the properties of the commutator and Lemma~\ref{lem:DecaychiHat}, we see that this term can easily be absorbed in the extra $ \lille^2 \ell^3 \int_{\{ |k| \geq \frac{1}{2} K_H^{-1} a^{-1}\}} D_k b_k^{\dagger} b_k$ omitted in \eqref{eq:CentralEstimateOfQ31}.

The two remaining terms in \eqref{eq:EndEstimateT1} can be estimated (using, in particular, Lemma~\ref{lem:DecaychiHat} and \eqref{eq:alphaInPH}) as
\begin{align}
\geq - C \rmu^2 a \ell^3 \lille K_{\ell}^{-3/2} (\rmu a^3)^{\frac{1}{4}} ({\mathcal M} + (K_{\ell}^3 K_L^3)) 
\left( K_{\ell}^{-2} \lille \right)^{\frac{M-1}{2}}.
\end{align}

This finishes the proof of \eqref{eq:EstimateQ3Tilde1}
\end{proof}

Now we have established all three inequalities \eqref{eq:EstimateQ3Tilde1}, \eqref{eq:EstimateQ3Tilde2+3} and \eqref{eq:EstimateQ3Tilde4}.
This finishes the proof of Lemma~\ref{lem:EstimatesOnQ3-js}.
\end{proof}

\section{Proof of the main theorem}
In this section we will combine the results of the previous sections in order to prove Theorem~\ref{thm:LHY}.

\begin{proof}[Proof of Theorem~\ref{thm:LHY}]
As noted in Section~\ref{sec:Fock}, Theorem~\ref{thm:LHY} follows from 
Theorem~\ref{thm:LHY-Background}, which again---as observed in Section~\ref{subsec:LocHam}---follows from Theorem~\ref{thm:LHY-Box}.
We will use the concrete choice of parameters set down in \eqref{eq:Xparameter} and \eqref{eq:Xchoice} in Section~\ref{sec:params}.
Recall in particular the notation $X$ defined in \eqref{eq:Xchoice}.

To prove Theorem~\ref{thm:LHY-Box} let $\Psi \in {\mathcal
  F}_s(L^2(\Lambda))$ be a normalized $n$-particle trial state satisfying \eqref{eq:aprioriPsi}. Notice that if such a state does not exist, then there is nothing to prove. Using Lemma~\ref{lem:LocMatrices} there exists a normalized $n$-particle wave function $\widetilde{\Psi}\in {\mathcal F}_{\rm s}(L^2(\Lambda))$ satisfying \eqref{eq:LocalizedNPlus} and such that
\begin{align}\label{eq:LargeMatricesNew}
\langle\Psi,{\mathcal H}_\Lambda(\rmu)\Psi\rangle \geq
\langle\widetilde\Psi,{\mathcal H}_\Lambda(\rmu)\widetilde\Psi\rangle
- C X^2 {\mathcal R} \rmu^2 a \ell^3 (\rmu a^3)^{1/2}.
\end{align}
Notice that the error term in \eqref{eq:LargeMatricesNew} is consistent with the error term in Theorem~\ref{thm:LHY-Box}.

Using Proposition~\ref{prop:Hamilton2ndQuant} we find that our localized state $\widetilde\Psi$ satisfies 
\begin{align}
\langle \widetilde{\Psi}, {\mathcal H}_{\Lambda}(\rmu) \widetilde{\Psi} \rangle 
&\geq \langle \widetilde{\Psi}, {\mathcal H}_{\Lambda}^{\rm 2nd}(\rmu) \widetilde{\Psi} \rangle \nonumber \\
&\quad - C \rmu^2 a \ell^3  (\rmu a^3)^{1/2} \left( (\rmu a^3)^{\frac{348}{323}- \frac{1}{2}}+ X^3 (R a^{-1})^2 (\rmu a^3)^{1/2}  \right) , 
\end{align}
where the error is clearly consistent with the error term in Theorem~\ref{thm:LHY-Box}.

At this point, we can apply
Theorem~\ref{thm:Kafz} to get the lower bound
\begin{align}\label{eq:IntroK(z)-}
\langle \widetilde{\Psi}, {\mathcal H}_{\Lambda}^{\rm 2nd}(\rmu) \widetilde{\Psi} \rangle
\geq
\inf_{z \in {\mathbb R}_{+} } \inf_{\Phi} \langle \Phi, {\mathcal K(z)} \Phi \rangle
- C \rmu a,
\end{align}
where the second infimum is over all normalized $\Phi  \in {\mathcal F}(\Ran(Q))$ satisfying \eqref{eq:LocalizedAftercNumber}.

Since 
\begin{align}
\rmu a = \rmu^2 a \ell^3 \sqrt{\rmu a^3} K_{\ell}^{-3} = \rmu^2 a \ell^3 \sqrt{\rmu a^3} X^{\frac{9}{2}},
\end{align}
which is in agreement with the error term in Theorem~\ref{thm:LHY-Box},
this implies that we need to prove that
\begin{align}\label{eq:ToProve}
\inf_{\Phi} \langle \Phi, {\mathcal K(z)} \Phi \rangle &\geq
-4\pi \rmu^2 a \ell^3 +  4\pi \rmu^2 a  \ell^3 \frac{128}{15\sqrt{\pi}}(\rmu a^3)^{\frac{1}{2}} 
\nonumber \\
&\quad 
 - C \rmu^2 a \ell^3 (\rmu a^3)^{\frac{1}{2}} \left( \frac{R^2}{a^2} (\rmu a^3)^{\frac{1}{2}} +
X^{\frac{1}{5}} 
\right),
\end{align}
for all normalized $\Phi$ satisfying \eqref{eq:LocalizedAftercNumber}.

We will use that with our choice of parameters \eqref{eq:NewConditionNew} is satisfied.

If $\rho_z = |z|^2/\ell^3$ satisfies \eqref{eq:Close-I}, i.e. is `far away' from $\rmu$, then Proposition~\ref{prop:FarAway} provides a lower bound on $\langle \Phi, {\mathcal K(z)} \Phi \rangle$ which is larger than needed for \eqref{eq:ToProve} by a factor of $2$ on the LHY-term. Since \eqref{eq:NewConditionNew} is satisfied the assumptions of Proposition~\ref{prop:FarAway} are verified.

If $\rho_z$ satisfies the complementary inequality \eqref{eq:Close} and $\Phi$ satisfies \eqref{eq:LocalizedAftercNumber}, then by \eqref{eq:SimplifyKz} (using again that \eqref{eq:NewConditionNew} is satisfied) and Theorem~\ref{thm:BogHamDiag} combined with Lemma~\ref{lem:BogIntegral} we get
\begin{align}\label{eq:Combined1}
\langle \Phi, {\mathcal K(z)} \Phi \rangle 
&\geq - \frac{1}{2} \rmu^2 \ell^{3} \widehat{g}(0)
+ 4\pi \frac{128}{15\sqrt{\pi}} \rho_z a \sqrt{\rho_z a^3}\ell^3 \nonumber \\
&\quad +  \langle \Phi, \left(\frac{b}{4 \ell^2} n_{+} + \varepsilon_T\frac{b}{2 d^2 \ell^2} n_{+}^{H} 
  + {\mathcal Q}_1^{\rm ex}(z)+
  {\mathcal Q}_2^{\rm ex}(z) + {\mathcal Q}_3(z) \right) \Phi \rangle \nonumber \\
  &\quad + (2\pi)^{-3} \ell^3 \langle \Phi, \int {\mathcal D}_k b_k^{\dagger} b_{k}\,dk \, \Phi \rangle - {\mathcal E}_1,
\end{align}
where the error term ${\mathcal E}_1$ satisfies
\begin{align}
{\mathcal E}_1 &\leq C \rmu^2 a \ell^3 (\rmu a^3)^{\frac{1}{2}} \left( \frac{R^2}{a^2} (\rmu a^3)^{\frac{1}{2}} +
X^{\frac{1}{5}} + (\rmu a^3)^{\frac{1}{4}} (R a^{-1})^{\frac{1}{2}}
\right).
\end{align}
Here the error term in $X^{\frac{1}{5}}$ comes from the $\varepsilon(\rmu, \rho_z)$ in Lemma~\ref{lem:BogIntegral}. Notice that this error is compatible with \eqref{eq:ToProve} using Young's inequality.

Now we can apply Theorem~\ref{thm:Control3Q} to obtain the inequality
\begin{align}\label{eq:Combined2}
&(2\pi)^{-3} \ell^3 \langle \Phi, \int {\mathcal D}_k b_k^{\dagger} b_{k}\,dk \, \Phi \rangle \nonumber \\
&+  \langle \Phi, \left(\frac{b}{4 \ell^2} n_{+} + \varepsilon_T\frac{b}{2 d^2 \ell^2} n_{+}^{H} 
  + {\mathcal Q}_1^{\rm ex}(z)+
  {\mathcal Q}_2^{\rm ex}(z) + {\mathcal Q}_3(z) \right) \Phi \rangle 
  \geq - {\mathcal E}_2,
\end{align}
with error term
\begin{align}
{\mathcal E}_2 
&\leq C \rmu^2 a \ell^3 \sqrt{\rmu a^3} (\rmu a^3)^{\frac{1}{6}} X^{-\frac{11}{12}}.
\end{align}
Here the dominant contribution to the error (with our choice of parameters) comes from 
the $\KH^{-1} (\rmu a^3)^{\frac{5}{12}}$-term. This error is clearly consistent with \eqref{eq:ToProve}.

Combining \eqref{eq:Combined1} and \eqref{eq:Combined2}, we get
\begin{align}
\langle \Phi, {\mathcal K(z)} \Phi \rangle 
&\geq - \frac{1}{2} \rmu^2 \ell^{3} \widehat{g}(0)
+ 4\pi \frac{128}{15\sqrt{\pi}} \rmu a \sqrt{\rmu a^3}\ell^3 \nonumber \\
&\quad - ( {\mathcal E}_1 + {\mathcal E}_2 + C \left|\rmu a \sqrt{\rmu a^3} - \rho_z a \sqrt{\rho_z a^3}\right|\ell^3).
\end{align}
This establishes \eqref{eq:ToProve} for $\rho_z$ satisfying \eqref{eq:Close}, since by \eqref{eq:Close}, \eqref{eq:NewConditionNew} and \eqref{eq:Xparameter} we have
\begin{align}
\left|\rmu a \sqrt{\rmu a^3} - \rho_z a \sqrt{\rho_z a^3}\right|\ell^3 \leq C \rmu a \sqrt{\rmu a^3} \ell^3 K_{\ell}^{-2} = C \rmu a \sqrt{\rmu a^3} \ell^3 X^3.
\end{align}
This finishes the proof of \eqref{eq:ToProve} and therefore of Theorem~\ref{thm:LHY-Box}, which in turn implies Theorem~\ref{thm:LHY-Background} and Theorem~\ref{thm:LHY}.
\end{proof}

\appendix

\section{Bogolubov method}\label{sec:simplebog}
In this section we recall a simple consequence of the Bogolubov method (see \cite[Theorem~6.3]{LS} and \cite{BS})

\begin{theorem}[Simple case of Bogolubov's method]\label{thm:bogolubov-complete}\hfill\\
Let $a_\pm$ be operators on a Hilbert space satisfying $[a_+,a_-]=0$
For $\cA>0$, 
${\mathcal B} \in {\mathbb R}$ satisfying either $|{\mathcal B}| < {\mathcal A}$ or ${\mathcal B} = {\mathcal A}$ and arbitrary 
$\kappa\in\C$, we have the operator identity 
\begin{align}\label{eq:BogIdentity}
&\cA(a^{*}_+ a_+ +a^{*}_{-} a_{-})+{\mathcal B} (a^*_+a^*_{-}+a_+a_{-})+
\kappa(a^*_+ +a_{-})+\overline{\kappa}(a_+ +a_{-}^{*})\nonumber \\
&=
{\mathcal D} (b_{+}^* b_{+} + b_{-}^* b_{-})
-\frac{1}{2}(\cA-\sqrt{\cA^2-{\mathcal B}^2})
([a_{+},a^*_{+}]+[a_{-},a^*_{-}])-\frac{2|\kappa|^2}{{\mathcal A} + {\mathcal B}},
\end{align}

where
\begin{align}
{\mathcal D} := \frac{1}{2} \left( {\mathcal A} + \sqrt{{\mathcal A}^2 - {\mathcal B}^2}\right),
\end{align}
and
\begin{align}\label{eq:bpm}
b_{+}:= a_{+} + \alpha a_{-}^* + \overline{c_0},\qquad 
b_{-}:=a_{-} + \alpha a_{+}^{*} + c_0,
\end{align}
with
\begin{align}
\alpha:= {\mathcal B}^{-1} \left( {\mathcal A} - \sqrt{{\mathcal A}^2 - {\mathcal B}^2}\right),\qquad
c_0:=\frac{2\overline{\kappa}}{{\mathcal A} + {\mathcal B} + \sqrt{{\mathcal A}^2 - {\mathcal B}^2}}.
\end{align}
In particular,
\begin{align}\label{eq:BogIneq}
		&\cA(a^{*}_+ a_+ +a^{*}_{-} a_{-})+\cB(a^*_+a^*_{-}+a_+a_{-})+
		\kappa(a^*_+ +a_{-}^*)+\overline{\kappa}(a_+ +a_{-})\nonumber \\
		&\geq-\frac{1}{2}(\cA-\sqrt{\cA^2-\cB^2})
		([a_{+},a^*_{+}]+[a_{-},a*_{-}])-\frac{2|\kappa|^2}{\cA+\cB}.
\end{align}
\end{theorem}

\begin{proof}
The identity \eqref{eq:BogIdentity} is elementary. From here the
inequality \eqref{eq:BogIneq} follows by dropping the positive
operator term ${\mathcal D} (b_{+}^* b_{+} + b_{-}^* b_{-})$.
\end{proof}

\section{Localization to small boxes}
\label{SmallBoxes}
\newcommand{\cst}{\text{(cst.)}}
The Hamiltonian ${\mathcal H}_B(\rmu)$ defined in \eqref{eq:Def_HB} (with $u=0$) is localized to the box 
$\Lambda:=\Lambda(0)=[-\ell/2,\ell/2]^3$. In order to arrive at the a priori bounds in Theorem~\ref{thm:aprioribounds} 
we will localize again to boxes with a length scale $\ell d\ll (\rho a)^{-1/2}$. 
The reason for this second localization is that we need a larger Neumann gap in order to absorb errors. 
We therefore introduce a new family of boxes (some of which will have a rectangular shape) given by 
\begin{equation}
B(u)=[-\ell/2,\ell/2]^3\cap \left(\ell d u +[-\ell d/2, \ell d /2]^3\right),\quad u\in \R^3.
\end{equation}
The functions that localize to these boxes are
\begin{equation}\label{eq:chiBdefinition}
        \chi_{B(u)}(x)=\chi\left(\frac{x}{\ell}\right)\chi\left(\frac{x}{d\ell}-u\right), \quad u\in\R^3,
\end{equation}
where $\chi$ is given in \eqref{eq:Def_chi}
in terms of the positive integer $M$. Observe that
\begin{equation}\label{eq:chiBintegral}
\iint\chi_{B(u)}(x)^2dx du= \ell^3.
\end{equation}

As usual we consider the projections 
$$
P_{B(u)}\phi=|B(u)|^{-1}\langle\one_{B(u)},\phi\rangle \one_{B(u)},\quad
Q_{B(u)}\phi=\one_{B(u)}\phi -P_{B(u)}\phi.
$$
In these small boxes we consider the Hamiltonian
\begin{equation}\label{eq:HBU}
  {\mathcal H}_{B(u)}(\rmu)=\sum_{i=1}^N\left({\mathcal T}_{B(u),i}-\rmu \int w_{1,B(u)}(x_i,y)dy\right)
  +\frac12\sum_{i\ne j} w_{B(u)}(x_i,x_j)
\end{equation}
where (omitting the index $u$)
\begin{equation}\label{eq:appTB}
  {\mathcal T}_{B}=\frac12\varepsilon_T(1+\pi^{-2})^{-1}(d\ell)^{-2}Q_B
  +Q_B\chi_B[\sqrt{-\Delta}-(ds\ell)^{-1}]_+^{-2}\chi_{B}Q_B
\end{equation}
and 
\newcommand{\WS}{W^{\rm s}}
\begin{equation}
  w_{B}(x,y)=\chi_B(x)\WS(x-y)\chi_B(y),
  \quad w_{1,B}(x,y)=\chi_B(x)\WS_1(x-y)\chi_B(y)
\end{equation}          
with (where the subscript s refers to small)
\begin{equation}
  \WS(x)=\frac{W(x)}{\chi*\chi(x/(d\ell))},\quad 
  \WS_1(x)=\frac{W_1(x)}{\chi*\chi(x/(d\ell))}.
\end{equation}
As in the large boxes we will also need 
\begin{equation}
\quad w_{2,B}(x,y)=\chi_B(x)\WS_2(x-y)\chi_B(y),\quad 
\WS_2(x)=\frac{W_2(x)}{\chi*\chi(x/(d\ell))}.
\end{equation}
Since $\omega\leq 1$ we have 
\begin{equation}\label{eq:WSestimate}
  \int \WS_2\leq 2\int \WS_1(x)\leq C a 
\end{equation}
We have by a Schwarz inequality that 
\begin{equation}\label{eq:iintw1}
  \iint w_{1,B}(x,y)dxdy\leq \iint\chi_B(x)^2\WS_1(x-y)dxdy\leq \cst a\int\chi_B^2\leq Ca|B|.
\end{equation}
Observe also that 
\begin{equation}\label{eq:appw1sliding}
  \iiint w_{1,B(u)}(x,y)dxdy du= \ell^{3}\int g =8\pi a\ell^{3}.
\end{equation}

It was proved in \cite{BS} Theorem 3.10 that the operator ${\mathcal H}_{\Lambda}(\rmu)$ defined in
\eqref{eq:Def_HB} and \eqref{eq:Def_HB2} can be bounded below by (we are for the lower bound 
ignoring the third term in ${\mathcal T}$ in \eqref{eq:DefT})
\begin{equation}\label{eq:BSsliding2}
  {\mathcal H}_{\Lambda}(\rmu)\geq \sum_{i=1}^N \frac{b}{2} Q_{\Lambda,i}\ell^{-2}+
  \int_{\R^3} {\mathcal H}_{B(u)}(\rmu) du,
\end{equation}
if
\begin{equation}\label{eq:BSlidingassp}
\varepsilon_T,s,ds^{-1},\text{ and } (s^{-2}+d^{-3})(sd)^{-2}s^M
\end{equation}
are smaller than some universal constant. Note that, if $\rmu a^3$ is
small enough, this is satisfied for our choices in
Section~\ref{sec:params}, in particular, due to \eqref{con:d5s}.

In the integral above the operators ${\mathcal H}_{B(u)}(\rmu)$
are, however, not unitarily equivalent. Depending on $u$ the boxes $B(u)$ can
be rather small and rectangular. We denote by
$\lambda_1(u)\leq\lambda_2(u)\leq\lambda_3(u)\leq d\ell$ the side
lengths of the boxes $B(u)$. To avoid boxes that are very small, i.e.,
where $\lambda_1(u)\leq \rmu^{-1/3}$ we will restrict the integral above to
$u$ such that 
$$
\|\ell d u\|_\infty\leq \frac\ell2(1+d)-\rmu^{-1/3}.
$$
Note that since the full integral would be over the set where 
$\|\ell d u\|_\infty\leq \frac\ell2(1+d)$ we see that the 
restriction implies that all boxes will satisfy $\lambda_1(u)\geq\rmu^{-1/3}$.

For the kinetic energy and the repulsive potential
this restriction will only give a further lower bound. For the chemical potential term we will use the following result. 
\begin{lemma}For all $x\in\Lambda$ we have the estimate
\begin{align}
-\rmu\iint &w_{1,B(u)}(x,y)dy du \nonumber \\
\geq &
-\rmu\int_{\|u\|_\infty-\frac12\left(\frac1d+1\right)\leq -(\ell d\rmu^{1/3})^{-1}}\int w_{1,B(u)}(x,y)dy du\nonumber\\
&-3\rmu\int_{-2(\ell d\rmu^{1/3})^{-1}\leq \|u\|_\infty-\frac12\left(\frac1d+1\right)\leq -(\ell d\rmu^{1/3})^{-1}}\int w_{1,B(u)}(x,y)dy du.
\end{align}
\end{lemma}
\begin{proof} The estimate above follows if we can show that for all $x,y\in\Lambda$ we have 
\begin{align}
  \chi*\chi\left(\frac{x-y}{\ell d}\right)\leq& \int_{\|u\|_\infty-\frac12\left(\frac1d+1\right)\leq -(\ell d\rmu^{1/3})^{-1}}
  \chi\left(\frac{x}{\ell d}-u\right)\chi\left(\frac{y}{\ell d}-u\right)du\nonumber\\&
  +3\int_{-2(\ell d\rmu^{1/3})^{-1}\leq \|u\|_\infty-\frac12\left(\frac1d+1\right)\leq -(\ell d\rmu^{1/3})^{-1}}
  \chi\left(\frac{x}{\ell d}-u\right)\chi\left(\frac{y}{\ell d}-u\right)du.
\end{align}
We have 
\begin{align}
\chi*\chi\left(\frac{x-y}{\ell d}\right)-& \int_{\|u\|_\infty-\frac12\left(\frac1d+1\right)\leq -(\ell d\rmu^{1/3})^{-1}}
  \chi\left(\frac{x}{\ell d}-u\right)\chi\left(\frac{y}{\ell d}-u\right)du\nonumber\\=&
  \int_{\|u\|_\infty-\frac12\left(\frac1d+1\right)\geq-(\ell d\rmu^{1/3})^{-1}}
  \chi\left(\frac{x}{\ell d}-u\right)\chi\left(\frac{y}{\ell d}-u\right)du.
\end{align}
Since $x,y\in\Lambda$, the integral on the right is supported on
$\|u\|_\infty-\frac12\left(\frac1d+1\right)\leq 0$. Using the fact that $\rmu^{-1/3}< \ell d/2$ and that $\chi$
is a product of symmetric decreasing functions of the coordinates 
$u_1,u_2,u_3$ respectively, we may observe that for fixed $u_2,u_3$ we have
\begin{align}
  &\max_{\frac12\left(\frac1d+1\right)-(\ell d\rmu^{1/3})^{-1}\leq|u_1|\leq \frac12\left(\frac1d+1\right)}
  \chi\left(\frac{x}{\ell d}-u\right)\chi\left(\frac{y}{\ell d}-u\right)
  \leq\nonumber
\\ &\min_{\frac12\left(\frac1d+1\right)-2(\ell d\rmu^{1/3})^{-1}\leq|u_1|\leq \frac1{2d}+\frac12-(\ell d\rmu^{1/3})^{-1}}
  \chi\left(\frac{x}{\ell d}-u\right)\chi\left(\frac{y}{\ell d}-u\right).
\end{align}
Using this repeatedly (also with $u_1,u_2$ and $u_1,u_3$ fixed) gives the result in the lemma. 
\end{proof}
As a consequence of the lemma we find from \eqref{eq:BSsliding2}, if \eqref{eq:BSlidingassp} is satsifed, that
\begin{equation}\label{eq:appsliding}
  {\mathcal H}_{\Lambda}(\rmu)\geq 
  \frac{b}{2}\ell^{-2}\sum_{i=1}^NQ_{\Lambda,i}+\int_{\|\ell d u\|_\infty\leq \frac12 \ell (1+d)
    -\rmu^{-1/3}} {\mathcal H}_{B(u)}(m(u)\rmu) du,
\end{equation}
where $m(u)=1$ if 
$\|\ell d u\|_\infty\leq \frac12 \ell (1+d)
    -2\rmu^{-1/3}$
and $m(u)=4$ otherwise, i.e., for $u$ near the boundary. 

The goal in the rest of this section is to give a lower bound on the
ground state energy of the operators ${\mathcal H}_{B(u)}(m(u)\rmu)$ for
to conclude an a priori lower bound on the ground state
energy of ${\mathcal H}_{\Lambda}(\rmu)$. 
We may now assume that the shortest side length of $B(u)$ satisfies
$\lambda_1(u)\geq\rmu^{-1/3}$ and we will make use of the fact that the range $R$ of the potential 
satisfies $R\ll\rmu^{-1/3}$.
For simplicity we will often omit the parameter $u$. 
A main ingredient in getting a lower bound is to get a priori bounds on the operators 
\begin{equation}\label{eq:smallboxnn0n+}
n=\sum_{i=1}^N \one_{B,i},\quad n_0=\sum_{i=1}^NP_{B,i},\quad n_+=\sum_{i=1}^NQ_{B,i}.
\end{equation}
Note that the operator $n$ commutes with ${\mathcal H}_{B}$. We will not distinguish the operator $n$ 
from its value and talk about $n$-particle states. 

Applying the decomposition of the potential energy in Subsection~\ref{sec:potsplit} 
to the small boxes we arrive at the following lemma. 
\begin{lemma}\label{lm:appinteractionestimate} There is a constant $C>0$ such that on any small box $B$ we have 
\begin{align}\label{eq:SmallsimpleQs}
  -\rmu \sum_{i=1}^N \int w_{1,B}(x,y)\,dy+
   \frac{1}{2} \sum_{i\neq j}  w_B(x_i, x_j)
  \geq A_0+A_2-Ca (\rmu +n_0|B|^{-1})n_+  
\end{align}
where
\begin{align}
  A_0=&\frac{n_0(n_0-1)}{2|B|^2}\iint w_{2,B}(x,y)\,dx dy \nonumber\\
  &- \left(\rmu \frac{n_0}{|B|}+\frac14\left(\rmu
  -\frac{n_0+1}{|B|}\right)^2\right)\iint w_{1,B}(x,y)\,dx dy
  \label{eq:A0}
\end{align}
and 
\begin{equation}
  A_2= \frac{1}{2}\sum_{i\neq j} P_{B,i} P_{B,j} w_{1,B}(x_i,x_j) Q_{B,j} Q_{B,i} + h.c.
\end{equation} 
\end{lemma}
\begin{proof}
We use the identity \eqref{eq:potsplit} which also holds in the small boxes with 
$P,Q$ and $w,w_{1},w_{2}$ replaced by $P_B,Q_B$ and $w_B,w_{1,B},w_{2,B}$ respectively. 
Let us denote the corresponding terms ${\mathcal Q}_{i,B}^{\rm ren}$, $i=0,\ldots,4$. Then
$$
{\mathcal Q}_{0,B}^{\rm ren}=\frac{n_0(n_0-1)}{2|B|^2}\iint w_{2,B}(x,y)\,dx dy - \rmu \frac{n_0}{|B|}\iint w_{1,B}(x,y)\,dx dy.
$$
As in the proof of Lemma~\ref{lm:potapp} we apply a Cauchy-Schwarz inequality---using the positivity of $w_B$---
to absorb ${\mathcal Q}_{3,B}^{\rm ren}$ in ${\mathcal Q}_{4,B}^{\rm ren}$.
This results in the following inequality,
\begin{align}\label{eq:Q3+Q4CS}
&\,{\mathcal Q}_{3,B}^{\rm ren}+{\mathcal Q}_{4,B}^{\rm ren} \nonumber
  \\ &\,\geq -C\sum_{i\neq j} P_{B,i} Q_{B,j} w_{1,B}(x_i,x_j) Q_{B,j} P_{B,i}
  - \sum_{i\neq j} \Big( P_{B,i} Q_{B,j} w_{1,B} \omega(x_i,x_j) P_{B,j} P_{B,i} +
  h.c.\Big) \nonumber 
  \\ &\,\quad - 2 \sum_{i \neq j} \Big( P_{B,i} Q_{B,j}
  w_{1,B} \omega P_{B,j} Q_{B,i} + h.c. \Big) \nonumber 
  \\ &\,\geq -C\sum_{i\neq
    j} P_{B,i} Q_{B,j} w_{1,B}(x_i,x_j) Q_{B,j} P_{B,i} - \sum_{i\neq j} \Big(
  P_{B,i} Q_{B,j} w_{1,B} \omega(x_i,x_j) P_{B,j} P_{B,i} + h.c.\Big)\nonumber\\
  &\,\geq - \sum_{i\neq j} \Big(
  P_{B,i} Q_{B,j} w_{1,B} \omega(x_i,x_j) P_{B,j} P_{B,i} + h.c.\Big)-Can_0|B|^{-1}n_+ ,
\end{align}
where we have used the pointwise inequality $0 \leq \omega \leq 1$,
and additional Cauchy-Schwartz inequality in the second inequality,
and
\begin{align}\label{eq:BQwQleq}
  \sum_{i\neq j} P_{B,i} Q_{B,j} w_{1,B}(x_i,x_j) Q_{B,j} P_{B,i} \leq
  C \| \chi_{B} \|_{\infty}^2 n_0|B|^{-1} n_{+}\int \WS_1 \leq Can_0|B|^{-1} n_{+},
\end{align}
which follows from 
$$
\int \chi_{B}(x) \WS_1(x-y) \chi_{B}(y)\,dy \leq \|\chi_{B}\|_{\infty}^2 \int \WS_1.
$$
Notice that if we rewrite ${\mathcal Q}_{B,1}^{\rm ren,1}$ as in \eqref{eq:Q1n0}
the first term on the right side of \eqref{eq:Q3+Q4CS} cancels the second line of \eqref{eq:Q1n0}.
The remaining part of ${\mathcal Q}_{B,1}^{\rm ren,1}$  we estimate as follows.
\begin{align}
 |B&|^{-1}(n_0-\rmu|B|)\sum_{i} Q_{B,i} \chi_B(x_i) \WS_1*\chi_B(x_i)
 P_{B,i} +
 h.c.\nonumber\\ =&|B|^{-1}({n_0}^{1/2}+(\rmu|B|)^{1/2})\sum_{i}
 Q_{B,i} \chi_B(x_i) \WS_1*\chi_B(x_i)
 P_{B,i}((n_0+1)^{1/2}-(\rmu|B|)^{1/2}) +
 h.c.\nonumber\\ \geq&-4|B|^{-1}\left(n_0^{1/2}+(\rmu|B|)^{1/2}\right)^2\sum_{i}
 Q_{B,i} \chi_B(x_i) \WS_1*\chi_B(x_i) Q_{B,i}\nonumber\\&
 -\frac14|B|^{-1}\left((n_0+1)^{1/2}-(\rmu|B|)^{1/2}\right)^2\sum_{i}
 P_{B,i} \chi_B(x_i) \WS_1*\chi_B(x_i) P_{B,i}.
\end{align}
The first term above we estimate similarly to the estimate in \eqref{eq:BQwQleq}. 
The last term above is equal to
\begin{align}
&-\frac14\frac{n_0}{|B|}\left((n_0+1)^{1/2}-(\rmu|B|)^{1/2}\right)^2\iint w_{1,B}(x,y)\,dx dy\nonumber\\
&\geq-\frac14\left(\frac{n_0+1}{|B|}-\rmu\right)^2\iint w_{1,B}(x,y)\,dx dy\nonumber
\end{align}
which together with ${\mathcal Q}^{\rm ren}_{0,B}$ give the $A_0$ term in the lemma. 

The first three terms in ${\mathcal Q}_{2,B}^{\rm ren}$ are absorbed into the last term in \eqref{eq:SmallsimpleQs} 
using again the same Cauchy-Schwartz as in the second inequality in \eqref{eq:Q3+Q4CS}.
Finally, the last terms in ${\mathcal Q}_{2,B}^{\rm ren}$ are exactly the terms collected in $A_2$.
\end{proof}

We express the term $A_2$ from the lemma in second quantization. Introducing the operators
$$
b_p^\dagger=|B|^{-1/2}a^\dagger(Q_B\chi_Be^{-ipx})a_0
$$
we can write 
$$
A_2=\frac12 (2\pi)^{-3}\int \widehat{\WS_1}(p)(b_p^\dagger b_{-p}^\dagger+b_{-p}b_p)dp.
$$
We shall control $A_2$ using Bogolubov's method. In order to do this we will add and subtract a term 
\begin{equation}\label{eq:A1}
  A_1= (2\pi)^{-3}K_{\rm s} a \int (b_p^\dagger b_{p}+b^\dagger_{-p}b_{-p})dp,
\end{equation}
with the constant $K_{\rm s}>0$ chosen appropriately. 
Note that we have 
\begin{equation}\label{eq:A1estimate}
  A_1\leq K_{\rm s} a \frac{n_0+1}{|B|} n_+\|\chi_B\|_\infty^2\leq CK_{\rm s} a \frac{n_0+1}{|B|} n_+.
\end{equation}
\begin{lemma}[Bogolubov's method in small boxes]\label{lm:appbogolubov} 
There exists a constant $C>0$ such that 
\begin{align}\label{eq:refinedaprioribog}
  \sum_{i=1}^NQ_{B,i}&\chi_{B,i}[\sqrt{-\Delta_i}-(ds\ell)^{-1}]_+^{-2}\chi_{B,i}Q_{B,i}+A_2
  \geq\nonumber\\&-\frac12(1+C(R/(d\ell))^2)(1+C(ds\ell)^{-1}a)\widehat{g\omega}(0)\frac{(n+1)n}{|B|^2}\int\chi_B^2\nonumber\\
  &-C\left(a^2(ds\ell)^{-1}\ln(ds\ell a^{-1})\frac{n+1}{|B|}
  +a^4(ds\ell)^{3}\left(\frac{n+1}{|B|}\right)^3 +a(ds\ell)^{-3}\right)\frac{n}{|B|}\int\chi_B^2\nonumber\\&
  -Ca\frac{n+1}{|B|}n_+.
\end{align}
Moreover, for all $\varepsilon>0$ there is a $C_\varepsilon>0$ such  that if 
\begin{equation} \label{eq:appbogassumptions}
  (R/d\ell)^2<C_\varepsilon^{-1},\quad a(ds\ell)^{-1}\ln(ds\ell a^{-1})<C_\varepsilon^{-1}
\end{equation}
then
\begin{align}\label{eq:simpleaprioribog}
  \sum_{i=1}^NQ_{B,i}\chi_{B,i}[\sqrt{-\Delta_i}-(ds\ell)^{-1}]_+^{-2}&\chi_{B,i}Q_{B,i}+A_2
  \geq-\frac12\left((1+\varepsilon)\widehat{g\omega}(0)+\varepsilon a\right)\frac{(n+1)n}{|B|^2}\int\chi_B^2\nonumber\\
  &-C_\varepsilon a(ds\ell)^{-3}\frac{n}{|B|}\int\chi_B^2
  -C_\varepsilon a \frac{n+1}{|B|}n_+.
\end{align}
\end{lemma}
\begin{proof}
We add $A_1$ from \eqref{eq:A1} to the term we want to estimate. Using $n_0\leq n$ we may write
$$
\sum_{i=1}^N Q_{B,i}\chi_{B,i}[\sqrt{-\Delta_i}-(ds\ell)^{-1}]_+^{-2}\chi_{B,i}Q_{B,i}+A_1+A_2\geq
(2\pi)^{-3}\frac12\int h(p) dp
$$
where $h$ is the operator
$$
h(p)=\left(\frac{|B|}{n+1}[|p|-(ds\ell)^{-1}]^2_++2K_{\rm s} a \right)(b_p^\dagger b_{p}+b^\dagger_{-p}b_{-p})+
\widehat{\WS_1}(p)(b_p^\dagger b_{-p}^\dagger+b_{-p}b_p).
$$
We observe that 
$$
[b_p,b_p^\dagger]\leq n_0|B|^{-1}\int \chi_B^2\leq n|B|^{-1}\int \chi_B^2.
$$
We will now apply the simple case of Bogolubov's method in Theorem~\ref{thm:bogolubov-complete} with 
$$
{\mathcal A}(p)= \frac{|B|}{n+1}[|p|-(ds\ell)^{-1}]^2_++2K_{\rm s} a,\quad
{\mathcal B}(p)= \widehat{\WS_1}(p)
$$
We have by \eqref{eq:WSestimate} that 
$$
|{\mathcal B}(p)|= |\widehat{\WS_1}(p)|\leq \int\WS_1\leq C_0a .
$$ 
If we therefore choose $K_{\rm s}\geq C_0$ we see that $|{\mathcal B}|/{\mathcal A}\leq 1/2$ and 
we get the following lower bound from Theorem~\ref{thm:bogolubov-complete}. 
$$
h(p)\geq -\frac12\left({\mathcal A}(p)-\sqrt{{\mathcal A}(p)^2-{\mathcal B}(p)^2}\right)n_0|B|^{-1}\int \chi_B^2.
$$ 
Using that $|{\mathcal
  B}|/{\mathcal A}\leq 1/2 $ we have 
$$
h(p)\geq -C\frac{{\mathcal B}(p)^2}{{\mathcal A}(p)}n|B|^{-1}\int \chi_B^2.
$$
We use this for $|p|<2(ds\ell)^{-1}$ and we find
for the integral over $|p|<2(ds\ell)^{-1}$ 
\begin{align}
  \int_{|p|<2(ds\ell)^{-1}} \frac{{\mathcal B}(p)^2}{{\mathcal A}(p)}dp \leq& \frac{C_0^2}{K_{\rm s}}a\int_{|p|<2(ds\ell)^{-1}}1dp\leq
    \frac{C_0^2}{K_{\rm s}} (ds\ell)^{-3}.
\end{align}
For the simple bound \eqref{eq:simpleaprioribog} we may choose $K_s$ large depending on $\varepsilon$ to have 
$$
h(p)\geq -\frac12(1+\varepsilon/2)\frac{{\mathcal B}(p)^2}{{\mathcal A}(p)}n|B|^{-1}\int \chi_B^2
$$
and use this in the range $|p|>2(ds\ell)^{-1}$.
For the more refined bound \eqref{eq:refinedaprioribog}, in the range $|p|>2(ds\ell)^{-1}$, we use 
$$
h(p)\geq -\left(\frac12\frac{{\mathcal B}(p)^2}{{\mathcal A}(p)}+C\frac{{\mathcal B}(p)^4}{{\mathcal A}(p)^3}\right)
n|B|^{-1}\int \chi_B^2.
$$
For $|p|>2(ds\ell)^{-1}$ we have
$$
  \frac{{\mathcal B}(p)^2}{{\mathcal A}(p)}\leq\frac{n+1}{|B|}
  \frac{\widehat{\WS_1}(p)^2}{(|p|-(ds\ell)^{-1})^2}
  \leq \frac{\widehat{\WS_1}(p)^2}{p^2}(1+C(ds\ell)^{-1}|p|^{-1})\frac{n+1}{|B|}
$$
and hence by splitting the integral over the error in $|p|<a^{-1}$ and $|p|>a^{-1}$ we obtain
$$
\int_{|p|>2(ds\ell)^{-1}}\frac{{\mathcal B}(p)^2}{{\mathcal A}(p)}dp \leq(1+C(ds\ell)^{-1}a)\frac{n+1}{|B|}
\int_{\R^3} \frac{\widehat{\WS_1}(p)^2}{p^2}dp +Ca^2(ds\ell)^{-1}\frac{n+1}{|B|}\ln(ds\ell a^{-1}). 
$$
Finally, we use that 
\begin{align}
  \frac14(2\pi)^{-3}\int_{\R^3} \frac{\widehat{\WS_1}(p)^2}{p^2}dp 
  =&\frac14\iint\frac {\WS_1(x)\WS_1(y)}{4\pi|x-y|}dxdy\nonumber\\
  \leq&\frac14(1+C(R/(d\ell))^2)\iint\frac {g(x)g(y)}{4\pi|x-y|}dxdy\nonumber\\
  =&\frac14(1+C(R/(d\ell))^2)(2\pi)^{-3}\int_{\R^3} \frac{\widehat{g}(p)^2}{p^2}dp
  \nonumber\\
  =&\frac12(1+C(R/(d\ell))^2)\widehat{g\omega}(0).
\end{align}
Finally, to get \eqref{eq:refinedaprioribog} we estimate
\begin{align}
  \int_{|p|>2(ds\ell)^{-1}}\frac{{\mathcal B}(p)^4}{{\mathcal A}(p)^3}
  \leq \cst a^4\left(\frac{n+1}{|B|}\right)^3\int_{|p|>2(ds\ell)^{-1}}  |p|^{-6} dp
  =Ca^4\left(\frac{n+1}{|B|}\right)^3(ds\ell)^3.
\end{align}
Using the estimate \eqref{eq:A1estimate} on $A_1$ gives the
last term in \eqref{eq:refinedaprioribog}. 
\end{proof}
 
In order to use this lemma we will control the negative term quadratic
in $n$ in \eqref{eq:simpleaprioribog} in terms of the positive term
quadratic in $n_0$ in \eqref{eq:A0}. The difference between $n$ and $n_0$ will be absorbed in the Neumann gap 
of ${\mathcal T}_B$. It is, however, important to establish the result in the following lemma
\begin{lemma}\label{lm:w2>w1}
There is a constant $C>0$ such that if the shortest side length $\lambda_1$ of the box 
$B$ satisfies $R\leq \frac12C^{-1/2}\min\{\lambda_1,\ell d\}$ then
\begin{equation}\label{eq:w1>}
  \iint w_{1,B}(x,y)dx dy \geq 8\pi a\left(1-C
  \Big(\frac{R}{\lambda_1}\Big)^2\right)\left(1-C \Big(\frac{R}{\ell
    d}\Big)^2\right)\int\chi_B^2.
\end{equation}
\begin{equation}\label{eq:w>g}
  \iint w_{2,B}(x,y)dx dy\geq \iint w_{1,B}(x,y)dx dy +\left(1-C
  \Big(\frac{R}{\lambda_1}\Big)^2\right)\left(1-C \Big(\frac{R}{\ell
    d}\Big)^2\right)\widehat{g\omega}(0)\int\chi_B^2.
\end{equation}
Moreover for any $0<\varepsilon<1/10$ we can find a $C_\varepsilon>0$ such that if 
$R\leq C_\varepsilon^{-1}\min\{\lambda_1,\ell d\}$ then
\begin{equation}\label{eq:wepsilong}
  \iint w_{2,B}(x,y)dx dy\geq \frac34\iint w_{1,B}(x,y)dx dy 
  +\left((1+\varepsilon)\widehat{g\omega}(0)+\varepsilon a\right)\int\chi_B^2.
\end{equation}
\end{lemma}
\begin{proof}
The estimate \eqref{eq:w>g} follows from
\begin{align}
  \iint w_{2,B}(x,y)dx dy&- \iint w_{1,B}(x,y)dx dy \nonumber \\
  & =\iint \omega(x-y)w_{1,B}(x,y)\nonumber\\ 
  &\geq \int \omega(x)W_1^{\rm s}(x)dx \left(\int\chi_B^2-CR^2\|\nabla^2\chi_B\|_\infty\int\chi_B\right)\nonumber\\
  & \geq \left(1-C (R\lambda_1^{-1})^2\right)\int \omega(x) W_1^{\rm s}(x)dx\int\chi_B^2\nonumber\\
  & \geq\left(1-C (R\lambda_1^{-1})^2\right)\left(1-C (R/(\ell d))^2\right)\left(\int g\omega\right)\int\chi_B^2,\nonumber
\end{align}
where we have used that $\omega W$ is spherically symmetric, that
$|B|^{-1}\left(\int\chi_B\right)^2\leq \int\chi_B^2$, and that 
\begin{equation}\label{eq:chiBestimate}
  \|\partial_i\partial_j\chi_B\|_\infty\leq C_M \lambda_1^{-2}|B|^{-1}\int\chi_B,
\end{equation}
which is a simple exercise (see Appendix~\ref{sec:chiproperties}).
The estimate \eqref{eq:w1>} follows in the same way without $\omega$ and using $\int g=8\pi a$. 
Finally, \eqref{eq:wepsilong} follows from $\omega\leq1$. 
\end{proof}

We are now ready to give the bound on the energy in the small boxes. 
\begin{theorem}[Lower bound on energy in small boxes]\label{thm:smallbox}
Assume $B$ is a box with shortest side length
$\lambda_1\geq\rmu^{-1/3}$.  There are universal constants $C, C'>1$
and $0<c<1/2$ such that for all $1\leq K_B\leq C'^{-1}(\rmu
a^3)^{-1/6}$ we have for the Hamiltonian defined in \eqref{eq:HBU}
restricted to $n$-particle states that
\begin{align}
  {\mathcal H}_B(\rmu)\geq& 
  \left(c\frac{(n|B|^{-1}-\rmu)^2}{1+\frac{n}{|B|\rmu}}-\frac12\rmu^2\right)\iint w_{1,B}(x,y)dxdy\nonumber\\
  &-C\rmu^2 a\left((R\lambda_1^{-1})^2+K_B^3({\rmu a^3})^{1/2}\right)\int\chi_B^2-C\rmu a,
\end{align}
if 
\begin{equation}\label{eq:appdellsassumption}
  C'(\rmu a^3)^{1/2}\leq \varepsilon_T^{1/2}  a(d\ell)^{-1}\leq a (ds\ell)^{-1}\ln(ds\ell a^{-1})\leq K_B(\rmu a^3)^{1/2},
\end{equation}
and
\begin{equation}\label{eq:appRassumption}
  R\leq K_B^{1/2}(\rmu a^3)^{1/4}(\rmu a)^{-1/2}.
\end{equation}
We are assuming that $\varepsilon_T, s, d\leq 1$. 
\end{theorem}
Note that all the assumptions on $K_B$, $R$, $\varepsilon_T$, $s$, and
$d$ are satisfied with our choices in Section~\ref{sec:params}, if
$\rmu a^3$ is small enough. In particular, the assumption on $K_B$ is
a consequence of \eqref{con:KB}, \eqref{eq:appdellsassumption} follows
from \eqref{con:eTdK} and \eqref{con:sdKellKB}
and \eqref{eq:appRassumption} was given in \eqref{eq:AssumptionK_R}.
\begin{proof}
Note that \eqref{eq:appdellsassumption} implies that 
$$
 (\rmu a)^{-1/2}K_B^{-1}\leq \frac{sd\ell}{\ln(ds\ell a^{-1})}\leq d\ell\leq \frac{\varepsilon_T^{1/2}}{C'} (\rmu a)^{-1/2}.
$$ 
This, in particular, implies that
\begin{equation}\label{eq:appgapestimate}
  \sum_{i=1}^N\frac12\varepsilon_T(1+\pi^{-2})^{-1}(\ell d)^{-2}Q_{B,i}\geq C'^2(1+\pi^{-2})^{-1}\rmu a n_+.
\end{equation}
Moreover we see from \eqref{eq:appRassumption} that
$$ 
R\rmu^{1/3}\leq K_B^{1/2}(\rmu a^3)^{1/12},\quad R/(\ell
d)\leq K_B^{3/2}(\rmu a^3)^{1/4}.
$$ 

We now first choose $\varepsilon$ so small, e.g., to be $1/20$ so that
we can apply Lemma~\ref{lm:w2>w1}. Hence if $C'$ is large enough, we can, 
since $\lambda_1> \rmu^{-1/3}$, use both \eqref{eq:w>g} and
\eqref{eq:wepsilong} from Lemma~\ref{lm:w2>w1}. We choose the same
$\varepsilon$ in \eqref{eq:simpleaprioribog} and again by assuming
that $C'$ large enough we can ensure that
\eqref{eq:appbogassumptions} is satisfied.

We may of course assume that $n>0$ as the inequality we want to prove is
obviously satisfied if $n=0$ as the operator is $0$ whereas the lower bound is
negative in this case.
We choose a constant $C_n>2$ to be determined precisely below
to depend only on $C$ and $C_\varepsilon$ in
Lemma~\ref{lm:appbogolubov}. Observe that $\rmu|B|\geq1$. Hence we can
choose an integer $n'\geq C_n$ in the interval $[C_n\rmu|B|,(C_n+1)\rmu|B|)$ and
we may write $n=mn'+n''$ with $m, n',n''$ non-negative integers and 
$n''<n'<(C_n+1)\rmu|B|$. We will get a lower bound
on the energy if in the Hamiltonian we think of dividing the
particles in $m$ groups of $n'$ particles and one group of $n''$
particles ignoring the positive interaction between the groups. It
is not important that the Hamiltonian is no longer symmetric between
the particles as we are not considering it as an operator, but only
calculating its expectation value in a symmetric state. We
arrive at the conclusion that if we denote by $e_B(n,\rmu)$ the
ground state energy of ${\mathcal H}_B(\rmu)$ in an $n$-particle
state then
\begin{equation}
  e_B(n,\rmu)\geq me_B(n',\rmu)+e_B(n'',\rmu).
\end{equation}
We have that both $n'$ and $n''$ are less than $(C_n+1)\rmu|B|\leq 2C_n\rmu|B|$.  This means
that the last terms in \eqref{eq:SmallsimpleQs}, \eqref{eq:refinedaprioribog}, 
and \eqref{eq:simpleaprioribog} in both cases can be absorbed in the positive term from
\eqref{eq:appgapestimate} if we choose $C'\geq CC_n$.  Using
\eqref{eq:iintw1} we see that the same is also true
for the errors we get by replacing $n'_0$ and $n_0''$ by $n'$ and $n''$ respectively everywhere in $A_0$ in
\eqref{eq:A0}.  

In the case of the $m$ groups of $n'$ particles we will use Lemma
\ref{lm:appinteractionestimate} and \eqref{eq:simpleaprioribog} to arrive at 
\begin{align}
e_B(n', \rmu)\geq& \frac{{n'}^2}{2|B|^2}\iint w_{2,B}(x,y)\,dx dy 
  - \left(\rmu \frac{{n'}}{|B|}+\frac14\left(\rmu
  -\frac{{n'}}{|B|}\right)^2\right)\iint w_{1,B}(x,y)\,dx dy\nonumber \\
  &-\frac12\left((1+\varepsilon)\widehat{g\omega}(0)+\varepsilon a\right)\frac{{n'}^2}{|B|^2}\int\chi_B^2
  -C\rmu a {n'}|B|^{-1} \int\chi_B^2.\nonumber
\end{align}
where we have used that \eqref{eq:appdellsassumption} implies that $a(ds\ell)^{-3}\leq C\rmu a$.
We have also used that the error in replacing $n'+1$ by $n'$ in several terms can also be absorbed in the last term. 
Thus applying \eqref{eq:wepsilong} we arrive at 
\begin{align}
  e_B(n', \rmu)\geq \frac18 \left(\rmu
  -\frac{n'}{|B|}\right)^2\iint w_{1,B}(x,y)\,dx dy 
 -C\rmu a {n'}|B|^{-1} \int\chi_B^2.\nonumber
\end{align}
It follows, using \eqref{eq:w1>} that if we choose the constant $C_n$ large enough then 
\begin{align}
  e_B(n', \rmu)\geq \frac19 \left(\rmu
  -\frac{n'}{|B|}\right)^2\iint w_{1,B}(x,y)\,dx dy \geq 
  \frac1{18} \rmu\frac{n'}{|B|}\iint w_{1,B}(x,y)\,dx dy \geq 0\nonumber.
\end{align}
Hence 
\begin{equation}
  me_B(n', \rmu)\geq \frac1{18} \rmu\frac{m n'}{|B|}\iint w_{1,B}(x,y)\,dx dy 
  = \frac1{18} \rmu\frac{n-n''}{|B|}\iint w_{1,B}(x,y)\,dx dy 
\end{equation}

We turn to the group of $n''$ particles.  If we apply Lemma~\ref{lm:appinteractionestimate} 
and \eqref{eq:refinedaprioribog} we see that
since $n''\leq 2C_n\rmu|B|$ we have 
\begin{align}
e_B(n'', \rmu)\geq& \frac{n''^2}{2|B|^2}\iint w_{2,B}(x,y)\,dx dy 
  - \left(\rmu \frac{n''}{|B|}+\frac14\left(\rmu
  -\frac{n''}{|B|}\right)^2\right)\iint w_{1,B}(x,y)\,dx dy\nonumber \\
  &-\frac12\widehat{g\omega}(0)\frac{n''^2}{|B|^2}\int\chi_B^2
  -CC_n^4\rmu^2 aK_B^3({\rmu a^3})^{1/2}\int\chi_B^2-CC_n\rmu a.\nonumber
\end{align}
The last term comes from repeatedly replacing $n''+1$ by $n''$ in the
leading terms, which leads to an error $n''a |B|^{-2}\int\chi_B^2\leq
C n''|B|^{-1} a$. In the error terms we can for the same replacement
alternatively use that $1 \leq \rmu|B|$.

If we now apply the estimate \eqref{eq:w>g} in Lemma~\ref{lm:w2>w1} we find that
\begin{align}
e_B(n'', \rmu)\geq& \frac14 \left(\frac{n''}{|B|}-\rmu\right)^2 \iint w_{1,B}(x,y)\,dx dy 
  - \frac12 \rmu^2\iint w_{1,B}(x,y)\,dx dy \nonumber \\
  &-C\rmu^2 a\left((R\lambda_1^{-1})^2+K_B^3({\rmu a^3})^{1/2}\right)\int\chi_B^2-C\rmu a.
\end{align}
where we have now ignored the explicit dependence on $C_n$, which is after all now a chosen constant. 
 
We have arrived at the bound that
\begin{align}
  e_B(n,\rmu)\geq &
  \left(\frac14 \left(\frac{n''}{|B|}-\rmu\right)^2 +\frac1{18} \rmu\frac{n-n''}{|B|}\right)\iint w_{1,B}(x,y)\,dx dy 
  \nonumber\\
  &- \frac12 \rmu^2\iint w_{1,B}(x,y)\,dx dy \nonumber 
  -C\rmu a^2\left((R\lambda_1^{-1})^2+K_B^3({\rmu a^3})^{1/2}\right)\int\chi_B^2-C\rmu a.
\end{align}
This easily implies the result in the theorem. 
\end{proof}

We will now apply the small box estimate from the previous theorem to
get an a priori bound on the energy and on the number of particles $n$
and excited particles $n_+$ in the large box.

\begin{theorem}[A priori estimates in large box] \label{thm:aprioriHLambda}
Assume \eqref{eq:def_ell}, \eqref{eq:appdellsassumption}, \eqref{eq:appRassumption}. Then 
there is a constant $C>0$ such that if again $1\leq K_B\leq C'^{-1}(\rmu a^3)^{-1/6}$ and 
$\rmu a^3$ is smaller than some universal constant we have 
\begin{equation}\label{eq:aprioriHLambda}
{\mathcal H}_\Lambda(\rmu)\geq -4\pi \rmu^2 a\ell^3(1+CK_B^3(\rmu a^3)^{1/2}). 
\end{equation}
Moreover, if there exists a normalized $n$-particle 
$\Psi\in{\mathcal F}_{\rm s}(L^2(\Lambda))$ such that \eqref{eq:aprioriPsi} holds
for a $0<J\leq K_B^3$ then the a priori bounds \eqref{eq:apriorinn+} hold. 
\end{theorem}
As explained just after Theorem~\ref{thm:smallbox} the
assumptions \eqref{eq:appdellsassumption}, \eqref{eq:appRassumption}, and
the assumption on $K_B$ are satisfied with our choices in
Section~\ref{sec:params}.
\begin{proof}
We use \eqref{eq:appsliding} together with the estimate in
Theorem~\ref{thm:smallbox}.  We will denote by $n(u), n_0(u)$, and
$n_+(u)$ the operators defined in \eqref{eq:smallboxnn0n+}. The
corresponding operators in the large box $\Lambda$ will be denoted $n,n_0$, and $n_+$.
On the set 
$$
{\mathcal I}=\Big\{u\in \Big[-\frac12(1+\frac1d),\frac12(1+\frac1d)\Big]^3\ \Big|\ 
\frac12 \ell(1+d)-2\rmu^{-1/3}\leq \|\ell du\|_\infty\leq \frac12 \ell(1+d)-\rmu^{-1/3}\Big\}
$$ 
we have that 
$\rmu$ is replaced by $4\rmu$. On this set we have according to \eqref{eq:chiBsupspecial} that 
$|\chi_{B(u)}(x)|\leq C(\rmu^{-1/3}/\ell)^M\leq C(\rmu a^3)^{M/6}$ with ($C$ depending on $M$) and therefore 
\begin{equation}\label{eq:Iintegral}
  \int_{\mathcal I} 
  \int\chi_{B(u)}(x)^2dx\, du\leq C(\rmu a^3)^{M/3}(\ell d)^3d^{-3}\leq C(\rmu a^3)^{M/3}\ell^3
\end{equation}
If we use Theorem~\ref{thm:smallbox} and \eqref{eq:iintw1} 
to get the
the rough estimate
$$
{\mathcal H}_B(4\rmu)\geq -C\rmu^2 a\int\chi_B^2-C\rmu a
$$
we obtain
\begin{equation}\label{eq:badset}
  \int_{\mathcal I} 
      {\mathcal H}_B(4\rmu)\geq -C\rmu^2 a(\rmu a^3)^{M/3}\ell^3 -C\rmu ad^{-3}.
\end{equation}
In order to apply the estimate in Theorem~\ref{thm:smallbox} over the remaining $u$ we need to control
\begin{align}
  \int (R\lambda_1(u)^{-1})^2\int\chi_{B(u)}^2(x)dxdu\leq& C R^2(d\ell)^{-2} \int (\lambda_1(u)/(d\ell))^{M-2}
  \int\chi_{B(u)}(x)dxdu\nonumber\\
  \leq& 
  C(R/(\ell d)^2) \ell^{3}\leq CK_B^3\ell^3,
\end{align}
where we have used \eqref{eq:chiBsupgeneral}, i.e., $\|\chi_B\|_\infty\leq C(\lambda_1/(d\ell))^M$ 
and $\iint\chi_{B(u)}(x)dx du=C\ell^3$.
If we combine this with \eqref{eq:badset} (with $M=8$), \eqref{eq:appsliding}, \eqref{eq:appw1sliding}, \eqref{eq:chiBintegral},
and the estimate in
Theorem~\ref{thm:smallbox} we arrive at the final a priori lower bound
\begin{align}
  \langle\Psi,{\mathcal H}_\Lambda(\rmu)\Psi\rangle\geq& 
  {\mathcal R}_\Psi+
  \langle\Psi,\frac{b}{2\ell^2}n_+\Psi\rangle 
    -4\pi\rmu^2 a\ell^3-CK_B^3(\rmu a^3)^{1/2}\ell^3-C\rmu a d^{-3}\nonumber\\
  \geq& {\mathcal R}_\Psi
  +\langle\Psi,\frac{b}{2\ell^2}n_+\Psi\rangle \nonumber-4\pi\rmu^2 a\ell^3\left(1+CK_B^3(\rmu a^3)^{1/2}\right),
\end{align}
where 
$$
0\leq {\mathcal R}_\Psi=\left\langle\Psi,\left(\int_{{\mathcal I}_-}  
F\left(\frac{n(u)}{|B(u)|}\right)\iint w_{1,B(u)}(x,y)dxdy du\right)\Psi\right\rangle
$$
with $F(t)=c\frac{(t-\rmu)^2}{1+t\rmu^{-1}}$ and 
$$
{\mathcal I}_-=\{u\in\R^3\  |\  \|\ell du\|_\infty\leq \frac12 \ell(1+d)-2\rmu^{-1/3}\}. 
$$
Since${\mathcal R}_\Psi$ and $n_+$ are non-negative this immediately gives \eqref{eq:aprioriHLambda} and 
\begin{equation}\label{eq:aprioriRn+}
  R_\Psi\leq C\rmu^2 a\ell^3K_B^3(\rmu
  a^3)^{1/2},\quad\text{and}\quad \langle\Psi,
  n_+\Psi\rangle\leq C\rmu\ell^3K_B^3K_\ell^2(\rmu
  a^3)^{1/2}
\end{equation}
for a normalized $n$-particle $\Psi$ satisfying \eqref{eq:aprioriPsi}.
It remains to establish the a priori bound on $n$ in \eqref{eq:apriorinn+}.

Using that the function $F$ is convex and denoting 
$$
{\mathcal C}=\int_{{\mathcal I}_-}\iint w_{1,B(u)}(x,y)dxdy du
$$
we obtain
\begin{equation}\label{eq:Rconvex}
  R_\Psi\geq {\mathcal C} F\left({\mathcal C} ^{-1}
  \left\langle\Psi, \left(\int_{{\mathcal I}_-}\frac{n(u)}{|B(u)|}\iint w_{1,B(u)}(x,y)dxdy du\right)
  \Psi\right\rangle\right)
\end{equation}

We have by \eqref{eq:appw1sliding} that
$$
 8\pi a\ell^3(1-C(\rmu a^3)^{M/3})\leq{\mathcal C}\leq 8\pi a\ell^3,
$$
where we used \eqref{eq:iintw1} and as in \eqref{eq:Iintegral} that 
$|\chi_{B(u)}(x)|\leq C(\rmu^{-1/3}/\ell)^M\leq C(\rmu a^3)^{M/6}$ for $u$ outside ${\mathcal I}_-$.

We may write  
$$
{\mathcal C}^{-1}\int_{{\mathcal I}_-}\frac{n(u)}{|B(u)|}\iint w_{1,B(u)}(x,y)dxdy du=\sum_{i=1}^N U(x_i)
$$
where
$$
U(z)={\mathcal C}^{-1}\int_{{\mathcal I}_-}|B(u)|^{-1}\one_{B(u)}(z)\iint w_{1,B(u)}(x,y)dxdy du.
$$
Using the form of $F$ and the a priori bound on ${\mathcal R}_\Psi$ in \eqref{eq:aprioriRn+} we see 
that 
\begin{equation}\label{eq:aprioriU}
\left|\langle\Psi,\sum_i U(x_i)\Psi\rangle-\rmu\right|\leq C\rmu K_B^{3/2}(\rmu a^3)^{1/4}
\end{equation}
Note that by \eqref{eq:iintw1} and $\int\one_{B(u)}du=\one_\Lambda $ we have that $U(z)\leq C\ell^{-3}$ and that 
$$
P_\Lambda UP_\Lambda= P_\Lambda |\Lambda|^{-1}\int_\Lambda U(z)dz= P_\Lambda \ell^{-3}.
$$
Using that for all $\varepsilon>0$  
\begin{align}
  (1-\varepsilon)\sum_{i=1}^N (P_\Lambda UP_\Lambda)_i-\varepsilon^{-1}
  \sum_{i=1}^N& (Q_\Lambda UQ_\Lambda)_i\leq \sum_{i=1}^N U(x_i)\leq\nonumber\\ 
  &(1+\varepsilon)\sum_{i=1}^N (P_\Lambda UP_\Lambda)_i+(1+\varepsilon^{-1})
  \sum_{i=1}^N (Q_\Lambda UQ_\Lambda)_i\nonumber
\end{align}
we see that 
$$
(1-\varepsilon)n_0\ell^{-3}-C\varepsilon^{-1}n_+\ell^{-3}\leq \sum_{i=1}^N U(x_i)
\leq (1+\varepsilon)n_0\ell^{-3}+(1+\varepsilon^{-1})Cn_+\ell^{-3}.
$$
Choosing $\varepsilon=K_B^{3/2}K_\ell(\rmu a^3)^{1/4}$ and using 
the a priori bounds on the expectation values of $n_+$ in \eqref{eq:aprioriRn+} and $U$
in \eqref{eq:aprioriU} we conclude the result in the theorem. 
\end{proof}

\section{The explicit localization function}\label{sec:chiproperties}

In this section we discuss the explicit choice of the localization function 
$\chi$ and its properties.
Define
$$
\zeta(y) = \begin{cases}
\cos(\pi y), & |y| \leq 1/2,\\ 0,& |y|>1/2,
\end{cases}
$$
and
\begin{align}
\label{eq:Def_chi}
\chi(x) = C_{\cM} \left(\zeta(x_1) \zeta(x_2) \zeta(x_3)\right)^{\cM}.
\end{align}
Here $\cM \in {\mathbb N}$ is to be chosen large enough and we explained the need to choice $M=\Mvalue$ in 
Section~\ref{sec:params}. The constant $C_{\cM}$ is chosen 
such that the normalization $\int\chi^2=1$ from \eqref{eq:chinormalization} holds.
We have $0\leq\chi \in C^{\cM-1}(\R^3)$.

\begin{lemma}\label{lem:DecaychiHat}
Let $\chi$ be the localization function from \eqref{eq:Def_chi}.
Let $\widetilde{M} = \max\{ n \in {\mathbb Z} | 2n \leq M\}$.
Then, for all $k \in {\mathbb R}^3$,
\begin{align}
|\widehat{\chi}(k)| \leq C_{\chi} (1+ |k|^2)^{-\widetilde{M} },
\end{align}
where
\begin{align}
C_{\chi} = \int \left| (1- \Delta)^{\widetilde{M} } \chi 
\right|
\end{align}
In particular, when $|k| \geq \frac{1}{2} \KH^{-1} (\rmu a^3)^{\frac{5}{12}} a^{-1}$, with the notation from \eqref{eq:KHtilde-cutoff}, we have
\begin{align}\label{eq:DecayPH}
|\widehat{\chi_{\Lambda}}(k )| =
\ell^3 |\widehat{\chi}(k \ell)|\leq C \ell^3 ( K_{\ell}^{-2} \KH^2 (\rmu a^3)^{\frac{1}{6}} )^{\widetilde{M}}.
\end{align}
\end{lemma}

The proof of Lemma~\ref{lem:DecaychiHat} is elementary and will be omitted.

The explicit choice of $\chi$ is important when we analyze the behavior of the small box localization function. 
Recall that according to \eqref{eq:chiBdefinition} and the explicit choice of $\chi$ we may write
$
\chi_B(x)=C_M^2F(x)^M
$
where $F(x)=h_{u_1}(x_1)h_{u_2}(x_2)h_{u_3}(x_3)$ and 
$$
h_v(t)=\zeta\left(\frac{t}{\ell}\right)\zeta\left(\frac{t}{\ell d}-v\right).
$$
If we denote by $\lambda_1$ the shortest side length in the box $B$ we see by estimating one of the 
$\zeta$ factors of scale $d\ell$ and using that it must vanish at one of the sides that    
\begin{equation}\label{eq:chiBsupgeneral}
        \chi_B(x)\leq C C_M^2 (\lambda_1/(d\ell))^M.
\end{equation}
If the shortest side length $\lambda_1$ of the box $B$ satisfies that $\lambda_1<d\ell$ we can improve this slightly to 
\begin{equation}\label{eq:chiBsupspecial}
        \chi_B(x)\leq C C_M^2 (\lambda_1/\ell)^M.
\end{equation}
This follows by estimating a $\zeta$ factor of scale $\ell$ and using that {\it it} vanishes at one of the sides.

In this rest of this short appendix 
we will briefly sketch how to get the estimate \eqref{eq:chiBestimate} on $\chi_B$. 
Recall that according to \eqref{eq:chiBdefinition} and the explicit choice of $\chi$ we may write
$
\chi_B(x)=C_M^2F(x)^M
$
where $F(x)=h_{u_1}(x_1)h_{u_2}(x_2)h_{u_3}(x_3)$ and 
$$
h_v(t)=\zeta\left(\frac{t}{\ell}\right)\zeta\left(\frac{t}{\ell d}-v\right)
$$

Our first claim is that 
$$
\|\chi_B\|_\infty\leq C'_M|B|^{-1}\int\chi_B
$$
It is enough to show this for the function $h_v(t)^M$.  
Since $\zeta$ is concave on its support we have that if $h_v$ is supported on $[a,b]$ and takes its maximum in $c$
then
$$
h_v(t)\geq \|h_v\|_\infty \min\left\{\frac{(t-a)^2}{(c-a)^2},\frac{(t-b)^2}{(c-b)^2}\right\}.
$$ 
In particular, $h_v$ is bigger than $\frac14\|h_v\|_\infty$ on half the interval. The claim follows from this. 

Our second claim is that 
$$
\max_i\|\partial_i\chi_B\|_\infty\leq C'_M\lambda_1^{-1}\|\chi_B\|_\infty,
\quad \max_{i,j}\|\partial_i\partial_j\chi_B\|_\infty\leq C'_M\lambda_1^{-2}\|\chi_B\|_\infty.
$$
It is easy to see that it is enough to show these properties
for $h_v$, i.e., that 
$$
\|h'_v\|_\infty\leq C'(b-a)^{-1}\|h_v\|_\infty,\qquad \|h''_v\|_\infty\leq C'(b-a)^{-2}\|h_v\|_\infty.
$$ 
In the case when $(b-a)<\ell d$ we have that one factor in $h_v$
vanishes at one end point and the other factor vanishes at the other
endpoint. It is then easy to see that $\|h_v'\|_\infty\leq C(b-a)/(d\ell^2)$, 
$\|h_v''\|_\infty\leq C(\ell^{-2}+(\ell d)^{-2})\|h_v\|_\infty+C(\ell^2 d)^{-1}$, and 
$\|h_v\|_\infty\geq c(b-a)^2(d\ell^2)^{-1}$. In case $b-a=\ell d$. Both endpoints occur when the second $\zeta$ factor
in $h_v$ vanish. Without loss of generality we may consider $v>0$ and 
let $D=|\ell(1/2-dv)|$ denote the distance from the middle of the support of 
$h_v$, i.e., $ldv$ to the right endpoint of the support of the first $\zeta$ factor, i.e., $\ell/2$. Then 
$\ell d/2\leq D\leq \ell/2$ and 
$$
\|h_v'\|\leq C'(\ell^2d)^{-1}D,\quad \|h_v''\|\leq C'(\ell^{-2}+(\ell d)^{-2})\|h_v\|_\infty+C'(\ell^2 d)^{-1},\quad
\|h_v\|_\infty\geq cD/\ell. 
$$

\bibliographystyle{siam}
\bibliography{LHY-O}

%****************************************
%\bibliographystyle{plain}
%\bibliography{biblio}
\end{document}